\documentclass[citation]{documentclass}

\usepackage{hyperref}

\usepackage{multirow}	
\usepackage{cancel}
\usepackage{bigints}
\usepackage[export]{adjustbox} 

\jyear{2024}%

\theoremstyle{thmstyleone}%
\newtheorem{theorem}{Theorem}

\theoremstyle{thmstyletwo}%

\theoremstyle{thmstylethree}%

\newcommand\rc{\mathrm{c}}
\newcommand\rd{\mathrm{d}}

\newcommand\req{\mathrm{eq}}
\newcommand\rcap{\mathrm{cap}}
\newcommand\rex{\mathrm{ex}}
\newcommand\rmax{\mathrm{max}}
\newcommand\rmin{\mathrm{min}}
\newcommand\rP{\mathrm{P}} 
\newcommand\rt{\mathrm{t}}
\newcommand\rTr{\mathrm{Tr}}

\newcommand\kB{k_\mathrm{B}}
\newcommand\kT{k_\mathrm{B}T}

\renewcommand{\vec}[1]{\mathbf{#1}}
\newcommand{\mat}[1]{\mathbf{\underline{#1}}}

\begin{document}

\small

\title[Geometric measures of convex uniaxial solids of revolution in $\mathbb{R}^4$]{Geometric measures of uniaxial solids of revolution in $\mathbb{R}^4$ and their relation to the second virial coefficient}

\author{\fnm{Markus} \sur{Kulossa}}\email{markus.kulossa@uni-rostock.de}

\author*{\fnm{Joachim} \sur{Wagner}}\email{joachim.wagner@uni-rostock.de}

\affil{\orgdiv{Institut f\"{u}r Chemie}, \orgname{Universit\"{a}t Rostock}, \orgaddress{\postcode{18051} \city{Rostock}, \country{Germany}}}

\abstract{We provide analytical expressions for the second virial coefficients of hard, convex, monoaxial solids of revolution in $\mathbb{R}^4$. 
The excluded volume per particle and thus the second virial coefficient is calculated using quermassintegrals and rotationally invariant mixed volumes based on the Brunn-Minkowski theorem. 
We derive analytical expressions for the mutual excluded volume of four-dimensional hard solids of revolution in dependence on their aspect ratio $\nu$ including the limits of infinitely thin oblate and infinitely long prolate geometries.  Using reduced second virial coefficients $B_2^*=B_2/V_{\rP}$ as size-independent quantities with $V_{\rP}$ denoting the $D$-dimensional particle volume, the influence of the 
particle geometry to the mutual excluded volume is analyzed for various shapes. Beyond the aspect ratio $\nu$, the detailed particle shape influences the reduced second virial coefficients $B_2^*$.
We prove that for $D$-dimensional spherocylinders in arbitrary-dimensional Euclidean spaces $\mathbb{R}^D$ their excluded volume solely depends on at most three intrinsic volumes, whereas for different convex geometries
$D$ intrinsic volumes are required. For $D$-dimensional ellipsoids of revolution, the general parity $B_2^*(\nu)=B_2^*(\nu^{-1})$ is proven.
}

\keywords{convex solids, quermassintegrals, excluded volumes}

\maketitle

\section{Introduction}\label{sec1}

Place two identical coins on a table and roll one coin completely around the second coin's rim: The moving coin completes two revolutions during the turn around the
first coin which is known as the coin rotation paradox. 
To analyze this phenomenon, in addition to a cardioid curve generated by a point on the perimeter of the second coin, the excluded volume between two rigid bodies can be used: 
The excluded volume is the volume inaccessible for the second body in the vicinity of the first one. The perimeter of the excluded volume is the curve generated
by the center of the second body during its turn around the first one preserving contact distance. 
As visualized in Fig. \ref{fig:vex}, for two circles with radius $r_0$ the two-dimensional excluded volume (or area) is a circle with radius $2r_0$ and area $4\pi r_0^2$, i.e.,
per circle $2\pi r_0^2$  and thus twice the two-dimensional volume of one circle. Hence, the number of revolutions of the turning coin is the ratio of the excluded volume's perimeter
to the coin's perimeter.

While this problem can be generalized to coins with unequal radii, also rolling a disk around other, possibly anisotropic shapes is feasible. 
When both objects are anisotropic, the excluded area additionally depends on the relative orientation of these bodies as shown in Fig. \ref{fig:vex} (right) for stadia. 
Here, an orientation-averaged, two-dimensional excluded volume results which can be generalized to arbitrary, convex hard bodies in Euclidean spaces $\mathbb{R}^D$.

The excluded volume of hard solids has implications on the physics of multi-particle systems consisting of such objects: The phase behavior and equation of state of
hard-body systems are governed by their geometry \cite{Kihara1953,Torquato2010}. 
Kamerling-Onnes \cite{KamerlinghOnnes1902} proposed the virial expansion
\begin{align}
\label{eq:virial_expansion_rho}
Z &=\dfrac{p}{\varrho \kB T} = 1 + B_2 \varrho + B_3 \varrho^2 + \dots = 1 + \sum\limits_{i=2}^\infty B_i \varrho^{i-1}
\end{align}
for the real gas factor $Z$ of imperfect gases as a heuristic equation of state, an expansion in powers of the number density $\varrho$ 
where $p$ denotes the pressure, $\kB$ Boltzmann's constant, and $T$ the temperature. The coefficients $B_i$ are
virial coefficients of order $i$. Introducing the packing fraction $\eta=\varrho V_{\rm P}$ with $V_{\rm P}$ denoting the $D$-dimensional
particle volume, Eq. \eqref{eq:virial_expansion_rho} can be reformulated as an expansion in powers of the packing fraction $\eta$
\begin{align}
\label{eq:real_gas_factor_volume_fraction}
Z &= \dfrac{p V_\rP}{\eta\kB T} = 1 + B_2^*\eta + B_3^*\eta^2 + \dots = 1 + \sum\limits_{i=2}^\infty B_i^* \eta^{i-1}
\end{align}
with the dimensionless, reduced virial coefficients $B_i^*=B_i/V_{\rm P}^{i-1}$ as expansion coefficients. Introducing a generalized pressure as
ratio of energy and $D$-dimensional volume, the virial expansion can be used in Euclidean spaces $\mathbb{R}^D$ with arbitrary dimension $D$.

\begin{figure}[ht]
\centering
\includegraphics[width=0.30\textwidth, valign=c]{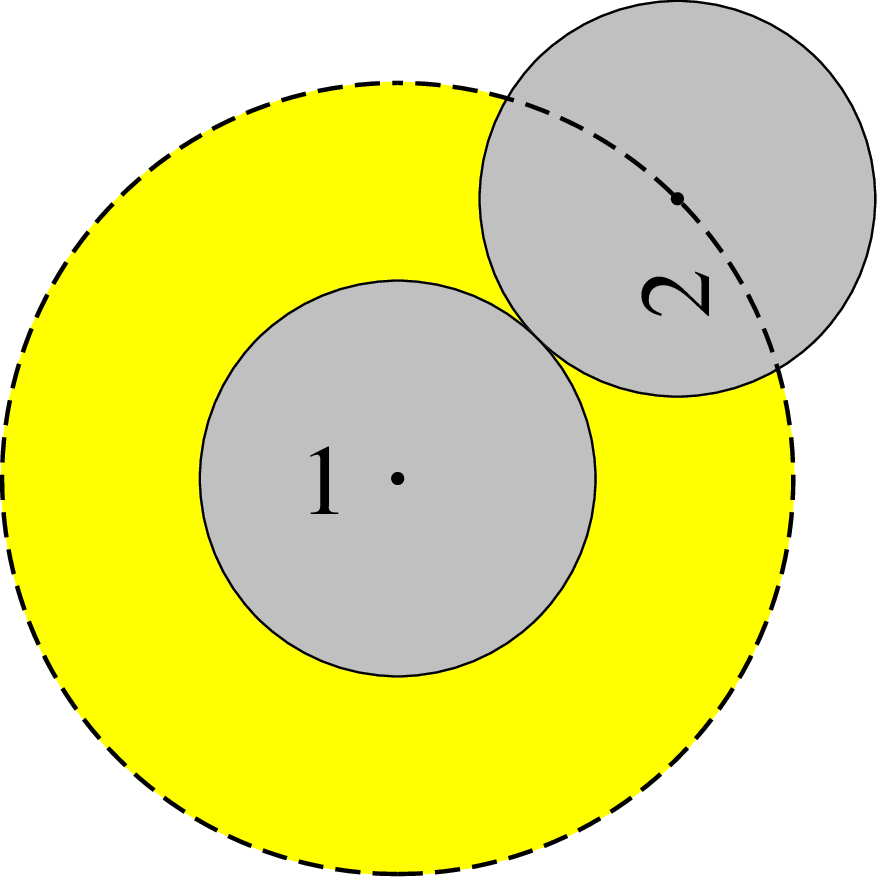} \qquad \quad
\includegraphics[width=0.40\textwidth, valign=c]{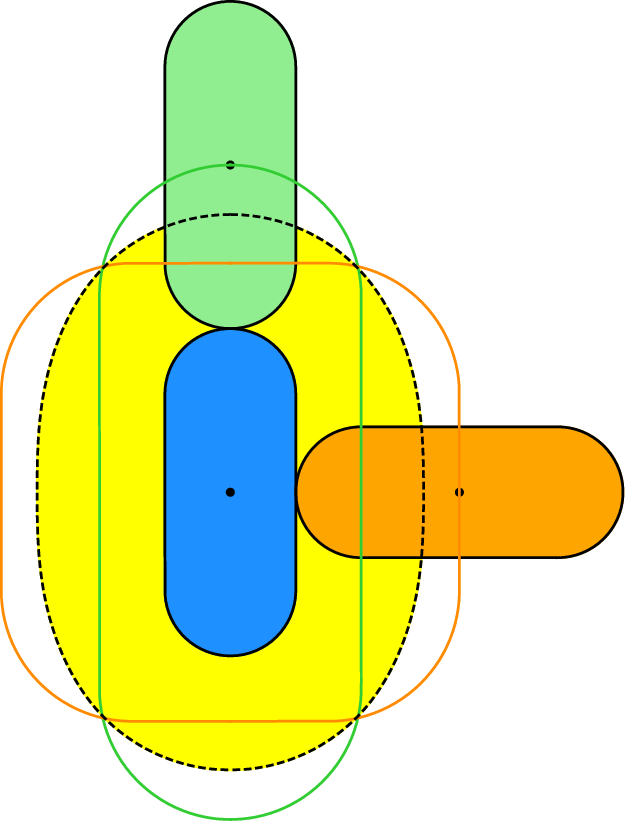}
\caption{Left: Excluded volume of two hard disks. Right: Orientation-averaged excluded volume of two stadia with aspect ratio $\nu=2.5$.}\label{fig:vex}
\end{figure}

The first virial coefficients were calculated for hard spheres in $\mathbb{R}^3$. Starting point is the calculation of the second virial coefficient by Jäger \cite{Jaeger1896}. Later, Boltzmann, van der Waals, and van Laar provided analytical solutions for virial coefficients of hard spheres up to the fourth order \cite{Boltzmann1896,Waals1899,Boltzmann1899,VanLaar1899,Nijboer1952}. Up to now, analytical solutions for virial coefficients of order five and higher are not available. Later on, analytical expressions for virial coefficients of isotropic particles in $\mathbb{R}^2$ \cite{Tonks1936,Rowlinson1964,Hemmer1965} and Euclidean spaces with dimension $D>3$ \cite{Luban1982,Clisby2004,Lyberg2005,Urrutia2022} were reported up to the fourth order.  

Starting point for studies of non-spherical particles is the seminal work of Onsager \cite{Onsager1949} with infinitely thin hard rods as a model system for liquid crystals \cite{degennes1993,Mederos2014,khoo2022}.
Few years later, Isihara and Hadwiger independently provided an analytical expression for the second virial coefficient of arbitrary convex particles in $\mathbb{R}^3$ \cite{Isihara1950,Isihara1951,Hadwiger1951}. The analogue 
for convex solids in $\mathbb{R}^2$ is provided by Boubl{\'\i}k \cite{Boublik1975, Tarjus1991}. For concave geometries in $\mathbb{R}^2,\mathbb{R}^3$, and $\mathbb{R}^4$ only few analytical or semianalytical expressions
for the second virial coefficient are available \cite{Isihara1951a,Rowlinson1985,Boublik1986,Kulossa2023}. For selected concave, rigid shapes, numerical values for their second virial coefficients as the mutually excluded volume per particle have been provided \cite{Gravish2012,Kornick2021}. 

Higher order virial coefficients, even for convex, anisometric particles are only numerically available, e.g., by means of Monte Carlo integration \cite{Singh2004,Marienhagen2021}. 
Analytical expressions for the second virial coefficients of solids of revolution in $\mathbb{R}^3$ employing the Isihara-Hadwiger theorem are provided in a previous work \cite{Herold2017}.
Recently, the Brunn-Minkowski theorem has been applied to calculate second virial coefficients of various convex particles in Euclidean spaces with dimension $D>3$ \cite{Schneider2013,Santalo2004,Torquato2022}. 

In this work, we use the Brunn-Minkowski theorem to provide analytical expressions for the second virial coefficients $B_2$ of convex, uniaxial solids of revolution in $\mathbb{R}^4$. 
Therefore, analytical expressions for their quermassintegrals (and intrinsic volumes) are derived depending on their meridian curve. The results for selected geometries are compared and the detailed influence of the particle shape is analyzed. The comparison of solids of revolution with identical meridian curve in Euclidean spaces with different dimensionality can provide fruitful insights to the self organization of condensed matter in $\mathbb{R}^3$. Beyond the impact on real physical systems, e.g., equation-of-state data and phase behavior, insights into close packings in higher dimensional spaces can be expected based on the provided geometric measures \cite{Torquato2013,Torquato2018,Schultz2022,Nezbeda2024}. 


\section{Theoretical background}\label{sec2} 

\subsection{Virial theory}\label{virial_theory}
Mayer and Mayer have shown by means of statistical mechanics, starting from the grand canonical partition function, that the virial coefficient of order $i$ depends on interactions in an $i$-particle
cluster \cite{Mayer1937,Mayer1940}. The second virial coefficient, describing the initial departure from the ideal-gas behavior in the low-density limit, reads as
\begin{align}
B_2 = -\dfrac{1}{2V} \left(Z_2 - Z_1^2 \right)
\end{align}
with $Z_N$ denoting the configuration integral for $N$ particles in a system volume $V$. Let $U(\vec{r}_1,\vec{r}_2)$ be the potential energy of two particles with centers located at $\vec{r}_1$ and $\vec{r}_2$, the second virial coefficient can be written as 
\begin{align}
B_2 = -\dfrac{1}{2V} \iint \left\{\exp\left[-\dfrac{U(\vec{r}_1,\vec{r}_2)}{\kT} \right]-1 \right\} \,\rd^D \vec{r}_1 \,\rd^D \vec{r}_2\,.
\end{align}
With the coordinate system's origin at $\vec{r}_1$ and the distance vector $\vec{r}_{12}=\vec{r}_2-\vec{r}_1$,
performing the integration over $\rd^D\vec{r}_1$, the second virial coefficient can be reformulated as
\begin{align}
\label{eq:B2_r12}
B_2 = - \dfrac{1}{2} \int \left\{\exp\left[-\dfrac{U(\vec{r}_{12})}{\kT} \right]-1 \right\} \,\rd^D \vec{r}_{12}\,.
\end{align}
Here, the integration over $\rd^D\,\vec{r}_1$ with
\begin{align}
\rd^D \vec{r}  &= r^{D-1}\,\rd r \prod\limits_{i=0}^{D-2} \sin^i \left( \phi_i \right) \,\rd \phi_i = \rd^D V
\end{align}
as $D$-dimensional volume element using polar coordinates results in the system volume $V$.
The integrand in Eq.\eqref{eq:B2_r12} 
\begin{align}
\label{eq:Mayer_f_function}
f_{12} = \exp\left[-\dfrac{U (\vec{r}_{12})}{\kT}\right]-1 
\end{align}
is the Mayer $f$ function depending on the interaction of both particles. For hard-body interaction, the
potential reads as 
\begin{align}
\label{eq:hard_body_interaction}
U(\vec{r}_{12}) = \left\{ \begin{array}{rcl}
\infty & : & r_{12} < \sigma \\
0 & : & r_{12} \geq \sigma \\
\end{array} \right.
\end{align} 
with $\sigma$ denoting the contact distance which in the case of uniaxial, anisometric shapes depends on the unit vectors 
$\hat{\vec{u}}_1$, $\hat{\vec{u}}_2$ denoting their orientation, and the direction of the distance vector $\hat{\vec{r}}_{12}=\vec{r}_{12}/r_{12}$.
Hence, combining Eqs. \eqref{eq:Mayer_f_function} and \eqref{eq:hard_body_interaction} results in 
\begin{align}
f_{12} = \left\{ \begin{array}{rcl}
-1 & : & r_{12} < \sigma  \\
 0 & : & r_{12} \geq \sigma \\
\end{array} \right.
\end{align} 
for the Mayer $f$ function, irrespective of the thermal energy $\kB T$: For an overlap, $f_{12}=-1$ is obtained while this integrand vanishes for non-overlap configurations.

For $D$-dimensional spheres with radius $r_0$ as radially symmetric bodies, the contact distance is, independent of the orientation, $\sigma=2 r_0$,
leading to 
\begin{align}
B_2 = \dfrac{1}{2} \beta_D \int_0\limits^\sigma r^{D-1} \rd r
\end{align}
with 
\begin{align}\label{eq:beta_D}
\beta_D = \dfrac{2 \pi^{D/2}}{\Gamma(D/2)}
\end{align}
denoting the surface area of a $D$-dimensional sphere with radius $r_0=1$ where
\begin{align}
\Gamma(z) = \int\limits_0^\infty t^{z-1} \exp\left( -t \right) \,\rd t
\end{align}
denotes the Euler gamma function. Performing the remaining integration over $\rd r$ leads
to
\begin{align}\label{eq:B2_sph}
B_2 = \dfrac{1}{2} \kappa_D (2 r_0)^D
\end{align}
with the volume 
\begin{align}\label{eq:kappa}
\kappa_D = \dfrac{\pi^{D/2}}{\Gamma (1+D/2)}
\end{align}
of the $D$-dimensional unit sphere for the second virial coefficient of a $D$-dimensional sphere.

Since in the case of uniaxial, anisometric shapes the contact distance depends on the orientations $\hat{\vec{u}}_1$, $\hat{\vec{u}}_2$, and $\hat{\vec{r}}_{12}$, the 
orientation-averaged Mayer $f$ function has to be used as integrand leading to
\begin{align}
B_2 = - \dfrac{1}{2} \int \langle f_{12}\rangle_{\hat{\vec{u}}_2} \,\rd^D \vec{r}_{12}
\end{align}
where the direction $\hat{\vec{u}}_1$ defines the orientation of the coordinate system. 

For hard-body interaction, the second virial coefficient 
\begin{align}\label{eq:B2Vex}
B_2 = \dfrac{1}{2} V_\rex
\end{align}
is the mutual excluded volume per particle and thus solely related to geometric properties of the considered shape.


\subsection{Excluded volumes of convex solids}

Let $K,L\in\mathbb{R}^D$ be two convex bodies. The rotation-averaged excluded volume $V_{\rex}(K,L)$ is
\begin{align}\label{eq:Vex_KL}
V_\rex (K,L)=\dfrac{1}{\kappa_D} \sum\limits_{i=0}^D \binom{D}{i}  W_i(K)\, W_{D-i}(L)\,,
\end{align}
known as Brunn-Minkowski theorem \cite{Groemer1972,Schneider2013,Torquato2022}.
The quantity $W_i(K)$ is the mixed volume of $(D-i)$ convex bodies $K$ and $i$ $D$-dimensional unit spheres 
denoted as $i$-th quermassintegral of $K$ which relates to the intrinsic volume $\upsilon_i$ as
\begin{align}\label{eq:intrinsic_volumes}
\upsilon_i (K) = \binom{D}{i} \dfrac{W_{D-i}(K)}{\kappa_{D-i}}
\end{align}
for a convex set. 

If $L$ is a $D$-dimensional sphere with radius $\varepsilon$, using 
\begin{align}
W_i^{(\rm sph)}(\varepsilon)&= \kappa_D \varepsilon^{D-i}
\end{align}
as a sphere's quermassintegral of order $i$ \cite{Schneider2013}, Eq. \eqref{eq:Vex_KL}
reads as
\begin{align}
V_\rex(K,\varepsilon) &= \sum\limits_{i=0}^D \binom{D}{i} W_i(K)\varepsilon^i\,,
\end{align}
also known as Steiner formula \cite{Morvan2008}. This excluded volume $V_\rex(K,\varepsilon)$ is also the volume of $K$'s $\varepsilon$-neighborhood as visualized in Fig. \ref{fig:varepsilon_neighorhood}. 

\begin{figure}
\centering
\includegraphics[width=0.3\columnwidth]{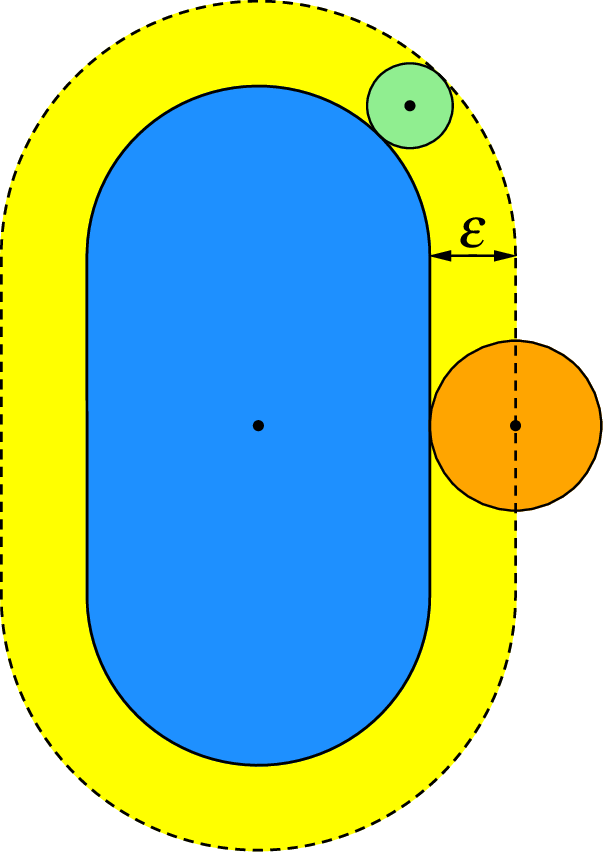} 
\caption{\label{fig:varepsilon_neighorhood} The area of the $\varepsilon$-neighborhood of a stadium is identical to the excluded area of a stadium and a disk with radius $\varepsilon$. } 
\end{figure}

For identical convex bodies $K=L$, the Brunn-Minkowski theorem reads as 
\begin{align}\label{eq:Vex_K}
V_\rex(K) = V_\rex(K,K) = \dfrac{1}{\kappa_D} \sum\limits_{i=0}^D \binom{D}{i} W_i(K) W_{D-i}(K) 
\end{align}
and thus provides an analytical expression for the second virial coefficient $B_2(K) =V_\rex(K)/2$ of a convex particle since due to the rotational invariance of quermassintegrals directly the
orientation average is obtained.

Quermassintegrals $W_i$ of orders $i=\{0,1\}$ and $i=\{D-1,D\}$ are related to common geometric measures as
\begin{subequations}\label{eq:W}
\begin{align}
W_0(K) &= V_\rP(K)\,, \label{eq:W_0} \\
W_1(K) &= \dfrac{1}{D} \,S_\rP(K)\,, \\
W_{D-1}(K) &= \tilde{R}_\rP(K) \kappa_D\,, \label{eq:W_D-1} \\
W_D(K) &= \kappa_D
\end{align}
\end{subequations}
where $V_\rP$ denotes the volume of $K$, $S_\rP$ its total surface area, and $\tilde{R}_\rP$ its mean radius of curvature \cite{Torquato2013,Torquato2022}.

Combining Eqs. \eqref{eq:B2Vex}, \eqref{eq:Vex_K}, \eqref{eq:W},  and \eqref{eq:kappa}, in the two- and three-dimensional Euclidean space, virial coefficients for arbitrary convex solids $K$
can be calculated. In the two-dimensional Euclidean space $\mathbb{R}^2$, 
\begin{align}\label{eq:B2_2d}
B_2(K) = V_\rP(K) + \dfrac{1}{2} S_\rP(K)\, \tilde{R}_\rP(K) = V_\rP(K) + \dfrac{S_\rP(K)^2}{4 \pi}
\end{align}
is obtained \cite{Boublik1975}. The relation for the three-dimensional Euclidean space $\mathbb{R}^3$
\begin{align}
B_2(K) &= V_\rP(K) + S_\rP(K)\, \tilde{R}_\rP(K)
\end{align}
is known as Isihara-Hadwiger theorem, independently derived by Isihara \cite{Isihara1950,Isihara1951,Isihara1951a}  and Hadwiger \cite{Hadwiger1951}.  
Using these relations, second virial coefficients for various convex shapes in $\mathbb{R}^2$ and $\mathbb{R}^3$ have been calculated \cite{Torquato2013,Herold2017,Irrgang2017,Kulossa2023b}.

In $\mathbb{R}^4$, the second virial coefficient reads as
\begin{align}\label{eq:B2(K)}
B_2(K) &= V_\rP(K) + S_\rP(K) \, \tilde{R}_\rP(K) + \dfrac{6}{\pi^2} W_2^2(K)\,.
\end{align}
To calculate $W_2(K)$, the surface integral over continuous surface curvature
\begin{align}\label{eq:W2(K)}
W_2(K) &= \dfrac{1}{D}\oint H_\rP(K) \,\rd^{D-1} S
\end{align}
is required where $H_\rP$ is the mean curvature
\begin{align}\label{eq:HP_D}
H_\rP &= \dfrac{1}{D-1}\sum\limits_{i=1}^{D-1} \dfrac{1}{R_i}
\end{align}
depending on the principal radii of curvature $R_i$ of a convex solid \cite{Groemer1972,Herold2017}.


\section{Geometric measures of four-dimensional uniaxial solids of revolution}

Let $K$ be a convex particle in $\mathbb{R}^4$. Its second virial coefficient $B_2(K)$ can analytically be calculated with Eq. \eqref{eq:B2(K)} using its geometric measures volume $V_\rP(K)$, surface area $S_\rP(K)$, second quermassintegral $W_2(K)$, and mean radius of curvature $\tilde{R}_\rP(K)$. 
Uniaxial solids of revolution with meridian curve $r(z)$ can be parameterized conveniently in hypercylindrical coordinates  
\begin{align}\label{eq:hcyl_coordinates}
\Psi = \begin{pmatrix}
w \\
x \\
y \\
z \\
\end{pmatrix}= \begin{pmatrix}
r \sin \vartheta \sin \chi \cos \varphi \\
r \sin \vartheta \sin \chi \sin \varphi \\
r \sin \vartheta \cos \chi \\
r \cos \vartheta \\
\end{pmatrix} = \begin{pmatrix}
r(z) \sin \chi \cos \varphi \\
r(z) \sin \chi \sin \varphi \\
r(z) \cos \chi \\
z \\
\end{pmatrix}\,,
\end{align}
related to Cartesian or polar coordinates.
Since the volume of a uniaxial solid of revolution in $\mathbb{R}^D$ is simply
\begin{align}\label{eq:V_P^D}
V_\rP(K) = \kappa_{D-1} \int r^{D-1}(z)\,\rd z\,,
\end{align}
using hypercylindrical coordinates, the expression
\begin{align}\label{eq:VP}
V_\rP(K) 
	  &= \dfrac{4}{3} \pi  \int\limits_{z_{\rm min}}^{z_{\rm max}} r^3(z)\,\rd z
\end{align}
results in $\mathbb{R}^4$ for a particle with length  $l=z_{\rm max}-z_{\rm min}$.
If $r(z_{\rm min})=r(z_{\rm max})=0$, its surface area can be written as
\begin{align}\label{eq:S_P^D}
S_\rP(K) = \beta_{D-1} \int r^{D-2}(z) \left\{1+\left[\dfrac{\rd r(z)}{\rd z} \right]^2 \right\}^{1/2}\,\rd z
\end{align}
in $\mathbb{R}^D$ with 
\begin{align}\label{eq:ds}
\rd s &= \left\{\left[\rd r(z)\right]^2 + \left[\rd z\right]^2 \right\}^{1/2} 
= \left\{1+\left[\dfrac{\rd r(z)}{\rd z} \right]^2 \right\}^{1/2}\,\rd z 
\end{align}
being the infinitesimal arc length element of the meridian curve with the result
\begin{align}\label{eq:MP}
S_\rP(K) &= 4 \pi \int\limits_{z_{\rm min}}^{z_{\rm max}} r^2(z) \left[1+\dot{r}^2(z) \right]^{1/2}\,\rd z
\end{align}
in $\mathbb{R}^4$ \cite{Herold2017,Kulossa2023}.

The second quermassintegral $W_2(K)$ depends on the principal radii of curvature $R_i$ of the convex solid [Eqs. \eqref{eq:W2(K)} and \eqref{eq:HP_D}],
accessible as reciprocal eigenvalues of the Weingarten map 
\begin{align}
\underline{\mathcal{W}} &= \mat{I}_\rP^{-1}  \mat{II}_\rP\,.
\end{align}
The Weingarten map $\underline{\mathcal{W}}$ is also known as shape operator where $\mat{I}_{\rP}$ and $\mat{II}_{\rP}$ denote the first and second fundamental form of the solid's surface $\Psi$.

Using the first derivatives
\begin{align}
\dfrac{\partial \Psi}{\partial \varphi} &= \begin{pmatrix}
-r(z) \sin \chi \sin \varphi \\
r(z) \sin \chi \cos \varphi \\
0 \\
0
\end{pmatrix}\,, ~
\dfrac{\partial \Psi}{\partial \chi} = \begin{pmatrix}
r(z) \cos \chi \cos \varphi \\
r(z) \cos \chi \sin \varphi \\
-r(z) \sin \chi \\
0
\end{pmatrix}\,,
~ 
\dfrac{\partial \Psi}{\partial z} = \begin{pmatrix}
\dot{r}(z) \sin \chi \cos \varphi \\
\dot{r}(z) \sin \chi \sin \varphi \\
\dot{r}(z) \cos \chi \\
1
\end{pmatrix}
\end{align}
the first fundamental form 
\begin{align}
\mat{I}_\rP = \begin{bmatrix}
\left(\dfrac{\partial \Psi}{\partial \varphi} \right) \cdot \left(\dfrac{\partial \Psi}{\partial \varphi} \right) & \quad \left(\dfrac{\partial \Psi}{\partial \varphi} \right) \cdot \left(\dfrac{\partial \Psi}{\partial \chi} \right) & \quad \left(\dfrac{\partial \Psi}{\partial \varphi} \right) \cdot \left(\dfrac{\partial \Psi}{\partial z} \right) \vspace{3mm}\\
\left(\dfrac{\partial \Psi}{\partial \chi} \right) \cdot \left(\dfrac{\partial \Psi}{\partial \varphi} \right) &\quad \left(\dfrac{\partial \Psi}{\partial \chi} \right) \cdot \left(\dfrac{\partial \Psi}{\partial \chi} \right) &\quad \left(\dfrac{\partial \Psi}{\partial \chi} \right) \cdot \left(\dfrac{\partial \Psi}{\partial z} \right) \vspace{3mm} \\
\left(\dfrac{\partial \Psi}{\partial z} \right) \cdot \left(\dfrac{\partial \Psi}{\partial \varphi} \right) &\quad \left(\dfrac{\partial \Psi}{\partial z} \right) \cdot \left(\dfrac{\partial \Psi}{\partial \chi} \right) &\quad \left(\dfrac{\partial \Psi}{\partial z} \right) \cdot \left(\dfrac{\partial \Psi}{\partial z} \right) \\
\end{bmatrix}
\end{align}
of a uniaxial solid of revolution in $\mathbb{R}^4$ reads as 
\begin{align}
\mat{I}_\rP = \begin{bmatrix}
r^2(z) \sin^2 \chi & 0 & 0 \vspace{3mm}\\
0 & r^2(z) & 0 \vspace{3mm} \\
0 & 0 & 1+\dot{r}^2(z) \\
\end{bmatrix}\,.
\end{align}
The second fundamental form can be written as 
\begin{align}\label{eq:2nd_fundamental}
\mat{II}_\rP = \begin{bmatrix}
\hat{\vec{n}} \cdot \dfrac{\partial^2 \Psi}{\partial \varphi^2} &\quad \hat{\vec{n}} \cdot \dfrac{\partial^2 \Psi}{\partial \varphi \partial \chi } &\quad \hat{\vec{n}} \cdot \dfrac{\partial^2 \Psi}{\partial \varphi \partial z } \vspace*{3mm}\\
\hat{\vec{n}} \cdot \dfrac{\partial^2 \Psi}{\partial \chi \partial \varphi} &\quad \hat{\vec{n}} \cdot \dfrac{\partial^2 \Psi}{\partial \chi^2 } &\quad \hat{\vec{n}} \cdot \dfrac{\partial^2 \Psi}{\partial \chi \partial z } \vspace*{3mm}\\
\hat{\vec{n}} \cdot \dfrac{\partial^2 \Psi}{\partial z \partial \varphi} &\quad \hat{\vec{n}} \cdot \dfrac{\partial^2 \Psi}{\partial z \partial \chi } &\quad \hat{\vec{n}} \cdot \dfrac{\partial^2 \Psi}{\partial z^2 }\\
\end{bmatrix}
\end{align}
with the second derivatives 
\begin{subequations}
\begin{align}
\dfrac{\partial^2 \Psi}{\partial \varphi^2} &= \begin{pmatrix}
-r(z) \sin \chi \cos \varphi \\
-r(z) \sin \chi \sin \varphi \\
0 \\
0
\end{pmatrix}\,, \quad ~
\dfrac{\partial^2 \Psi}{\partial \varphi \partial \chi } = \dfrac{\partial^2 \Psi}{\partial \chi \partial \varphi } = \begin{pmatrix}
-r(z) \cos \chi \sin \varphi \\
r(z) \cos \chi \cos \varphi \\
0 \\
0
\end{pmatrix}\,, \\
\nonumber \\
\dfrac{\partial^2 \Psi}{\partial \chi^2} &= \begin{pmatrix}
-r(z) \sin \chi \cos \varphi \\
-r(z) \sin \chi \sin \varphi \\
-r(z) \cos \chi \\
0
\end{pmatrix}\,, \quad ~
\dfrac{\partial^2 \Psi}{\partial \chi \partial z } = \dfrac{\partial^2 \Psi}{\partial z \partial \chi } = \begin{pmatrix}
\dot{r}(z) \cos \chi \cos \varphi \\
\dot{r}(z) \cos \chi \sin \varphi \\
-\dot{r}(z) \sin \chi \\
0
\end{pmatrix}\,, \\
\nonumber \\
\dfrac{\partial^2 \Psi}{\partial z^2} &= \begin{pmatrix}
\ddot{r}(z) \sin \chi \cos \varphi \\
\ddot{r}(z) \sin \chi \sin \varphi \\
\ddot{r}(z) \cos \chi \\
0
\end{pmatrix}\,, \qquad 
\dfrac{\partial^2 \Psi}{\partial \varphi \partial z } = \dfrac{\partial^2 \Psi}{\partial z \partial \varphi } = \begin{pmatrix}
-\dot{r}(z) \sin \chi \sin \varphi \\
\dot{r}(z) \sin \chi \cos \varphi \\
0 \\
0
\end{pmatrix}\,,
\end{align}
\end{subequations}
and the normal field $\hat{\vec{n}} \equiv \hat{\vec{n}}(\varphi, \chi, z)$ of the solid's surface. 
 
A normal, perpendicular to three derivatives $\vec{a}=\partial\Psi/\partial\varphi$, $\vec{b}=\partial\Psi/\partial\chi$ and $\vec{c}=\partial\Psi/\partial z$ can in $\mathbb{R}^4$ be written as
\begin{align}
\vec{n}(\vec{a},\vec{b},\vec{c}) = \det \left(\begin{array}{cccc}
e_w & a_w & b_w & c_w \\
e_x & a_x & b_x & c_x \\
e_y & a_y & b_y & c_y \\
e_z & a_z & b_z & c_z \\
\end{array} \right)
\end{align}
using the canonical basis $(e_w,e_x,e_y,e_z)$ with the parities 
\begin{align}
  \vec{n} \left(\vec{a}, \vec{b}, \vec{c}\right) &= - \vec{n} \left(\vec{a}, \vec{c}, \vec{b}\right) \nonumber \\
= \vec{n} \left(\vec{b}, \vec{c}, \vec{a }\right)&= - \vec{n} \left(\vec{b}, \vec{a}, \vec{c}\right) \nonumber \\
= \vec{n} \left(\vec{c}, \vec{a}, \vec{b}\right) &= - \vec{n} \left(\vec{c}, \vec{b}, \vec{a}\right).
\end{align}
A uniaxial solid's of revolution second fundamental form using the normal field 
\begin{align}
\hat{\vec{n}} = \dfrac{\vec{n}}{\Vert\vec{n}\Vert}
 = \dfrac{1}{\left[1+\dot{r}^2(z) \right]^{1/2}} \begin{pmatrix}
-\sin \chi \cos \varphi \\
-\sin \chi \sin \varphi \\
-\cos \chi \\
\dot{r}(z) \\
\end{pmatrix}
\end{align}
reads as
\begin{align}
\mat{II}_\rP = \dfrac{1}{\left[1+\dot{r}^2(z) \right]^{1/2}} \begin{bmatrix}
r(z) \sin^2 \chi & 0 & 0 \\
0 & r(z) & 0 \\
0 & 0 & -\ddot{r}(z) \\
\end{bmatrix}\,.
\end{align}
From both fundamental forms, the Weingarten map 
\begin{align}
\underline{\mathcal{W}} = \begin{pmatrix}
 \dfrac{1}{r(z) \left[1+\dot{r}^2(z)\right]^{1/2}} & 0 & 0\\
0 &  \dfrac{1}{r(z) \left[1+\dot{r}^2(z) \right]^{1/2}} & 0 \\
0 & 0 & - \dfrac{\ddot{r}(z)}{\left[1+\dot{r}^2(z) \right]^{3/2}} \\
\end{pmatrix}
\end{align}
is obtained as a main diagonal matrix. Hence, the eigenvalues are directly accessible from the diagonal elements indicating the principal curvatures.
The mean curvature thus can be written with the trace of the Weingarten map as 
\begin{align}\label{eq:HP_W}
H_\rP = \dfrac{1}{D-1} \rTr\left(\underline{\mathcal{W}} \right) =  \dfrac{1}{D-1} \rTr\left(\mat{I}_\rP^{-1}  \mat{II}_\rP \right)\,.
\end{align}
Hence, for a uniaxial solid of revolution in $\mathbb{R}^4$ three principal radii of curvature 
\begin{align}
R_1(z) &= R_2(z) = r(z) \left[1+\dot{r}^2(z)\right]^{1/2} \label{eq:R1R2} \\
R_3(z) &= - \dfrac{1}{\ddot{r}(z)} \left[1+\dot{r}^2(z)\right]^{3/2} \label{eq:R3} 
\end{align}
result as inverses of the principal curvatures where two of them are identical.

Combining Eqs. \eqref{eq:R1R2}, \eqref{eq:R3}, and \eqref{eq:HP_D}, the mean curvature can be written as 
\begin{align}\label{eq:HP_R_4d}
H_\rP = \dfrac{1}{3} \left[\dfrac{1}{R_1(z)} + \dfrac{1}{R_2(z)} + \dfrac{1}{R_3(z)} \right]
\end{align}
in dependence on the meridian curve $r(z)$. Using the surface element $\rd^3 S = r^2(z) \sin \chi \, \rd \varphi \, \rd \chi \, \rd s$, the second quermassintegral in $\mathbb{R}^4$ with Eqs. \eqref{eq:W2(K)} and \eqref{eq:ds}, performing the integration over the angular coordinates $\varphi$ and $\chi$, finally reads as 
\begin{align}\label{eq:W2_an}
W_2(K) = \dfrac{\pi}{3}  \int\limits_{z_{\rm min}}^{z_{\rm max}}  r^2(z)  \left[\dfrac{1}{R_1(z)}+\dfrac{1}{R_2(z)}+\dfrac{1}{R_3(z)}\right] \left[1+\dot{r}^2(z) \right]^{1/2} \,\rd z
\end{align}
for convex solids with continuous surface curvatures.

The mean radius of curvature of a convex solid with continuous surface curvature can, in general, be calculated from its principal radii of curvature as 
\begin{align}
\tilde{R}_\rP = \dfrac{1}{\beta_D} \oint \dfrac{1}{D-1} \left[ \sum\limits_{i=1}^{D-1} R_i(\phi_0, \phi_1, \dots \phi_{D-2}) \right] \prod\limits_{j=0}^{D-2} \sin^j \left(\phi_j \right)\,\rd \phi_j
\end{align}
normalized to the surface area of a $D$-dimensional unit sphere.
In $\mathbb{R}^4$,
\begin{align}
\tilde{R}_\rP(K) = \dfrac{1}{6 \pi^2} \int\limits_0^{\pi} \int\limits_0^\pi \int\limits_0^{2\pi} \left[R_1(\vartheta, \chi, \varphi) + R_2(\vartheta, \chi, \varphi) + R_3(\vartheta, \chi, \varphi) \right] \sin^2 \vartheta \sin \chi\, \rd \varphi \, \rd \chi \, \rd \vartheta
\end{align}
is obtained. Employing the surface element 
\begin{align}\label{eq:SP_R1R2R3}
\rd^3 S = R_1 R_2 R_3 \sin^2 \vartheta \sin \chi \,\rd \varphi \,\rd \chi \,\rd \vartheta = r^2(z) \sin \chi \left[1+ \dot{r}^2(z) \right]^{1/2}  \,\rd \varphi \, \rd \chi \, \rd z\,,
\end{align}
the angular integral can be written as a surface integral 
\begin{align}\label{eq:RP_an}
\tilde{R}_\rP(K) = \dfrac{2}{3\pi} \int\limits_{z_{\rm min}}^{z_{\rm max}} r^2(z) \left[\dfrac{1}{R_2(z) R_3(z)}+\dfrac{1}{R_1(z) R_3(z)}+\dfrac{1}{R_1(z) R_2(z)}\right]  \left[1+\dot{r}^2(z) \right]^{1/2} \,\rd z
\end{align}
where again the integration over the angular coordinates $\varphi$ and $\chi$ results in a factor $4\pi$.

Using the relations to the other quermassintegrals in Eq. \eqref{eq:W}, analytical expressions for them can be summarized as 
\begin{subequations}\label{eq:Wi_K}
\begin{align}
W_0(K) &= V_\rP(K) \\
W_1(K) &= \dfrac{1}{4} \oint  \, \rd^3 S = \dfrac{1}{4} S_\rP (K) \\
W_2(K) &= \dfrac{1}{4} \oint \dfrac{1}{3} \left[\dfrac{1}{R_1} + \dfrac{1}{R_2} + \dfrac{1}{R_3} \right] \rd^3 S \label{eq:W2K}\\
W_3(K) &= \dfrac{1}{4} \oint \dfrac{1}{3} \left[\dfrac{1}{R_1 R_2} + \dfrac{1}{R_1 R_3} + \dfrac{1}{R_2 R_3} \right] \rd^3 S = \dfrac{1}{4}\tilde{R}_\rP \beta_4 = \tilde{R}_\rP \kappa_4 \label{eq:W3K}\\
W_4(K) &= \dfrac{1}{4} \oint \dfrac{1}{R_1 R_2 R_3} \,\rd^3 S = \dfrac{1}{4} \beta_4 = \kappa_4 
\end{align}
\end{subequations}
for convex solids with continuous surface curvature in $\mathbb{R}^4$. 

For arbitrary Euclidean spaces $\mathbb{R}^D$, this can, in general, be written as 
\begin{align}
\label{eq:W_i_continuous}
W_i(K) = \left\{ \begin{array}{lcl}
V_\rP(K) & :  &i=0 \vspace*{2mm}\\
\dfrac{1}{D} \bigointsss \,\rd^{D-1} S & :  &i=1 \vspace*{2mm}\\
\dfrac{1}{D} \bigointsss \left[\begin{pmatrix}D-1 \\ i-1 \end{pmatrix} \right]^{-1} \left[\sum\limits_{j=1}^{D-i+1} a_{D-i+1,j}\right] \, \rd^{D-1} S & : &2 \leq i \leq D
\end{array} \right. 
\end{align}
with 
\begin{align}
a_{ij} = \left\{ \begin{array}{lcl}
\dfrac{1}{R_j} & : & i=D-1 \vspace*{3mm}\\
\dfrac{1}{R_j}\sum\limits_{k=j+1}^{D-1} \dfrac{1}{R_k}\quad & : & i=D-2 \vspace*{3mm}\\ 
\dfrac{1}{R_j} \sum\limits_{k=j+1}^{i+1} a_{i+1,k} \quad & : & i<D-2
\end{array}\right. 
\end{align}
and 
\begin{align}
\rd^D S = \left[\,\prod\limits_{i=1}^{D} R_i \right] \left[ \prod\limits_{j=0}^{D-1} \sin^j\left(\phi_j\right) \,\rd \phi_j \right]
\end{align}
for the surface element. 

For convex solids with discontinuities of their surface curvature, additional contributions to quermassintegrals $W_i(K)$ need to be considered:
Additional contributions to $W_1(K)$ exist for uniaxial solids of revolution in $\mathbb{R}^4$ if $r(z_{\rm min})\neq 0$ or $r(z_{\rm max})\neq 0$. One- and two-dimensional singularities of the surface curvature where in the case of uniaxial solids of revolution in $\mathbb{R}^4$ only two-dimensional singularities can exist 
cause additional contributions to $W_2(K)$ and $W_3(K)$. Zero-dimensional singularities which can for uniaxial solids of revolution only exist at $r(z)=0$ do not cause additional contributions to $W_i$ for $0\le i\le D$.   
In the following, we provide quermassintegrals for various uniaxial solids of revolution in $\mathbb{R}^4$ including shapes with discontinuities in their surface curvature.


\section{Specific geometric measures for selected solids of revolution}\label{eq_sec:measures_4d_specific}

In the following subsections, we provide the geometric measures volume $V_\rP$, total surface area $S_\rP$, second quermassintegral $W_2$,
and mean radius of curvature $\tilde{R}_\rP$ for selected convex geometries. 
Uniaxial solids of revolution in $\mathbb{R}^4$ are defined by their meridian curve $r(z)$.
Based on the meridian curve $r(z)$, specific geometric measures of these geometries and the related, 
second virial coefficients are analytically accessible employing Eqs. \eqref{eq:W}, \eqref{eq:Vex_K}, and \eqref{eq:B2Vex}.

\subsection{Solids of revolution with continuous surface curvature}\label{sec:measures_without_sing}

\subsubsection{Hypersphere}

The simplest solid of revolution in the four-dimensional Euclidean space $\mathbb{R}^4$ is the isotropic hypersphere. 
Placing the center of a hypersphere in the coordinate system's origin, the meridian curve $r(z)$ of the hypersphere reads as 
\begin{align}\label{eq:Hsp_rz}
r(z) = \left(r_0^2 -z^2 \right)^{1/2}
\end{align}
with $r_0$ being the radius of the hypersphere and $z$ the coordinate on the axis of revolution resulting from hypercylindrical coordinates [Eq. \eqref{eq:hcyl_coordinates}]. 
Using Eq. \eqref{eq:VP}, the hypervolume of a hypersphere can simply be written as 
\begin{align}\label{eq:hsp_V}
V_\rP &= \dfrac{4}{3} \pi \int\limits_{-r_0}^{r_0}  \left(r_0^2-z^2 \right)^{3/2} \, \rd z =  \dfrac{1}{2} \pi^2 r_0^4\,
\end{align}
which is identical to $V_\rP = \kappa_4 r_0^4$ using Eq. \eqref{eq:kappa}. 
The total area $S_\rP$ of a surface without singularities is simply the lateral surface area $S_\rP^\prime$. 
The lateral surface area $S_\rP^\prime$ of a hypersphere can be calculated with Eqs. \eqref{eq:MP} and  \eqref{eq:Hsp_rz} 
leading to 
\begin{align}\label{eq:hsp_S}
S_\rP &= S_\rP^\prime = 4 \pi r_0  \int\limits_{-r_0}^{r_0} \left(r_0^2-z^2 \right)^{1/2} \, \rd z =2 \pi^2 r_0^3
\end{align}
as the expected results, identical to $S_\rP = \beta_4 r_0^3$ using Eq. \eqref{eq:beta_D} with $D=4$. 

For the second quermassintegral $W_2$ and the mean radius of curvature $\tilde{R}_\rP$ of a four-dimensional, convex solids of revolution the principal radii of curvature $R_1$, $R_2$, and $R_3$ are required [Eqs. \eqref{eq:W2_an} and \eqref{eq:RP_an}]. 
These can be calculated in dependence on the meridian curve as 
\begin{align}\label{eq:hsp_R1R2R3}
R_1 = R_2 = R_3 = r_0
\end{align}
for a hypersphere employing Eqs. \eqref{eq:R1R2} and \eqref{eq:R3}.  
Its principal radii of curvature are simply the radius of the hypersphere $r_0$.  
For the second quermassintegral $W_2$, Eq. \eqref{eq:W2_an} leads to 
\begin{align}\label{eq:hsp_W2}
W_2 &= \pi \int\limits_{-r_0}^{r_0} \left(r_0^2-z^2 \right)^{1/2}\,\rd z = \dfrac{1}{2} \pi^2 r_0^2
\end{align}
and the mean radius of curvature from Eq. \eqref{eq:RP_an} with 
\begin{align}\label{eq:hsp_R}
\tilde{R}_\rP &= \dfrac{2}{\pi r_0} \int\limits_{-r_0}^{r_0} \left(r_0^2 - z^2 \right)^{1/2} \, \rd z = r_0
\end{align}
is again simply the radius of the hypersphere $r_0$: The radius of curvature of a hypersphere is just the radius of the hypersphere. 

Using the relations in Eq. \eqref{eq:W}, the quermassintegrals $W_i$ of a four-dimensional hypersphere can be written as
\begin{align}
W_i = \dfrac{1}{2} \pi^2 r_0^{4-i},
\end{align}
agreeing with the general solution for $D$-dimensional hyperspheres 
\begin{align}
W_i = \kappa_D r_0^{D-i}
\end{align}
from the literature \cite{Santalo2004,Torquato2022}.

\subsubsection{Hyperellipsoid}

A uniaxial hyperellipsoid of revolution is the affine transformation of a hypersphere with radius $r_\req$ in $z$-direction by a factor $\nu$. 
The meridian curve $r(z)$ of this hyperellipsoid of revolution with the center in the origin of the coordinate system reads as 
\begin{align}\label{eq:Hell_rz}
r(z) = \left[r_\req^2 - \left(\dfrac{z}{\nu} \right)^2 \right]^{1/2}\,.
\end{align}
For the aspect ratio $\nu =1$, a hypersphere results. Hyperellipsoids of revolution with $\nu >1$ are prolate and with $\nu < 1$ are oblate solids of revolution.
The two-dimensional sections of such geometries are shown in Fig. \ref{fig:hell}. 

\begin{figure}[ht]
\centering
\includegraphics[width=0.2\textwidth,valign=c]{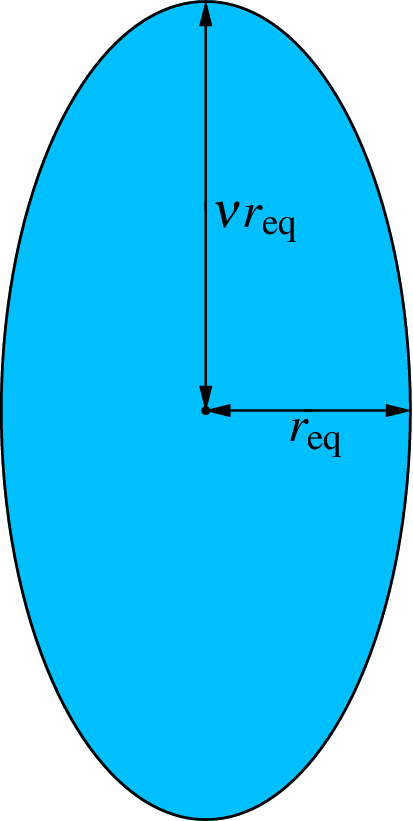} \qquad \qquad
\includegraphics[width=0.4\textwidth,valign=c]{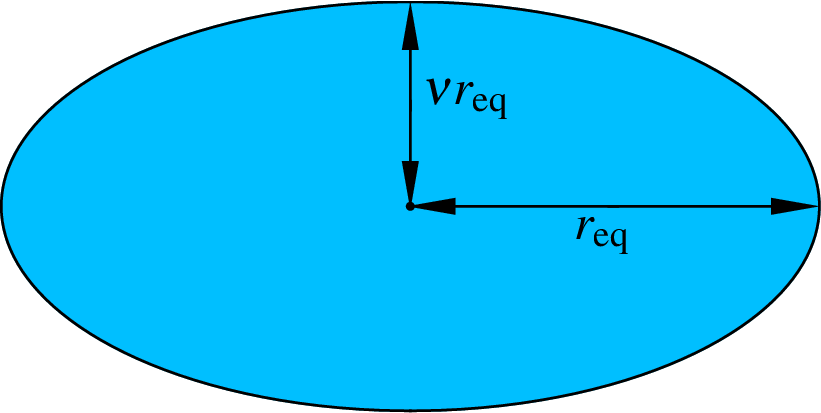}
\caption{Two-dimensional sections of prolate and oblate uniaxial hyperellipsoids of revolution with radius $r_\req$, aspect ratio $\nu=2$ (lhs) and aspect ratio $\nu=1/2$ (rhs).}\label{fig:hell}
\end{figure}

For the volume of a four-dimensional hyperellipsoid of revolution, Eq. \eqref{eq:VP} leads to 
\begin{align}\label{eq:Hell_VP}
V_\rP &= \dfrac{4}{3} \pi \int\limits_{- \nu r_\req}^{\nu r_\req} \left(r_\req^2 - \dfrac{z^2}{\nu^2} \right)^{3/2} \,\rd z = \dfrac{1}{2} \pi^2 \nu r_\req^4
\end{align}
for prolate and oblate hyperellipsoids. 
The total surface area $S_\rP$ for this geometry is the lateral surface area $S_\rP^\prime$ which can with Eq. \eqref{eq:MP} be written as  
\begin{align}
S_\rP &= S_\rP^\prime = 4 \pi  \int\limits_{- \nu r_\req}^{\nu r_\req} \left(r_\req^2 - \dfrac{z^2}{\nu^2}  \right) \left[\dfrac{\nu^4 r_\req^2 + z^2 \left(1-\nu^2 \right)}{\nu^2 \left(\nu^2 r_\req^2 -z^2 \right)} \right]^{1/2} \,\rd z \nonumber \\
&= 2 \pi^2 \nu^2 r_\req^3  \,_2\mathcal{F}_1\left(\dfrac{5}{2}, \dfrac{1}{2}, 2; 1-\nu^2 \right)
\end{align} 
where $_2\mathcal{F}_1$ is the Gauss hypergeometric function defined as 
\begin{align}\label{eq:F_P}
_2\mathcal{F}_1 \left(a,b,c;z \right) 
&= \dfrac{\Gamma(c)}{\Gamma(b)\Gamma(c-b)} \int\limits_0^1 t^{b-1} \left(1-t \right)^{c-b-1} \left(1-zt \right)^{-a} \rd t \nonumber \\
&= \sum\limits_{n=0}^\infty \dfrac{\left(a \right)_n \left(b \right)_n}{\left(c \right)_n} \dfrac{z^n}{n!} 
\end{align}
and 
\begin{align}
\left(q \right)_n = \prod_{k=0}^{n-1} \left(q+k \right) = \dfrac{\Gamma(q+n)}{\Gamma(q)}
\end{align}
the Pochhammer symbol. 

Using linear transformation and identities of the hypergeometric function \cite{NIST:DLMF}, the total surface area can also be denoted as 
\begin{align}\label{eq:SP_hell_EK}
S_\rP = \left\{\begin{array}{ccl}
\dfrac{8 \pi \nu r_\req^3}{3 (\nu^2-1)} \left[\left(2\nu^2-1 \right) \mathcal{E}\left(\dfrac{\sqrt{\nu^2-1}}{\nu} \right)- \mathcal{K} \left(\dfrac{\sqrt{\nu^2-1}}{\nu} \right) \right] 
 & : & \nu > 1 \vspace{3mm} \\
\dfrac{8 \pi \nu r_\req^3}{3 (1-\nu^2)} \left[\left(\dfrac{1-2\nu^2}{\nu} \right) \mathcal{E}\left(\sqrt{1-\nu^2} \right)+\nu \mathcal{K} \left(\sqrt{1-\nu^2} \right) \right] & : & \nu < 1
\end{array} \right. 
\end{align} 
with the complete elliptic integral of the first kind 
\begin{align}\label{eq:K_F}
\mathcal{K}(k) = \dfrac{\pi}{2} \,_2\mathcal{F}_1\left(\dfrac{1}{2}, \dfrac{1}{2},1;k^2 \right)
\end{align}
and the complete elliptic integral of the second kind
\begin{align}\label{eq:E_F}
\mathcal{E}(k) = \dfrac{\pi}{2} \,_2\mathcal{F}_1\left(-\dfrac{1}{2}, \dfrac{1}{2},1;k^2 \right)\,.
\end{align}
For the limiting case of a hypersphere with $\nu \to 1$, for both branches 
\begin{align}
\lim\limits_{\nu \to 1^+} S_\rP(\nu) = \lim\limits_{\nu \to 1^-} S_\rP(\nu) = 2 \pi^2 r_\req^3
\end{align}
results which agrees with the expression in Eq. \eqref{eq:hsp_S}. 
For infinitely thin hyperellipsoids of revolution 
\begin{align}
\lim\limits_{\nu \to 0} S_\rP(\nu) = \dfrac{8}{3} \pi r_\req^3
\end{align}
results. 

For the second quermassintegral $W_2$ and the mean radius of curvature $\tilde{R}_\rP$, again, the principal radii of curvature $R_1(z)$, $R_2(z)$, and $R_3(z)$ need to be calculated. 
Using Eqs. \eqref{eq:R1R2} and \eqref{eq:R3} 
\begin{align}
R_1(z) = R_2(z) = \dfrac{1}{\nu^2} \left[\nu^4 r_\req^2 + z^2\left(1-\nu^2 \right) \right]^{1/2}
\end{align}
and
\begin{align}
R_3(z) = \dfrac{1}{\nu^4 r_\req^2} \left[\nu^4 r_\req^2 + z^2\left(1-\nu^2 \right) \right]^{3/2}
\end{align}
are obtained. 
Then the second quermassintegral is, using Eq. \eqref{eq:W2_an}, 
\begin{align}
W_2 &= \dfrac{\pi}{3 \nu^3} \int\limits_{-\nu r_\req}^{\nu r_\req}  \left(\nu^2 r_\req^2 - z^2 \right)^{1/2} \left[\dfrac{\nu^4 r_\req^2}{\nu^4 r_\req^2 + z^2\left(1-\nu^2 \right)} + 2\nu^2 \right] \,\rd z \nonumber \\
&=\dfrac{1}{3} \pi^2 r_\req^2 \left(\dfrac{1-\nu^3}{1-\nu^2} \right)\,,
\end{align}
both for prolate and oblate hyperellipsoids of revolution. In the limit of a four-dimensional hypersphere 
\begin{align}\label{eq:hell_W2_lim1}
\lim\limits_{\nu\to 1^+} W_2\left( \nu \right) = \lim\limits_{\nu\to 1^-} W_2\left( \nu \right) = \dfrac{1}{2} \pi^2 r_\req^2
\end{align}
is obtained while 
\begin{align}\label{eq:hell_W2_lim0}
\lim\limits_{\nu\to 0} W_2\left( \nu \right) = \dfrac{1}{3} \pi^2 r_\req^2
\end{align}
results for infinitely thin hyperellipsoids of revolution.  Eq. \eqref{eq:hell_W2_lim1} is identical to the expression in Eq. \eqref{eq:hsp_W2}. 

Finally, the mean radius of curvature $\tilde{R}_\rP$ for hyperellipsoids of revolution can be written as 
\begin{align}\label{eq:RP_hell_F}
\tilde{R}_\rP &= \dfrac{2}{3 \pi \nu^3 }  \int\limits_{-\nu r_\req}^{\nu r_\req}  \left(\nu^2 r_\req^2 -z^2 \right)^{1/2} \left\{ \dfrac{2 \nu^6 r_\req^2 + \nu^4 \left[\nu^4 r_\req^2 + z^2 \left(1-\nu^2 \right) \right]}{\left[\nu^4 r_\req^2 + z^2 \left(1-\nu^2 \right) \right]^{3/2}} \right\} \,\rd z \nonumber \\
&= \nu^4 r_\req \,_2\mathcal{F}_1 \left(\dfrac{5}{2}, \dfrac{3}{2}, 2; 1-\nu^2 \right)
\end{align}
using Eq. \eqref{eq:RP_an}. 
Using complete elliptic integrals [Eqs. \eqref{eq:K_F} and \eqref{eq:E_F}], it can alternatively be denoted as
\begin{align}\label{eq:RP_hell_EK}
\tilde{R}_\rP=\left\{\begin{array}{lcc}
\dfrac{4}{3 \pi} \dfrac{\nu r_\req}{\nu^2-1} \left[\left(\nu^2-2 \right)\mathcal{E}\left(\dfrac{\sqrt{\nu^2-1}}{\nu}\right) +\mathcal{K}\left(\dfrac{\sqrt{\nu^2-1}}{\nu}\right) \right]  & : & \nu > 1 \vspace{3mm} \\
\dfrac{4}{3 \pi} \dfrac{r_\req}{1-\nu^2} \left[\left(2-\nu^2 \right)\mathcal{E}\left(\sqrt{1-\nu^2}\right) -\nu^2 \mathcal{K}\left(\sqrt{1-\nu^2}\right) \right]  & : & \nu < 1 \\
\end{array} \right.
\end{align}
for prolate and oblate hyperellipsoids of revolution. In the limit of a hypersphere,
\begin{align}
\lim\limits_{\nu\to 1^+} \tilde{R}_\rP\left( \nu \right) = \lim\limits_{\nu\to 1^-} \tilde{R}_\rP\left( \nu \right) = r_\req
\end{align}
results [Eq. \eqref{eq:hsp_R}] while for infinitely thin hyperellipsoids 
\begin{align}\label{eq:hell_RP_lim0}
\lim\limits_{\nu\to 0}\tilde{R}_\rP\left( \nu \right) = \dfrac{8}{3 \pi} r_\req
\end{align}
is obtained.

Using Eqs. \eqref{eq:W}, the general relation
\begin{align}
W_i &= \dfrac{1}{2} \pi^2 \nu^{i+1} r_\req^{4-i} \,_2\mathcal{F}_1\left(\dfrac{5}{2},\dfrac{i}{2}, 2; 1-\nu^2 \right)\,,
\end{align}
with
\begin{align}
_2\mathcal{F}_1\left(\dfrac{5}{2},0, 2; 1-\nu^2 \right) = 1
\end{align}
and 
\begin{align}
_2\mathcal{F}_1\left(\dfrac{5}{2},1, 2; 1-\nu^2 \right) = \dfrac{2}{3 \nu^3} \,\dfrac{1-\nu^3}{1-\nu^2}
\end{align}
is obtained in $\mathbb{R}^4$ \cite{NIST:DLMF}. 
This agrees with the general solution for oblate, $D$-dimensional ellipsoids of revolution 
\begin{align}\label{eq:hell_Wi}
W_i = \kappa_D \nu^{i+1} r_\req^{D-i} \,_2\mathcal{F}_1\left(\dfrac{D+1}{2},\dfrac{i}{2}, \dfrac{D}{2}; 1-\nu^2 \right)
\end{align} 
provided in the literature \cite{Santalo2004,Torquato2022}. 
Using the Pfaff transformation \cite{NIST:DLMF} with
\begin{align}
_2\mathcal{F}_1\left(a,b,c; z \right) = \left(1-z \right)^{-b}\,_2\mathcal{F}_1\left(c-a,b,c; \dfrac{z}{z-1} \right)\,,
\end{align}
the general quermassintegrals can also be reformulated as 
\begin{align}
W_i = \kappa_D \nu r_\req^{D-i} \, _2\mathcal{F}_1\left(-\dfrac{1}{2},\dfrac{i}{2}, \dfrac{D}{2}; 1-\dfrac{1}{\nu^2} \right)
\end{align}
for prolate, $D$-dimensional ellipsoids of revolution.

\subsubsection{Hyperspherocylinder} 

A $D$-dimensional spherocylinder is the union of a $D$-dimensional cylinder with height $h=2 \left(\nu-1 \right) r_\req$ and radius $r_\req$, capped by two $D$-dimensional hemispheres with radius $r_\req$. 
A four-dimensional hyperspherocylinder is a prolate solid of revolution without singularities on its surface curvature where the meridian curve $r(z)$ can be written as
\begin{align}\label{eq:Hsc_rz}
r(z)=\left\{\begin{array}{ccl}
\left\{r_\req^2-\left[z - \left(\nu -1\right) r_\req \right]^2 \right\}^{1/2} & : & \phantom{-} \left(\nu-1\right) r_\req < z \leq \nu r_\req \vspace{3mm}\\
r_\req & : & -\left( \nu-1\right) r_\req \leq z \leq \left(\nu-1 \right) r_\req \vspace{3mm} \\
\left\{r_\req^2-\left[z + \left(\nu -1\right) r_\req \right]^2 \right\}^{1/2} & : & \phantom{(-1)} -\nu r_\req \leq z <  -\left( \nu-1\right) r_\req \\
\end{array} \right. 
\end{align}
for aspect ratios $\nu \geq 1$ with the center of the hyperspherocylinder being in the origin of the coordinate system.
In the limit $\nu\to 1^+$, a hypersphere results. Opposite, in the limit $\nu\to\infty$, a hard needle in $\mathbb{R}^4$ results.  A two-dimensional section of this geometry is depicted in Fig. \ref{fig:hsc}. 
\begin{figure}[ht]
\centering
\includegraphics[width=0.2\textwidth]{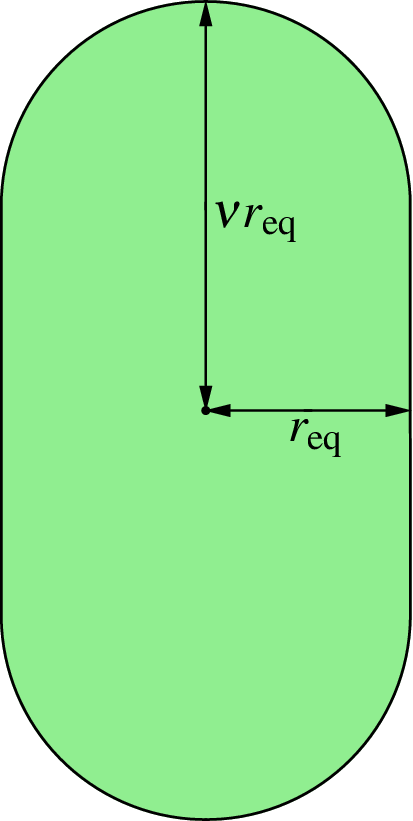}
\caption{Two-dimensional section of a hyperspherocylinder with radius $r_\req$ and aspect ratio $\nu=2$.}\label{fig:hsc}
\end{figure} 

With Eq. \eqref{eq:VP} and its meridian curve [Eq. \eqref{eq:Hsc_rz}], the volume of a hyperspherocylinder reads as
\begin{align}\label{eq:Hsc_VP}
V_\rP &= \left[\dfrac{1}{2} \pi^2 + \dfrac{8}{3} \left(\nu-1\right) \pi \right] r_\req^4\,.
\end{align}
The total surface area $S_\rP$ of a hyperspherocylinder, again, is the lateral surface area $S_\rP^\prime$ 
\begin{align}\label{eq:Hsc_SP}
S_\rP = S_\rP^\prime = \left[2 \pi^2 + 8 \left(\nu-1 \right) \pi \right] r_\req^3
\end{align}
[Eq. \eqref{eq:MP}]. 

A hyperspherocylinder's principal radii of curvature [Eqs. \eqref{eq:R1R2} and \eqref{eq:R3} using Eq. \eqref{eq:Hsc_rz}] are 
\begin{align}\label{eq:Hsc_R1R2}
R_1 = R_2 = r_\req
\end{align}
and 
\begin{align}\label{eq:Hsc_R3}
\vert R_3(z) \vert = \left\{\begin{array}{rcl} r_\req & : & \phantom{-}\left(\nu-1\right) r_\req < z \leq \nu r_\req \vspace{3mm}  \\
\infty & : & -\left( \nu-1\right) r_\req \leq z \leq \left(\nu-1 \right) r_\req \vspace{3mm} \\
r_\req & : & \phantom{(-1)}-\nu r_\req \leq z <  -\left( \nu-1\right) r_\req \\
\end{array} \right.
\end{align} 
for the cylindrical and hemispherical parts. 
Therefore, the second quermassintegral $W_2$ reads as 
\begin{align}\label{eq:Hsc_W2}
W_2 = \left[ \dfrac{1}{2} \pi^2 + \dfrac{4}{3} \left( \nu-1 \right) \pi  \right] r_\req^2
\end{align} 
using Eq. \eqref{eq:W2_an} and, using Eq. \eqref{eq:RP_an}, the mean radius of curvature as 
\begin{align}\label{eq:Hsc_RP}
\tilde{R}_\rP = \left[1+ \dfrac{4}{3 \pi} \left(\nu-1 \right) \right] r_\req\,.
\end{align}
Eqs. \eqref{eq:Hsc_VP}, \eqref{eq:Hsc_SP}, \eqref{eq:Hsc_W2}, and \eqref{eq:Hsc_RP} agree with the literature \cite{Kulossa2022} and, using Eqs. \eqref{eq:W}, 
\begin{align}
W_i= \left[\dfrac{1}{2} \pi^2 + \dfrac{8}{3} \left(\nu-1 \right) \dfrac{4-i}{4} \pi \right] r_\req^{4-i}
\end{align}
is obtained as general expression for the quermassintegrals of a four-dimensional hyperspherocylinder which is in accordance to the relation
\begin{align}\label{eq:Hsc_Wi_D}
W_i = \left[ \kappa_D + 2 \left(\nu-1 \right)\, \dfrac{D-i}{D}\, \kappa_{D-1} \right] r_\req^{D-i}
\end{align}
for $D$-dimensional spherocylinders as provided in the literature \cite{Santalo2004,Torquato2022}.


\subsection{Solids of revolution with singularities on their surface curvature} 

For convex solids of revolution with singularities of their surface curvature, the derivative of the meridian curve $\dot{r}(z)$ is discontinuous. 
In this case, in addition to the contributions from the piecewise integration over continuous parts  [Eqs. \eqref{eq:VP}, \eqref{eq:MP}, \eqref{eq:W2_an}, and \eqref{eq:RP_an}],
the singularities' contributions have to be considered. As in Sec. \ref{sec:measures_without_sing}, we describe selected geometries and calculate their 
geometric measures $V_\rP$, $S_\rP$, $W_2$, and $\tilde{R}_\rP$ in the following.

\subsubsection{Spherical hyperplates}

The spherical hyperplate as an infinitely thin hyperellipsoid is the limit of a hyperellipsoid with aspect ratio $\nu \to 0$ and radius $r_0$. 
It is a $(D-1)$-dimensional sphere with radius $r_0$ in $\mathbb{R}^D$. 
While in $\mathbb{R}^3$ a hard disk results, in the four-dimensional space a three-dimensional sphere  with a four-dimensional (hyper-)volume  
\begin{align}
V_\rP = 0
\end{align}
is obtained as analogue to the volumeless disk in $\mathbb{R}^3$. 
This can also be seen from Eq. \eqref{eq:Hell_VP} with a limiting aspect ratio of $\nu \to 0$. 
Despite for this geometry the $D$-dimensional volume $V_\rP$ vanishes, the further geometric measures are finite. 
Again, exemplarily the limiting case of a hyperellipsoid with $\nu \to 0$ can be used to obtain
\begin{align}\label{eq:Sph_plate_SP}
S_\rP = \dfrac{8}{3} \pi r_0^3 
\end{align}
for the total surface area of a four-dimensional spherical hyperplate. 
Note, that analogously to a disk in the three-dimensional space the hyperplate has a top and a bottom cell. 
Therefore, the total surface area $S_\rP$ of a $D$-dimensional spherical hyperplate is two times the volume of a $(D-1)$-dimensional sphere
and reads as 
\begin{align}
S_\rP = 2 \kappa_{D-1} r_0^{D-1}
\end{align} 
in Euclidean spaces $\mathbb{R}^D$. The second quermassintegral $W_2$ for a spherical hyperplate in $\mathbb{R}^4$ is
\begin{align}\label{eq:Sph_plate_W2}
W_2 = \dfrac{1}{3} \pi^2 r_0^2
\end{align}
[Eq. \eqref{eq:hell_W2_lim0}] and its mean radius of curvature
\begin{align}\label{eq:Sph_plate_RP}
\tilde{R}_\rP = \dfrac{8}{3 \pi} r_0
\end{align}
results from Eq. \eqref{eq:hell_RP_lim0}. 

Using Eqs. \eqref{eq:W}, the obtained results agree with the general formulation of quermassintegrals $W_i$ for $(D-1)$-dimensional spheres in $\mathbb{R}^D$
\begin{align}
W_i = \left\{\begin{array}{lcc}
0 & : & i=0 \vspace*{2mm} \\
\dfrac{\kappa_{D-1}}{\kappa_{i-1}} \dfrac{i}{D} \kappa_i r_0^{D-i} & : & i \geq 1\\
\end{array} \right.
\end{align}
as provided in the literature \cite{Torquato2022}.

\subsubsection{Hypercylinder}\label{sec:hcyl}

The meridian curve of a hypercylinder with height $h=2 \nu r_\req$ and radius $r_\req$ simply reads as 
\begin{align}\label{eq:Hc_rz}
r(z) = r_\req
\end{align} 
with vanishing derivatives $\dot{r}(z)=0$ and $\ddot{r}(z)=0$. Hence, only $R_1=R_2=r_\req$ are finite while $R_3$ is infinite with vanishing principal curvature $1/R_3$.
With $\nu=h/(2r_\req)$ as shown in Fig. \ref{fig:hc}, the volume of a four-dimensional hypercylinder is 
\begin{align}\label{eq:Hcyl_VP}
V_\rP 
&= \dfrac{8}{3} \pi  \nu r_\req^4
\end{align}
according to Eq. \eqref{eq:VP}. 
Its total surface area $S_\rP$ is the sum of the lateral surface area 
\begin{align}
S_\rP^\prime 
&= 8 \pi \nu r_\req^3
\end{align}
[Eq. \eqref{eq:MP}] and the two contributions $S_\rP^{\prime \prime}$ at $z=\pm \nu r_\req$ with
\begin{align}
S_\rP^{\prime \prime} = \dfrac{4}{3} \pi r_\req^3\,.
\end{align}
Hence, the total surface area of a four-dimensional hypercylinder reads as
\begin{align}\label{eq:Hcyl_SP}
S_\rP &= S_\rP^\prime + 2 S_\rP^{\prime \prime} = \dfrac{8}{3} \pi \left(3\nu + 1 \right) r_\req^3\,.
\end{align}
\begin{figure}[ht]
\centering
\includegraphics[width=0.2\textwidth]{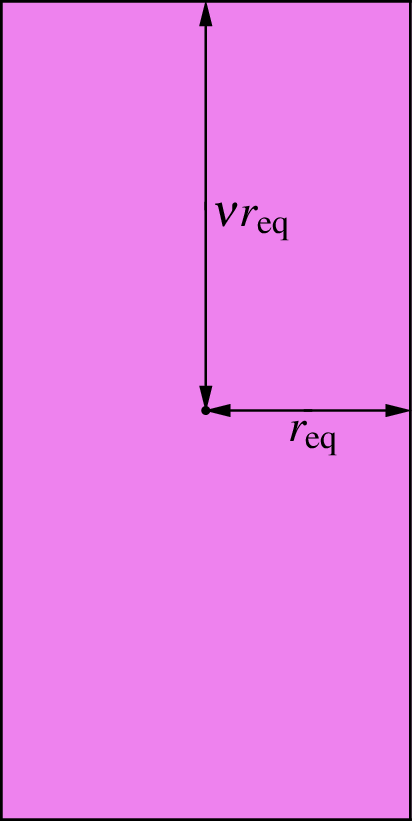}
\caption{Two-dimensional section of a hypercylinder with radius $r_\req$ and aspect ratio $\nu=2$.}\label{fig:hc}
\end{figure}

The contribution of the continuous part $W_2^\prime$ to the second quermassintegral $W_2$  is 
\begin{align}
W_2^\prime 
&= \dfrac{4}{3} \pi \nu r_\req^2
\end{align}
[Eq. \eqref{eq:W2_an}]. 
For both, top and bottom contribution of the singularity as for the spherical hyperplate 
\begin{align}
W_2^{\prime \prime} = \dfrac{1}{6} \pi^2 r_\req^2
\end{align}
results [Eq. \eqref{eq:Sph_plate_W2}]. Hence,
\begin{align}\label{eq:Hcyl_W2}
W_2 &= W_2^\prime + 2 W_2^{\prime \prime} = \dfrac{1}{3} \pi \left(4 \nu + \pi \right) r_\req^2
\end{align}
is obtained for the second quermassintegral $W_2$ of a four-dimensional hypercylinder. 

The contribution of the continuous cylindrical part $\tilde{R}_\rP^\prime$ to the mean radius of curvature $\tilde{R}_\rP$ is 
\begin{align}
\tilde{R}_\rP^\prime 
&= \dfrac{4}{3 \pi} \nu r_\req 
\end{align}
[Eq. \eqref{eq:RP_an}] and each contribution of the singularities at $z=\pm \nu r_\req$ reads as 
\begin{align}
\tilde{R}_\rP^{\prime \prime} = \dfrac{4}{3 \pi} r_\req
\end{align}
identical to the top and bottom contribution of the spherical hyperplate [Eq. \eqref{eq:Sph_plate_RP}]. In total,
\begin{align}\label{eq:Hcyl_RP}
\tilde{R}_\rP &= \tilde{R}_\rP^\prime + 2 \tilde{R}_\rP^{\prime \prime} = \dfrac{4}{3 \pi} \left(\nu +2 \right) r_\req
\end{align}
results for the mean radius of curvature $\tilde{R}_\rP$ of a four-dimensional hypercylinder.
Using Eqs. \eqref{eq:W}, the general expression for the quermassintegrals $W_i$ of this geometry 
\begin{align}
W_i = \left\{\begin{array}{lcc}
\dfrac{8}{3} \pi \nu r_\req^4 & : & i=0 \vspace*{2mm} \\
\dfrac{2}{3} \pi \left[\left(4-i \right) \nu + \dfrac{\beta_{i}}{2 \kappa_{i-1}}  \right] r_\req^{4-i} & : & i \geq 1\\
\end{array} \right.\,,
\end{align}
results in accordance to the general result for $D$-dimensional cylinders 
\begin{align}\label{eq:Hcyl_Wi_D}
W_i = \left\{\begin{array}{lcc}
2 \nu \kappa_{D-1}  r_\req^D & : & i=0 \vspace*{2mm} \\
\dfrac{2}{D} \kappa_{D-1} \left[\left(D-i \right) \nu + \dfrac{\beta_{i}}{2 \kappa_{i-1}}  \right] r_\req^{D-i} & : & i \geq 1\\
\end{array} \right.
\end{align}
provided in the literature \cite{Santalo2004,Torquato2022}. 


\subsubsection{Hyperspindle}

Analogously to a three-dimensional spindle \cite{Herold2017}, a hyperspindle with center in the origin of the coordinate system can be described by the meridian curve 
\begin{align}\label{eq:Hspindle_rz}
r(z) = \left(\dfrac{1- \nu^2}{2} \right) r_\req + \dfrac{1}{2} \left[\left(\nu^2 +1 \right)^2 r_\req^2 - 4z^2 \right]^{1/2}
\end{align}
as a prolate, four-dimensional, uniaxial solid of revolution with radius $r_\req$ and aspect ratio $\nu \geq 1$. 

\begin{figure}[ht]
\centering
\includegraphics[width=0.2\textwidth]{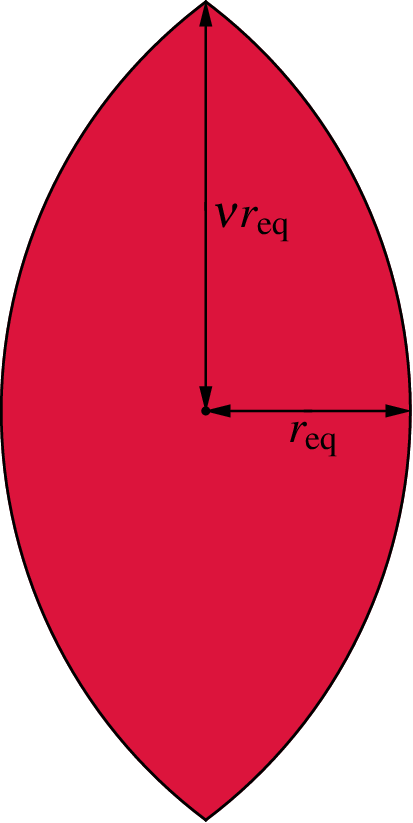}
\caption{Two-dimensional section of a hyperspindle with radius $r_\req$ and aspect ratio $\nu=2$.}\label{fig:hspindle}
\end{figure}

As shown by the two-dimensional section of this geometry in Fig. \ref{fig:hspindle}, a spindle has two singularities of curvature, located at $z=\pm \nu r_\req$ as an upper and lower apex. 
In the limit $\nu \to 1$, the hyperspindle approaches a hypersphere and the contributions of the singularities on the hypersurface vanish. 
The volume of a four-dimensional hyperspindle using Eq. \eqref{eq:VP} reads as 
\begin{align}
V_\rP 
&= \dfrac{5}{16} \pi r_\req^4 \left[\left(\nu^4 - \dfrac{6}{5}\nu^2+1 \right) \left(\nu^2+1 \right)^2 \arcsin \left(\dfrac{2\nu}{\nu^2 +1} \right) - \right. \nonumber \\
& \qquad \left. 2 \nu^7 - \dfrac{14}{15} \nu^5 + \dfrac{14}{15} \nu^3 + 2 \nu \right] 
\end{align}
in dependence on the radius $r_\req$ and the aspect ratio $\nu$. 
The total surface area $S_\rP$ can be written as 
\begin{align}
S_\rP = \dfrac{1}{2} \pi r_\req^3 \left[\left(3 \nu^6 + \nu^4 + \nu^2 +3 \right) \arcsin \left(\dfrac{2 \nu}{\nu^2 +1} \right) -6\nu^5 +6\nu \right]  
\end{align}
using Eq. \eqref{eq:MP} with Eq. \eqref{eq:Hspindle_rz} since at the apices $r(z_\rmin) = r(z_\rmax) =0$, the zero-dimensional singularities do not contribute to $W_1$ and $S_\rP$. 

To calculate the contribution of the apical singularities to $W_2$ and $\tilde{R}_\rP$, in the interval $\nu r_\req - \epsilon \leq \vert z \vert \leq \nu r_\req$ the meridian curve $r(z)$ can continuously be replaced by the meridian curve $r^\sharp (z)$ of a hypersphere 
\begin{align}
r^\sharp (z) = \left[r_0^2 - \left(z- \xi \right)^2 \right]^{1/2}
\end{align}
with 
\begin{align}
r \left( \nu r_\req - \epsilon \right) = r^\sharp \left( \nu r_\req - \epsilon \right)
\end{align}
and 
\begin{align}
\dot{r} \left( \nu r_\req - \epsilon \right) = \dot{r}^\sharp \left( \nu r_\req - \epsilon \right)
\end{align}
as visualized in Fig. \ref{fig:spindle_approx}. 
Here 
\begin{align}
\xi \left(r_\req, \nu, \epsilon \right) = r_\req \left(\nu^2-1 \right) \left(\nu r_\req - \epsilon \right) \left[\left(\nu^2 +1 \right)^2 r_\req^2 - 4 \left(\nu r_\req - \epsilon \right)^2 \right]^{-1/2}
\end{align} 
is the $z$-position of the hypersphere's center  and 
\begin{align}
r_0 \left(r_\req, \nu, \epsilon \right) = \dfrac{\nu^2+1}{2} r_\req \left\{1+ \dfrac{\left(1-\nu^2 \right) r_\req}{ \left[\left(\nu^2+1 \right)^2 r_\req^2 - 4 \left(\nu r_\req - \epsilon \right)^2 \right]^{1/2}} \right\}
\end{align}
its radius. 

\begin{figure}[ht]
\centering
\includegraphics[width=0.5\textwidth, valign=c]{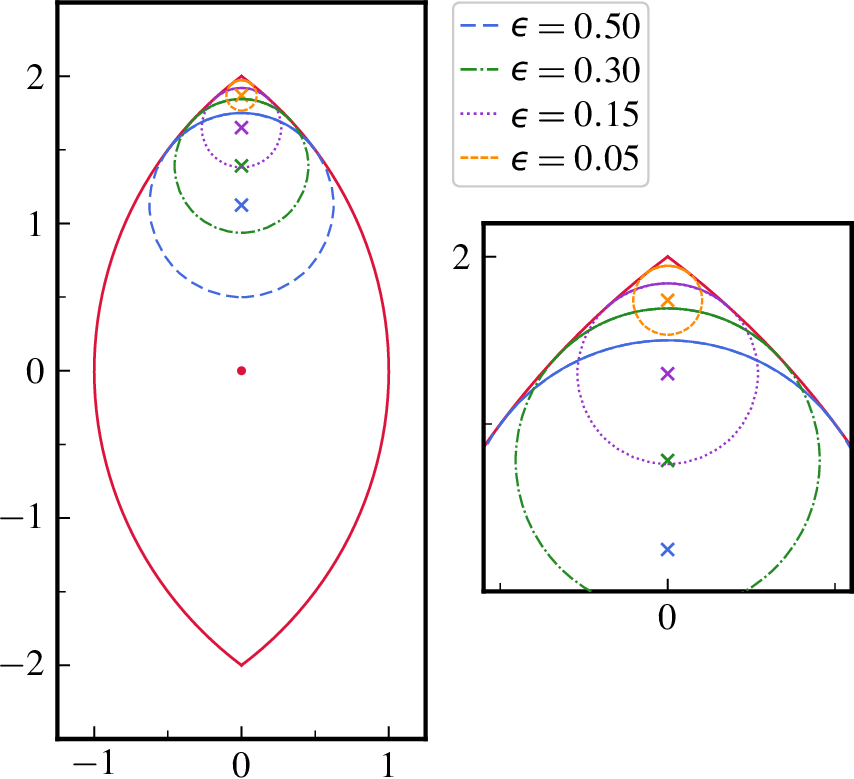} \qquad \quad 
\includegraphics[width=0.35\textwidth, valign=c]{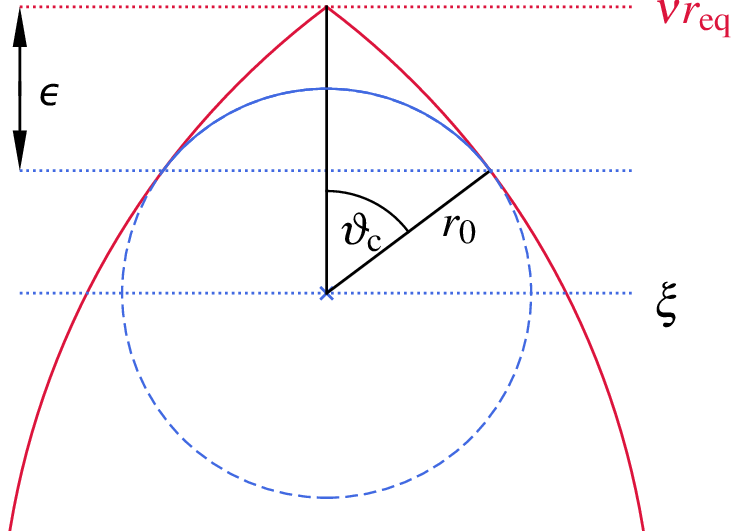}
\caption{Continuous completion of a hyperspindle with radius $r_\req=1$ and aspect ratio $\nu=2$ by hypersphere segments with radius $r_0$ and the vicinity to the apex $\epsilon$. The crosses indicate the centers of the hyperspheres with $z=\xi$.}\label{fig:spindle_approx}
\end{figure}

In the region continued by a hypersphere, the principal radii of curvature are 
\begin{align}
R_1=R_2=R_3=r_0
\end{align}
[Eqs. \eqref{eq:R1R2} and \eqref{eq:R3}]. 
Using Eq. \eqref{eq:SP_R1R2R3}, the second quermassintegral and the mean radius of curvature can, in the vicinity of the apices, be written as 
\begin{align}
W_2^{\prime \prime} (\epsilon) = \dfrac{1}{4} \int\limits_0^{\vartheta_\rc (\epsilon)} \int\limits_0^{\pi} \int\limits_0^{2 \pi} r_0^2 (\epsilon) \sin^2 \vartheta \sin \chi \,\rd \varphi \,\rd \chi \,\rd \vartheta
\end{align}
and
\begin{align}
\tilde{R}_\rP^{\prime \prime} (\epsilon) = \dfrac{1}{2 \pi^2} \int\limits_0^{\vartheta_\rc (\epsilon)} \int\limits_0^{\pi} \int\limits_0^{2 \pi} r_0 (\epsilon) \sin^2 \vartheta \sin \chi \,\rd \varphi \,\rd \chi \,\rd \vartheta
\end{align}
due to the continuous surface curvature of the hypersphere, using Eqs. \eqref{eq:Wi_K}. 
The critical angle $\vartheta_\rc (\epsilon)$ is 
\begin{align}
\vartheta_\rc (\epsilon) = \arccos \left[\dfrac{2 (\nu r_\req - \epsilon)}{(\nu^2+1)r_\req} \right]
\end{align}
(see Fig. \ref{fig:spindle_approx}) and performing the integration over $\varphi$ and $\chi$
\begin{align}
W_2^{\prime \prime} (\epsilon) = \pi \int\limits_0^{\vartheta_\rc (\epsilon)} r_0^2 (\epsilon) \sin^2 \vartheta \,\rd \vartheta
\end{align}
and 
\begin{align}
\tilde{R}_\rP^{\prime \prime} (\epsilon) = \dfrac{2}{\pi} \int\limits_0^{\vartheta_\rc (\epsilon)} r_0 (\epsilon) \sin^2 \vartheta \,\rd \vartheta
\end{align}
result. Hence, in the limit $\epsilon \to 0$  the contributions of the singularities to the second quermassintegral $W_2$
\begin{align}\label{eq:W2_apex}
W_2^{\prime \prime} = \lim\limits_{\epsilon \to 0} W_2^{\prime \prime} (\epsilon) = 0
\end{align}
and the mean radius of curvature $\tilde{R}_\rP$
\begin{align}\label{eq:RP_apex}
\tilde{R}_\rP^{\prime \prime} = \lim\limits_{\epsilon \to 0} \tilde{R}_\rP^{\prime \prime} (\epsilon) = 0
\end{align}
vanish. Thus, the second quermassintegral $W_2$ can be expressed as  
\begin{align}
W_2 &= W_2^\prime + 2 W_2^{\prime \prime}
	 = \dfrac{5}{12} \pi r_\req^2 \left[\left(\nu^4 + \dfrac{2}{5} \nu^2 +1 \right) \arcsin \left(\dfrac{2 \nu}{\nu^2 +1} \right) - 2 \nu^3 + 2 \nu \right] 
\end{align}
using Eqs. \eqref{eq:W2_an} and \eqref{eq:W2_apex} with the principal radii of curvature 
\begin{align}
R_1(z) = R_2(z) = \dfrac{\nu^2 +1}{2} r_\req + \dfrac{\left(1-\nu^2 \right) \left(\nu^2 +1 \right) r_\req^2}{2 \left[\left(\nu^2+1 \right)^2 r_\req^2 - 4 z^2 \right]^{1/2}} 
\end{align}
and 
\begin{align} 
R_3 = \dfrac{\nu^2 +1}{2} r_\req 
\end{align}
[Eqs. \eqref{eq:R1R2} and \eqref{eq:R3}]. For the mean radius of curvature, finally 
\begin{align}
\tilde{R}_\rP 	&= \tilde{R}_\rP^\prime + 2 \tilde{R}_\rP^{\prime \prime}
				= \dfrac{r_\req}{\pi} \left[\left(\nu^2 +1 \right) \arcsin\left(\dfrac{2 \nu}{\nu^2 +1}\right) - \dfrac{2}{3} \nu \,\dfrac{\nu^2-1}{\nu^2+1}  \right]
\end{align}
is obtained [Eqs. \eqref{eq:RP_an} and \eqref{eq:RP_apex}].

\subsubsection{Hyperlens}

A hyperlens is the section of two hyperspheres with identical radius $r_0$, separated by the distance $L$ with $0\le L < 2r_0$. 
It is an oblate solid of revolution with equatorial radius $r_\req$ and aspect ratio $\nu$ with $\nu \in [0,1]$.  
Using equatorial radius $r_\req$ and aspect ratio $\nu$, the radius of the intersecting hyperspheres can be written as 
\begin{align}
r_0 = \dfrac{\nu^2+1}{2\nu} r_\req
\end{align}
and for the distance of separation,  
\begin{align}
\dfrac{L}{2} = r_0 - \nu r_\req
\end{align}
results. 
Also, since a hyperlens consists of two equal hyperspherical caps, a critical angle $\theta_\rc$ with
\begin{align}\label{eq:hlens_thetac}
\cos\left(\dfrac{\theta_\rc}{2} \right) = \dfrac{1-\nu^2}{1+\nu^2}  
\end{align}
can be defined, identical to a lens in $\mathbb{R}^3$ \cite{Marienhagen2021} and a planar lens in $\mathbb{R}^2$ \cite{Kulossa2023b}. 
In the limit $\nu \to 1$, a hypersphere with radius $r_0=r_\req$ results while for the limit $\nu \to 0$, a spherical hyperplate with radius $r_\req$ is obtained. 
In Fig. \ref{fig:hlens}, the two-dimensional section of a hyperlens is shown with its characteristic parameters. 

\begin{figure}[ht]
\centering
\includegraphics[width=0.5\textwidth]{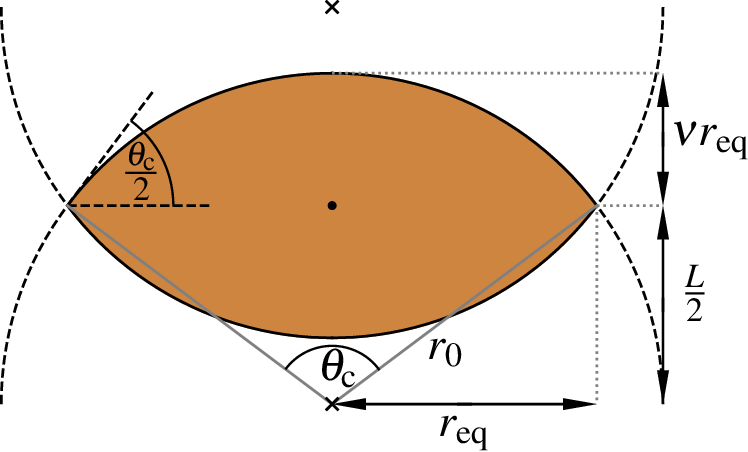}
\caption{Two-dimensional section of a hyperlens with radius $r_\req$ and aspect ratio $\nu=1/2$. The crosses indicate the centers of the generating hyperspheres with radius $r_0$ separated by the distance $L$. $\theta_\rc$ is the critical angle.}\label{fig:hlens}
\end{figure}

Placing the hyperlens' center in the origin of the coordinate system, the meridian curve $r(z)$ reads as
\begin{align}\label{eq:Hlens_rz}
r(z) = \left[\dfrac{\vert z \vert r_\req \left(\nu^2-1\right)+\nu \left(r_\req^2-z^2\right)}{\nu} \right]^{1/2}
\end{align}
where the hyperspheres' centers are located at $z=\pm (r_0 - \nu r_\req)$. 
Since a hyperlens consists of two attached hyperspherical caps, the principal radii of curvature are with
\begin{align}
R_1 = R_2 = R_3  = \dfrac{\nu^2 +1}{2 \nu} r_\req = r_0
\end{align}
equal to the radius of the hyperlens' generating hyperspheres [Eqs. \eqref{eq:R1R2} and \eqref{eq:R3}]. 

For the volume $V_\rP$ of a hyperlens, 
\begin{align}\label{eq:Hlens_VP}
V_\rP &= \dfrac{\pi r_\req^4}{32 \nu^4} \left\{\left(1+\nu^2 \right)^4 \left[\pi -2 \arcsin \left(\dfrac{1-\nu^2}{1+\nu^2} \right) \right]-4\nu\left[1-\nu^6 + \dfrac{11}{3} \nu^2 \left(1-\nu^2 \right) \right] \right\}
\end{align}
is obtained [Eqs. \eqref{eq:VP} and \eqref{eq:Hlens_rz}]. 
Its total surface area $S_\rP$ equals the lateral surface area $S_\rP^\prime$ and can be written as  
\begin{align}\label{eq:Hlens_SP}
S_\rP &= \dfrac{\pi \left(1+\nu^2 \right) r_\req^3}{4 \nu^3} \left\{\left(1+\nu^2 \right)^2 \left[\pi - 2 \arcsin \left(\dfrac{1-\nu^2}{1+\nu^2} \right) \right] - 4 \nu \left(1-\nu^2 \right) \right\}
\end{align}
using Eq. \eqref{eq:MP}. 
For the contribution of the continuous parts of the hyperlens to its second quermassintegral and mean radius of curvature,
\begin{align}
W_2^\prime &= \dfrac{\pi r_\req^2}{8 \nu^2} \left[\left(1+\nu^2 \right)^2 \pi - 2 \left(\nu^4+2\nu^2+1 \right) \arcsin\left(\dfrac{1-\nu^2}{1+\nu^2} \right) - 4 \nu \left(1-\nu^2 \right) \right]
\end{align}
and 
\begin{align}
\tilde{R}_\rP^\prime &=\dfrac{r_\req}{2 \pi \nu} \left\{\left(1+\nu^2 \right) \left[\pi - 2 \arcsin \left(\dfrac{1 - \nu^2}{1+\nu^2} \right) \right]-4 \nu \dfrac{1-\nu^2}{1+\nu^2} \right\}
\end{align} 
are obtained [Eqs. \eqref{eq:W2_an} and \eqref{eq:RP_an} with Eq. \eqref{eq:Hlens_rz}]. 

For the contributions of the equatorial singularity $W_2^{\prime \prime}$ and $\tilde{R}_\rP^{\prime \prime}$, the surface of a cut-out hyperlens in $-\epsilon \leq z \leq \epsilon$ can continuously be replaced by a toroidal hypercylinder with height $2 \epsilon$ and meridian curve 
\begin{align}\label{eq:rz_torus}
r^\sharp (z) = \xi + \left(r_\rt^2 -z^2 \right)^{1/2}
\end{align}
where 
\begin{align}
r^\sharp (\epsilon) = r(\epsilon)
\end{align}
and 
\begin{align}
\dot{r}^\sharp (\epsilon) = \dot{r} (\epsilon) 
\end{align}
are fulfilled (see Fig. \ref{fig:torus}). 
\begin{figure}[ht]
\centering
\includegraphics[width=0.45\textwidth,valign=c]{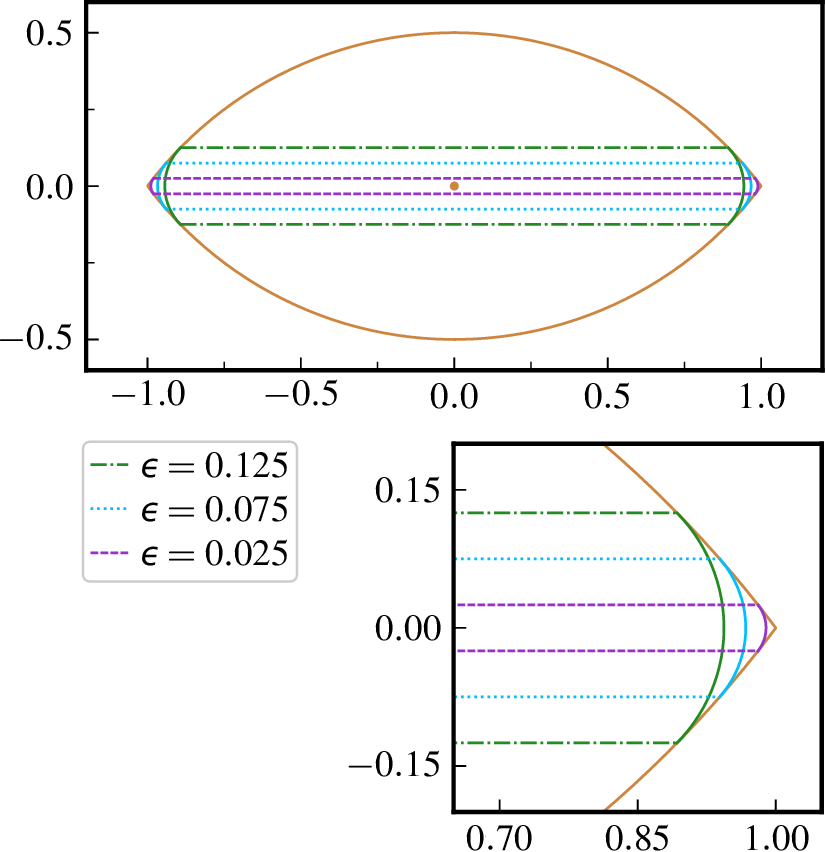} \qquad \qquad 
\includegraphics[width=0.25\textwidth,valign=c]{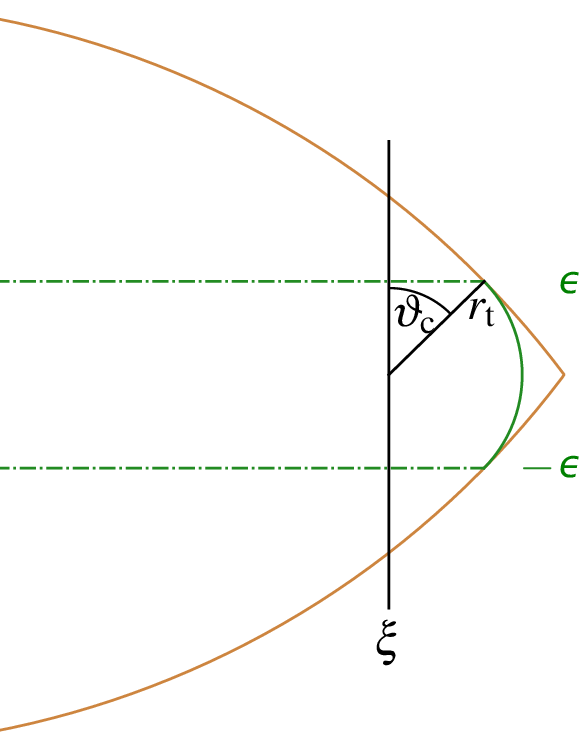}
\caption{Continuous completion of a hyperlens with radius $r_\req=1$ and aspect ratio $\nu=1/2$ by a hypertorus of height $2 \epsilon$ with meridian radius $r_\rt$ and distance $\xi$ between the center of the torus and the center of the meridian arc.}\label{fig:torus}
\end{figure}
In Eq. \eqref{eq:rz_torus} 
\begin{align}
r_\rt \left(\epsilon \right) = \dfrac{\epsilon\, r_\req \left(1+\nu^2 \right)}{2 \nu \epsilon + r_\req \left(1-\nu^2 \right)}
\end{align}
is the meridian radius of the hypertorus and 
\begin{align}
\xi\left(\epsilon \right) = -2 \epsilon \,\dfrac{\left[\nu \left(\nu r_\req - \epsilon \right) \left( \nu \epsilon + r_\req \right) \right]^{1/2}}{2 \nu \epsilon + \left(1-\nu^2 \right) r_\req} + \left[\dfrac{\left(\nu r_\req - \epsilon \right) \left( \nu \epsilon + r_\req \right)}{\nu} \right]^{1/2} 
\end{align}
the distance between the center of the hypertorus and the center of the meridian arc. 
In the region of the hypertorus, the principal radii of curvature can be written as 
\begin{subequations}
\begin{align}
R_1(z) &= R_2(z) = r_\rt \left[1+ \dfrac{\xi}{(r_\rt^2 - z^2)^{1/2}} \right] \,, \\
R_3(z) &= r_\rt
\end{align}
\end{subequations}
using Eqs. \eqref{eq:R1R2} and \eqref{eq:R3}. 
Employing $z=r_\rt \cos \vartheta$ in $\vartheta_\rc \leq \vartheta \leq \pi -\vartheta_\rc$ (see Fig. \ref{fig:torus}) with $\cos \vartheta_\rc (\epsilon) = \epsilon / r_\rt (\epsilon)$, the principal radii of curvature can be written as 
\begin{subequations}
\begin{align}
R_1(\vartheta) &= R_2(\vartheta) = r_\rt(\epsilon) + \dfrac{\xi(\epsilon)}{\sin \vartheta} \,, \\
R_3(\vartheta) &= r_\rt(\epsilon)
\end{align}
\end{subequations}
in dependence on the angular coordinate $\vartheta$. 
With Eq. \eqref{eq:SP_R1R2R3}, the second quermassintegral and the mean radius of curvature can, in the vicinity of the equatorial singularity, be written as 
\begin{align}
W_2^{\prime \prime} (\epsilon) &= \dfrac{\pi}{3} \int\limits_{\vartheta_\rc(\epsilon)}^{\pi-\vartheta_\rc(\epsilon)} \left[R_1(\vartheta) R_2(\vartheta) + R_1(\vartheta)R_3(\vartheta) + R_2(\vartheta)R_3(\vartheta) \right] \sin^2 \vartheta \,\rd \vartheta \nonumber \\
&= \dfrac{\pi}{3} \int\limits_{\vartheta_\rc(\epsilon)}^{\pi-\vartheta_\rc(\epsilon)} \left[3 r_\rt^2(\epsilon) \sin^2 \vartheta + 4 r_\rt(\epsilon) \xi(\epsilon) \sin \vartheta + \xi^2(\epsilon)\right] \,\rd \vartheta
\end{align}
and
\begin{align}
\tilde{R}_\rP^{\prime \prime} (\epsilon) &= \dfrac{2}{3\pi} \int\limits_{\vartheta_\rc(\epsilon)}^{\pi-\vartheta_\rc(\epsilon)} \left[R_1(\vartheta) + R_2(\vartheta) + R_3(\vartheta) \right] \sin^2 \vartheta \,\rd \vartheta \nonumber \\
&= \dfrac{2}{\pi} \int\limits_{\vartheta_\rc(\epsilon)}^{\pi-\vartheta_\rc(\epsilon)} \left[r_\rt(\epsilon) \sin^2 \vartheta + \dfrac{2}{3} \xi(\epsilon) \sin \vartheta \right]  \,\rd \vartheta 
\end{align}
due to the continuous surface curvature of the hypertorus [Eqs. \eqref{eq:Wi_K}]. 
The contributions of the equatorial singularity 
\begin{align}
W_2^{\prime \prime} = \lim\limits_{\epsilon \to 0} W_2^{\prime \prime} (\epsilon) = \dfrac{2}{3} \pi r_\req^2 \arcsin \left(\dfrac{1-\nu^2}{1+\nu^2} \right)
\end{align}
and 
\begin{align}
\tilde{R}_\rP^{\prime \prime}  &= \lim\limits_{\epsilon \to 0} \tilde{R}_\rP^{\prime \prime} (\epsilon)  = \dfrac{8}{3\pi} r_\req \, \dfrac{1-\nu^2}{1+\nu^2}
\end{align}
result in the limit $\epsilon \to 0$. 
Hence, for the second quermassintegral and the mean radius of curvature, the contributions 
\begin{align}\label{eq:Hlens_W2''}
W_2^{\prime \prime} = \dfrac{1}{3} \pi^2 r_\req^2 \, \dfrac{2}{\pi} \arcsin\left[ \cos \left(\dfrac{\theta_\rc}{2} \right) \right]
\end{align}
and 
\begin{align}\label{eq:Hlens_RP''}
\tilde{R}_\rP^{\prime \prime} = \dfrac{8}{3\pi} r_\req \, \cos\left(\dfrac{\theta_\rc}{2} \right)
\end{align}
result depending on the critical angle $\theta_\rc$ [Eq. \eqref{eq:hlens_thetac}]. In the limit of a spherical hyperplate with $\theta_\rc=0$, Eqs. \eqref{eq:Sph_plate_W2} and \eqref{eq:Sph_plate_RP}
are consistently obtained. Finally, 
\begin{align}\label{eq:Hlens_W2}
W_2 = W_2^\prime + W_2^{\prime \prime} &=\dfrac{\pi r_\req^2}{8 \nu^2}\left[\left(1+\nu^2\right)^2 \pi -2 \left(\nu^4-\dfrac{2}{3}\nu^2 +1 \right) \arcsin\left(\dfrac{1-\nu^2}{1+\nu^2} \right) -  \right. \nonumber \\ 
&\quad \left. 4 \nu \left(1-\nu^2 \right) \right] 
\end{align}
and 
\begin{align}\label{eq:Hlens_RP}
\tilde{R}_\rP &= \tilde{R}_\rP^\prime + \tilde{R}_\rP^{\prime\prime} = \dfrac{r_\req}{2 \pi \nu} \left\{\left(1+\nu^2 \right) \left[\pi - 2 \arcsin \left(\dfrac{1 - \nu^2}{1+\nu^2} \right) \right]+\dfrac{4}{3} \nu \dfrac{1-\nu^2}{1+\nu^2} \right\}
\end{align}
are obtained for the second quermassintgral $W_2$ and mean radius of curvature $\tilde{R}_\rP$ of a hyperlens.

Additionally, in the limit $\nu \to 0$, the geometric measures $S_\rP$, $W_2$, and $\tilde{R}_\rP$ with Eqs. \eqref{eq:Hlens_SP}, \eqref{eq:Hlens_W2}, and \eqref{eq:Hlens_RP} can, for numerical stability, be expanded in Maclaurin series as
\begin{align}
S_\rP &= \dfrac{8}{3} \pi r_\req^3 \left[1+ \dfrac{6}{5} \nu^2 + \mathcal{O}\left(\nu^4\right) \right] \,,
\end{align}
\begin{align}
W_2 &= \dfrac{1}{3} \pi r_\req^2 \left[\pi + \dfrac{32}{15} \nu^3 + \mathcal{O}\left(\nu^5\right) \right] \,,
\end{align}
and 
\begin{align}
\tilde{R}_\rP &= \dfrac{8}{3 \pi} r_\req \left[1+\dfrac{2}{5} \nu^4 + \mathcal{O}\left(\nu^6\right)\right]
\end{align}
to obtain the limits of a spherical hyperplate with Eqs. \eqref{eq:Sph_plate_SP}, \eqref{eq:Sph_plate_W2}, and \eqref{eq:Sph_plate_RP}. 
For the volume $V_\rP$ [Eq. \eqref{eq:Hlens_VP}], the expansion 
\begin{align}
V_\rP &= 16 \pi r_\req^4 \left[\dfrac{1}{15} \nu + \dfrac{1}{35} \nu^3 + \mathcal{O}\left(\nu^5\right)\right]
\end{align}
is obtained.

\subsubsection{Hyperspherical cap}\label{sec:hcap}

A hyperspherical cap is a hypersphere cut off by a single hyperplane resulting in half of a spherical hyperplate as a singularity on the surface of this geometry. 
The hyperspherical cap can be defined by the radius of the hypersphere $r_0$ and the critical angle $\theta_\rc$ (Fig. \ref{fig:hcap}). 
With these parameters, the height  
\begin{align}\label{eq:hcap_h}
    h = r_0 \left[1-\cos \left(\dfrac{\theta_\rc}{2}\right) \right]
\end{align}
and radius of the cap 
\begin{align}\label{eq:hcap_req}
    r_{\rm cap} = r_0 \sin \left(\dfrac{\theta_\rc}{2} \right) = \left[h \left( 2 r_0 -h \right) \right]^{1/2}
\end{align} 
are obtained, resulting in the aspect ratio 
\begin{align}\label{eq:hcap_nu}
\nu = \left\{\begin{array}{ccc}
\dfrac{h}{2 r_\rcap} & : & 0 \leq \theta_\rc \leq \pi \\
\\
\dfrac{h}{2 r_0} 	& : & \pi \leq \theta_\rc \leq 2 \pi
\end{array} \right.
\end{align} 
of hyperspherical caps. In the limit $\nu \to 1$, a hypersphere with radius $r_0$ results while for $\nu \to 0$, a spherical hyperplate with radius $r_\rcap$ is obtained. 
Placing the hypersphere's center in the origin of the coordinate system, the meridian curve $r(z)$ of the hyperspherical cap 
\begin{align*}
r(z) = \left(r_0^2 -z^2 \right)^{1/2}
\end{align*}
is identical to the meridian curve of the complete hypersphere [Eq. \eqref{eq:Hsp_rz}].

\begin{figure}[ht]
\centering
\includegraphics[width=0.35\textwidth,valign=t]{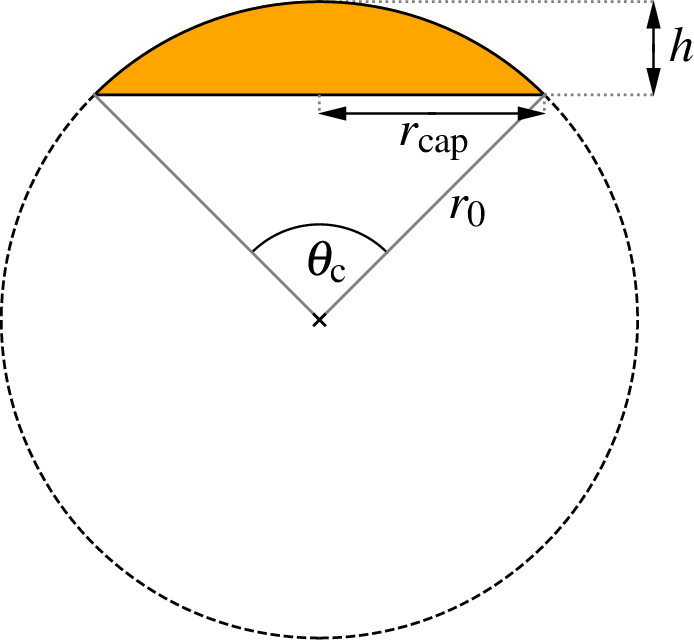} \qquad \qquad 
\includegraphics[width=0.35\textwidth,valign=t]{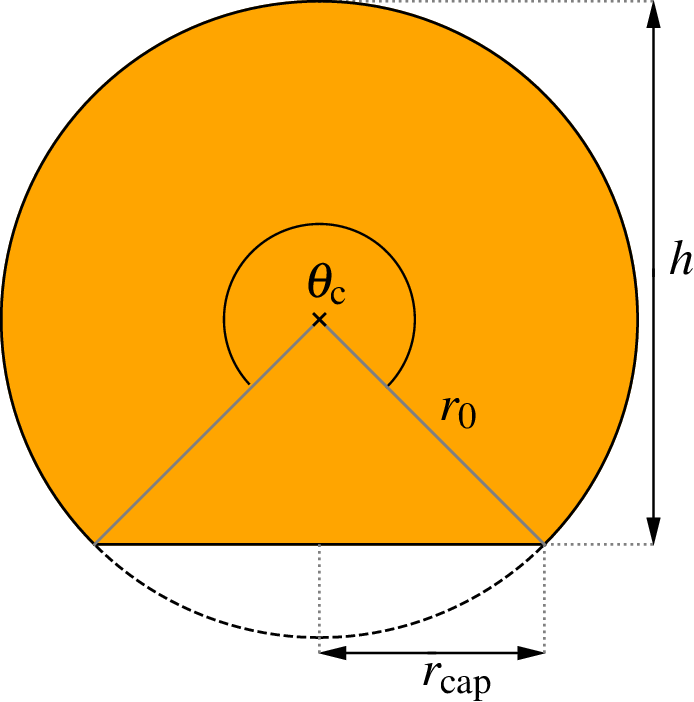}
\caption{Two-dimensional section of a hyperspherical cap with height $h$ and radius $r_\rcap$ . For a generating hypersphere with radius $r_0$, hyperspherical caps with critical angles $\theta_\rc = \pi/2$ (lhs) and $\theta_\rc= 3 \pi/2$ (rhs) are depicted.}\label{fig:hcap}
\end{figure} 

For the volume of a hyperspherical cap Eq. \eqref{eq:VP} can be used analogously to Eq. \eqref{eq:hsp_V} as 
\begin{align}\label{eq:hcap_VP}
V_\rP &= \dfrac{4}{3} \pi\int\limits_{r_0-h}^{r_0}  \left(r_0^2 -z^2 \right)^{3/2} \rd z \nonumber \\
&= \dfrac{\pi^2}{4} r_0^4 \left[1- \dfrac{2}{\pi} \arctan\left(\dfrac{r_0 -h}{r_\rcap} \right) - \dfrac{2 r_\rcap}{3 \pi r_0^4} \left(r_0-h \right) \left(3r_0^2 +2 r_\rcap^2 \right) \right]
\end{align}
by replacing the limits of the integral accordingly.   
This agrees with 
\begin{align}\label{eq:hcap_I_V}
V_\rP = \dfrac{\kappa_D}{2} r_0^D \mathcal{I}_{r_{\rm cap}^2/r_0^2}\left(\dfrac{D+1}{2},\dfrac{1}{2} \right) \qquad : \quad 0 \leq h \leq r_0
\end{align}
as the expression for the volume of $D$-dimensional spherical caps in $\mathbb{R}^D$ with $\theta_\rc \leq \pi$ provided in the literature \cite{Li2011}. 
Here $\mathcal{I}_x(a,b)$ with
\begin{align}
\mathcal{I}_x(a,b) &= \dfrac{\Gamma(a+b)}{\Gamma(a) \Gamma(b)} \int\limits_0^x t^{a-1} (1-t)^{b-1} ~\rd t  \nonumber \\
&= \dfrac{\Gamma(a+b)}{\Gamma(a) \Gamma(b)} \dfrac{x^a}{a} \,_2\mathcal{F}_1 \left(a,1-b,a+1;x \right) 
\end{align}
is the regularized incomplete beta function \cite{NIST:DLMF}. 
Also Eq. \eqref{eq:hcap_I_V} can easily be adapted to 
\begin{align}
V_\rP = \dfrac{\kappa_D}{2} r_0^D \left[ 2-  \mathcal{I}_{r_\rcap^2/r_0^2}\left(\dfrac{D+1}{2},\dfrac{1}{2} \right)\right] \qquad : \quad r_0 \leq h \leq 2r_0
\end{align}
for hyperspherical caps with critical angles $\theta_\rc \geq \pi$ while, in $\mathbb{R}^4$, Eq. \eqref{eq:hcap_VP} can be used for the volume of hyperspherical caps with arbitrary critical angles $\theta_\rc$.

Its total surface area $S_\rP$ can be written as the sum
\begin{align}
S_\rP = S_ \rP^\prime + S_\rP^{\prime \prime}
\end{align}
of the lateral surface area $S_\rP^\prime$ and the surface area at the singularity $S_\rP^{\prime \prime}$. 
For the lateral surface area $S_\rP^\prime$, Eqs. \eqref{eq:MP} and \eqref{eq:hcap_req} can be used to obtain
\begin{align}\label{eq:cap_S1}
S_\rP^\prime 
&= \pi^2 r_0^3 \left[1-\dfrac{2}{\pi} \arctan \left(\dfrac{r_0-h}{r_\rcap} \right) - \dfrac{2 r_\rcap}{\pi r_0^2} \left(r_0 -h \right) \right]
\end{align} 
while the surface area at the singularity 
\begin{align}
S_\rP^{\prime \prime} = \dfrac{4}{3} \pi r_\rcap^3
\end{align}
is half the surface area of a two-sided spherical hyperplate. 
Hence, the surface area of hyperspherical caps reads as
\begin{align}\label{eq:hcap_SP}
S_\rP = \pi^2 r_0^3 \left[1-\dfrac{2}{\pi} \arctan \left(\dfrac{r_0-h}{r_\rcap} \right) - \dfrac{2 r_\rcap}{\pi r_0^2} \left(r_0 -h \right) \right] + \dfrac{4}{3} \pi r_\rcap^3
\end{align}
in $\mathbb{R}^4$.

For the lateral surface area $S_\rP^\prime$ of $D$-dimensional spherical caps with $\theta_\rc\le \pi$
\begin{align}
S_\rP^\prime = \dfrac{\beta_D}{2}  r_0^{D-1} \mathcal{I}_{r_\rcap^2/r_0^2}\left(\dfrac{D-1}{2},\dfrac{1}{2}\right)  \qquad : \quad 0 \leq h \leq r_0
\end{align}
is available in the literature \cite{Li2011} which can be adapted to  
\begin{align}
S_\rP^\prime = \dfrac{\beta_D}{2}  r_0^{D-1} \left[2- \mathcal{I}_{r_\rcap^2/r_0^2}\left(\dfrac{D-1}{2},\dfrac{1}{2}\right)\right]  \qquad : \quad r_0 \leq h \leq 2 r_0
\end{align}
for $\theta_\rc > \pi$. In $\mathbb{R}^4$, Eq. \eqref{eq:cap_S1} can be used for the lateral surface area of hyperspherical caps with arbitrary critical angles $\theta_\rc$. 
Using 
\begin{align}
S_\rP^{\prime \prime} = \kappa_{D-1} r_\rcap^{D-1}
\end{align}
for the generalized contribution of the singularity, the expression 
\begin{align}
S_\rP = \left\{ \begin{array}{lcc}
\dfrac{\beta_D}{2}  r_0^{D-1} \mathcal{I}_{r_\rcap^2/r_0^2}\left(\dfrac{D-1}{2},\dfrac{1}{2}\right) + \kappa_{D-1} r_\rcap^{D-1} & : & 0 \leq h \leq r_0 \\
\dfrac{\beta_D}{2}  r_0^{D-1} \left[2- \mathcal{I}_{r_\rcap^2/r_0^2}\left(\dfrac{D-1}{2},\dfrac{1}{2}\right)\right] + \kappa_{D-1} r_\rcap^{D-1} & : & r_0 \leq h \leq 2 r_0
\end{array} \right.
\end{align}
results for the total surface area of $D$-dimensional spherical caps in $\mathbb{R}^D$.

For the second quermassintegral 
\begin{align}\label{eq:hcap_W2_0}
W_2 = W_2^\prime + W_2^{\prime \prime}\,,
\end{align}
the contribution for the continuous part $W_2^\prime$ can be written as 
\begin{align}\label{eq:hcap_W2_1}
W_2^\prime 
&= \dfrac{\pi^2}{4} r_0^2 \left[1-\dfrac{2}{\pi} \arctan \left(\dfrac{r_0-h}{r_\rcap} \right) - \dfrac{2 r_\rcap}{\pi r_0^2} \left(r_0 -h \right) \right]
\end{align}
using Eq. \eqref{eq:W2_an} with the principal radii of curvature  
\begin{align*}
R_1 = R_2 = R_3 = r_0
\end{align*}
resulting from the meridian curve $r(z)$ [Eq. \eqref{eq:hsp_R1R2R3}]. 
Using the regularized incomplete beta function $\mathcal{I}_x(a,b)$, also 
\begin{align}
W_2^\prime = \left\{ \begin{array}{ccc}
\dfrac{1}{4} \pi^2 r_0^{2} \mathcal{I}_{r_\rcap^2/r_0^2}\left(\dfrac{3}{2},\dfrac{1}{2}\right) & : & 0 \leq h \leq r_0 \vspace*{2mm} \\ 
\dfrac{1}{4} \pi^2 r_0^{2} \left[ 2-\mathcal{I}_{r_\rcap^2/r_0^2}\left(\dfrac{3}{2},\dfrac{1}{2}\right)\right] & : & r_0 \leq h \leq 2 r_0
\end{array} \right.
\end{align}
can be formulated in analogy to the previously described geometric measures. 
The contribution of the singularity $W_2^{\prime \prime}$ consists of half the contribution of a spherical hyperplate with radius $r_\rcap$ [Eq. \eqref{eq:Sph_plate_W2}] and half the contribution of the singularity of a hyperlens with radius $r_\rcap$ [Eq. \eqref{eq:Hlens_W2''}] where
\begin{align}\label{eq:hcap_W2_2}
W_2^{\prime \prime} &= \dfrac{1}{2} \, \dfrac{1}{3} \pi^2 r_\rcap^2 + \dfrac{1}{2} \, \dfrac{2}{3} \pi r_\rcap^2 \arcsin \left[ \cos \left(\dfrac{\theta_\rc}{2} \right) \right] \nonumber \\
&= \dfrac{1}{3} \pi r_\rcap^2  \left[\dfrac{\pi}{2}+ \arcsin\left( 1- \dfrac{h}{r_0} \right)\right]
\end{align}
is obtained. 
Using Eqs. \eqref{eq:hcap_W2_0}, \eqref{eq:hcap_W2_1}, and \eqref{eq:hcap_W2_2}
\begin{align}
W_2 &= \dfrac{\pi^2}{4} r_0^2 \left[1-\dfrac{2}{\pi} \arctan \left(\dfrac{r_0-h}{r_\rcap} \right) - \dfrac{2 r_\rcap}{\pi r_0^2} \left(r_0 -h \right) \right]\nonumber \\
&\qquad + \dfrac{\pi}{3}  r_\rcap^2  \left[\dfrac{\pi}{2} +  \arcsin\left(1-\dfrac{h}{r_0} \right) \right]
\end{align}
results for the second quermassintegral $W_2$ of hyperspherical caps in $\mathbb{R}^4$. 

For the mean radius of curvature 
\begin{align}
\tilde{R}_\rP = \tilde{R}_\rP^\prime + \tilde{R}_\rP^{\prime \prime}\,,
\end{align}
the contribution of the continuous part $\tilde{R}_\rP^\prime$ is 
\begin{align}
\tilde{R}_\rP^\prime 
&= \dfrac{r_0}{2} \left[1-\dfrac{2}{\pi} \arctan \left(\dfrac{r_0-h}{r_\rcap} \right) - \dfrac{2 r_\rcap}{\pi r_0^2} \left(r_0 -h \right) \right]
\end{align}
using Eqs. \eqref{eq:RP_an} and \eqref{eq:hcap_req}. 
Using the regularized incomplete beta function $\mathcal{I}_x(a,b)$, $\tilde{R}^\prime_\rP$ can be reformulated as 
\begin{align}
\tilde{R}_\rP^\prime = \left\{ \begin{array}{ccc}
\dfrac{r_0}{2}\, \mathcal{I}_{r_\rcap^2/r_0^2}\left(\dfrac{3}{2},\dfrac{1}{2}\right) & : & 0 \leq h \leq r_0 \vspace*{2mm} \\ 
\dfrac{r_0}{2} \left[ 2-\mathcal{I}_{r_\rcap^2/r_0^2}\left(\dfrac{3}{2},\dfrac{1}{2}\right)\right] & : & r_0 \leq h \leq 2 r_0
\end{array} \right.\,.
\end{align}
For the contribution at the singularity $\tilde{R}_\rP^{\prime \prime}$, again, half the contribution of a spherical hyperplate [Eq. \eqref{eq:Sph_plate_RP}] and half the contribution of the singularity of a hyperlens [Eq. \eqref{eq:Hlens_RP''}] are added as
\begin{align}
\tilde{R}_\rP^{\prime \prime} &= \dfrac{1}{2} \, \dfrac{8}{3\pi} r_\rcap + \dfrac{1}{2} \, \dfrac{8}{3 \pi} r_\rcap \cos \left(\dfrac{\theta_\rc}{2} \right) = \dfrac{4}{3 \pi} r_\rcap \left(2-\dfrac{h}{r_0} \right)
\end{align}
where
\begin{align}
\tilde{R}_\rP = \dfrac{r_0}{2} \left[1-\dfrac{2}{\pi} \arctan \left(\dfrac{r_0-h}{r_\rcap} \right) - \dfrac{2 r_\rcap}{\pi r_0^2} \left(r_0 -h \right) \right] + \dfrac{4}{3 \pi} r_\rcap \left(2-\dfrac{h}{r_0} \right)
\end{align}
results for the total mean radius of curvature $\tilde{R}_\rP$. 


\subsubsection{Hyperdoublecone} 

A hyperdoublecone is a uniaxial solid of revolution with inversion symmetry whose meridian curve can be written as
\begin{align}\label{eq:Hdcn_rz}
r(z) = r_\req - \dfrac{\vert z \vert}{\nu}
\end{align}
in dependence on the equatorial radius $r_\req$ and aspect ratio $\nu$.
Hyperdoublecones exist as oblate $(\nu\le 1)$ and prolate $(\nu\ge 1)$ solids with the limiting cases of a spherical hyperplate
for $\nu\to 0$ and a hard needle for $\nu\to\infty$ in $\mathbb{R}^4$ analogously to three-dimensional doublecones \cite{Herold2017}. 
The two-dimensional section of a hyperdoublecone with aspect ratio $\nu=2$ is shown in Fig. \ref{fig:hdc}. 
This geometry has an equatorial singularity with a critical angle $\theta_\rc$ which is defined by 
\begin{align}
\cos \left( \theta_\rc \right) = \dfrac{1}{\sqrt{1+\nu^2}}\,,
\end{align}
similar to hyperlenses. 
Additionally, two apical singularities exist at $z=\pm \nu r_\req$ similar to hyperspindles. 

\begin{figure}[ht]
\centering
\includegraphics[width=0.2\textwidth]{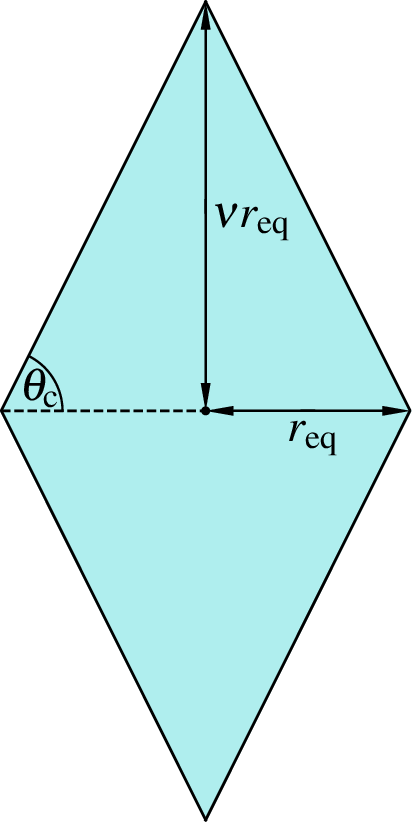}
\caption{Two-dimensional section of a hyperdoublecone with radius $r_\req$, aspect ratio $\nu=2$ and critical angle $\theta_\rc$.}\label{fig:hdc}
\end{figure}

The volume $V_\rP$ can be written as 
\begin{align}\label{eq:Hdcn_VP}
V_\rP &= \dfrac{4}{3} \pi \int\limits_{- \nu r_\req}^{\nu r_\req}  \left(r_\req - \dfrac{\vert z \vert}{\nu} \right)^3 \,\rd z = \dfrac{2}{3} \pi \nu r_\req^4
\end{align}
using Eqs. \eqref{eq:VP} and \eqref{eq:Hdcn_rz}. 
Since the lateral surface area $S_\rP^\prime$ of this geometry equals its total surface area $S_\rP$, Eq. \eqref{eq:MP} can be used and 
\begin{align}\label{eq:Hdcn_SP}
S_\rP = S_\rP^\prime = 4 \pi \dfrac{\sqrt{1+\nu^2}}{\nu}  \int\limits_{-\nu r_\req}^{\nu r_\req} \left(r_\req - \dfrac{\vert z \vert}{\nu} \right)^2  \,\rd z = \dfrac{8}{3} \pi r_\req^3 \sqrt{1+\nu^2}
\end{align}
is obtained.
For the principal radii of curvature 
\begin{align}
R_1(z) = R_2(z)= \left(r_\req - \dfrac{\vert z \vert }{\nu} \right) \dfrac{\sqrt{1+\nu^2}}{\nu}
\end{align}
and 
\begin{align}
\vert R_3\vert = \infty
\end{align}
result from Eqs. \eqref{eq:R1R2} and \eqref{eq:R3} using Eq. \eqref{eq:Hdcn_rz}.

For the second quermassintegral the contribution of the continuous part $W_2^\prime$ can be written as 
\begin{align}
\label{eq:Hdcn_W2a}
W_2^\prime = \dfrac{2}{3} \pi \int\limits_{-\nu r_\req}^{\nu r_\req} \left(r_\req - \dfrac{\vert z \vert}{\nu} \right)\,\rd z = \dfrac{2}{3} \pi \nu r_\req^2
\end{align}
[Eq. \eqref{eq:W2_an}]. 
For the contribution of the equatorial singularity $W_2^{\prime \prime}$, analogously to hyperlenses, an infinitely thin hypertorus can be used for a continuous completion. 
Therefore, 
\begin{align}
W_2^{\prime \prime} &= \dfrac{2}{3} \pi r_\req^2 \arcsin \left[\cos \left(\theta_\rc \right)  \right] = \dfrac{2}{3} \pi r_\req^2 \arcsin\left( \dfrac{1}{\sqrt{1+\nu^2}} \right)
\end{align} 
results, similar to hyperlenses [Eq. \eqref{eq:Hlens_W2''}]. 
For the contribution at the apical singularities $W_2^{\prime \prime \prime}$, similar to hyperspindles, 
\begin{align}
W_2^{\prime \prime \prime} = 0
\end{align}
is obtained [Eq. \eqref{eq:W2_apex}]. 
Therefore, the second quermassintegral can be written as 
\begin{align}\label{eq:Hdcn_W2}
W_2 &= W_2^\prime + W_2^{\prime \prime} + 2 W_2^{\prime \prime \prime} = \dfrac{2}{3} \pi r_\req^2 \left[\nu + \arcsin\left( \dfrac{1}{\sqrt{1+\nu^2}} \right)  \right]
\end{align}
in dependence on the equatorial radius $r_\req$ and the aspect ratio $\nu$. 

For the mean radius of curvature $\tilde{R}_\rP$, similar to the second quermassintegral $W_2$,  
\begin{align}\label{eq:Hdcn_W2'}
\tilde{R}_\rP^\prime 
= \dfrac{4}{3 \pi}  r_\req \dfrac{\nu^2}{\sqrt{1+\nu^2}}
\end{align}
is obtained from Eq. \eqref{eq:RP_an} for the contribution of the continuous part $\tilde{R}_\rP^\prime$, 
\begin{align}
\tilde{R}_\rP^{\prime\prime} = \dfrac{8}{3 \pi} r_\req \cos \left( \theta_\rc \right) = \dfrac{8}{3 \pi}  r_\req \dfrac{1}{\sqrt{1+\nu^2}} 
\end{align}
for the contribution of the equatorial singularity $\tilde{R}_\rP^{\prime\prime}$ [Eq. \eqref{eq:Hlens_RP''}], and 
\begin{align}
\tilde{R}_\rP^{\prime \prime \prime} = 0
\end{align}
for the contribution of the apical singularities $\tilde{R}_\rP^{\prime \prime \prime}$ [Eq. \eqref{eq:RP_apex}]. 
Hence, 
\begin{align}\label{eq:Hdcn_RP}
\tilde{R}_\rP &= \tilde{R}_\rP^{\prime}  + \tilde{R}_\rP^{\prime \prime }  +2 \tilde{R}_\rP^{\prime \prime \prime} = \dfrac{4}{3 \pi} r_\req \dfrac{\nu^2 +2}{\sqrt{1+\nu^2}}
\end{align}
results for the mean radius of curvature $\tilde{R}_\rP$.

\subsubsection{Hypercone}

Similar to a hyperdoublecone, placing the spherical base's center in the coordinate system's origin, the meridian curve of a hypercone reads as
\begin{align}\label{eq:Hcon_rz}
r(z) = r_0 - \dfrac{z}{\nu}
\end{align}
in dependence on its base radius $r_0$ and aspect ratio $\nu$.
In Fig. \ref{fig:hcn}, the two-dimensional section of a hypercone with aspect ratio $\nu =2$ is shown. 
Identical to the hyperdoublecone, a hypercone's critical angle $\theta_\rc$ is 
\begin{align}
\cos \left( \theta_\rc \right) = \dfrac{1}{\sqrt{1+\nu^2}}\,.
\end{align}

\begin{figure}[ht]
\centering
\includegraphics[width=0.4\textwidth]{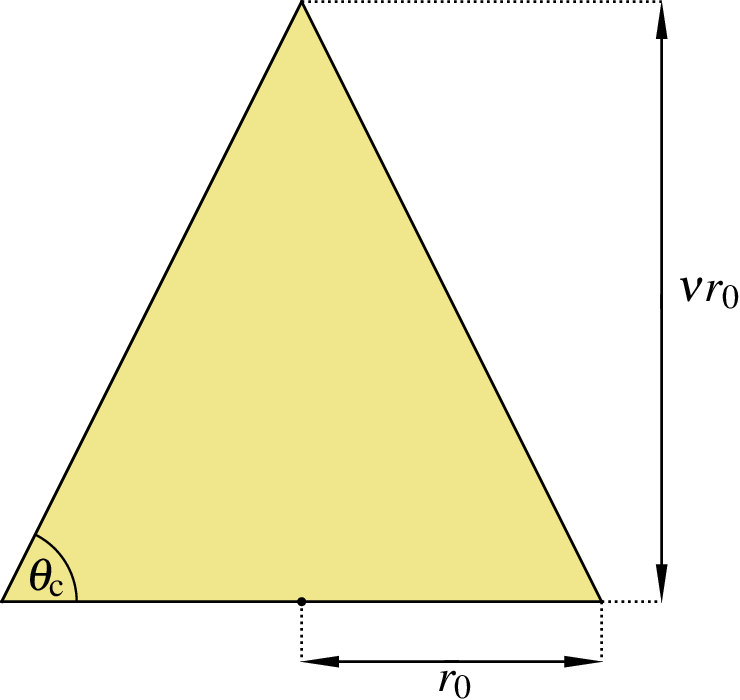}
\caption{Two-dimensional section of a hypercone with radius $r_0$, aspect ratio $\nu=2$ and critical angle $\theta_\rc$.}\label{fig:hcn}
\end{figure}

The volume of a four-dimensional hypercone reads as 
\begin{align}\label{eq:Hcon_VP}
V_\rP = \dfrac{4}{3} \pi \int\limits_0^{\nu r_0}  \left(r_0 - \dfrac{z}{\nu} \right)^3 \,\rd z 
= \dfrac{1}{3} \pi \nu r_0^4
\end{align}
[Eq. \eqref{eq:VP}], half the volume of a corresponding hyperdoublecone [Eq. \eqref{eq:Hdcn_VP}]. 
Its lateral surface area $S_\rP^\prime$ can be written as 
\begin{align}\label{eq:Hcon_SP'}
S_\rP^\prime 
= \dfrac{4}{3} \pi r_0^3 \sqrt{1+\nu^2} 
\end{align}
[Eq. \eqref{eq:MP}], half the total surface area of the corresponding hyperdoublecone [Eq. \eqref{eq:Hdcn_SP}]. With the contribution of the base 
\begin{align}
S_\rP^{\prime \prime} = \dfrac{4}{3} \pi r_0^3\,,
\end{align}
the total surface area reads as
\begin{align}\label{eq:Hcon_SP}
S_\rP &= S_\rP^\prime + S_\rP^{\prime \prime} = \dfrac{4}{3} \pi r_0^3 \left( \sqrt{1+\nu^2}\, +1 \right)
\end{align}
in dependence on the radius $r_0$ and aspect ratio $\nu$. 
For the contribution of the continuous part to the second quermassintegral
\begin{align}
W_2^\prime 
= \dfrac{1}{3} \pi \nu r_0^2
\end{align}
[Eq. \eqref{eq:W2_an}], half of the contribution of a corresponding hyperdoublecone is obtained [Eq. \eqref{eq:Hdcn_W2a}]. 
The principal radii of curvature with 
\begin{align}
R_1(z) = R_2(z) = \left(r_0 - \dfrac{z}{\nu} \right) \dfrac{\sqrt{1+\nu^2}}{\nu} 
\end{align} 
and
\begin{align}
\vert R_3\vert =  \infty
\end{align} 
of a hypercone are identical to those of a hyperdoublecone. Similar to hyperspherical caps, the contribution of the singularity at the base $W_2^{\prime \prime}$ consists of 
half the contribution of a spherical hyperplate and half the contribution of a hyperdoublecone's singularity.
Hence, the base singularity contributes 
\begin{align}
W_2^{\prime \prime} &= \dfrac{1}{6}\pi^2 r_0^2 + \dfrac{1}{3} \pi r_0^2 \arcsin\left[\cos\left(\theta_\rc \right) \right] \nonumber \\
&= \dfrac{1}{6} \pi r_0^2 \left[ \pi + 2 \arcsin\left(\dfrac{1}{\sqrt{1+\nu^2}}\right) \right]
\end{align}
to the second quermassintegral. Since the contribution of the apical singularity $W_2^{\prime \prime \prime}=0$ vanishes, the second quermassintegral $W_2$ of a hypercone can be written as 
\begin{align}\label{eq:Hcon_W2}
W_2 &= W_2^{\prime} + W_2^{\prime \prime} + W_2^{\prime \prime \prime} = \dfrac{1}{6} \pi r_0^2 \left[2\nu +\pi + 2\arcsin\left(\dfrac{1}{\sqrt{1+\nu^2}}\right) \right]\,.
\end{align}

Similar, the contribution $\tilde{R}_\rP^\prime$ of the continuous surface curvature to the mean radius of curvature  
\begin{align}
\tilde{R}_\rP^\prime 
=\dfrac{2}{3 \pi} r_0 \dfrac{\nu^2}{\sqrt{1+\nu^2}}
\end{align}
is half the contribution of a corresponding hyperdoublecone [Eq. \eqref{eq:Hdcn_W2'}]. 
For the contribution at the base singularity, analogously to the contribution of the second quermassintegral, half a spherical hyperplate and half a hyperdoublecone can be used, resulting in 
\begin{align}
\tilde{R}_\rP^{\prime\prime} &= \dfrac{1}{2} \, \dfrac{8}{3 \pi} r_0 + \dfrac{1}{2} \, \dfrac{8}{3 \pi} r_0 \cos\left(\theta_\rc \right)  = \dfrac{4}{3 \pi} r_0 \left( 1 + \dfrac{1}{\sqrt{1+\nu^2}} \right)\,.
\end{align}
Since the apical singularity with $ \tilde{R}_\rP^{\prime\prime\prime} = 0 $ [Eq. \eqref{eq:RP_apex}] does not contribute to the mean radius of curvature, $\tilde{R}_\rP$ 
of a four-dimensional hypercone can be written as
\begin{align}\label{eq:Hcon_RP}
\tilde{R}_\rP  &= \tilde{R}_\rP^\prime + \tilde{R}_\rP^{\prime\prime} + \tilde{R}_\rP^{\prime\prime\prime} = \dfrac{2}{3 \pi} r_0 \left(2+\dfrac{\nu^2+2}{\sqrt{1+\nu^2}} \right)\,.
\end{align}

\subsubsection{Truncated hypercone} 

When a hypercone with base radius $r_0$ and aspect ratio $\nu$ is cut off by a hyperplane parallel to the base at height $h$, a truncated hypercone results with radius $r_1$ at the second singularity. 
The section of such a truncated hypercone is shown for an aspect ratio of $\nu =2$ and height $h=r_0$ in Fig. \ref{fig:hct}. 
The meridian curve $r(z)$ is again 
\begin{align}
\label{eq:meridian_truncated_cone}
r(z) = r_0 - \dfrac{z}{\nu} = r_0 - \dfrac{r_0 -r_1}{h} \, z
\end{align}
with the center of the base with radius $r_0$ as origin of the coordinate system. 
The aspect ratio $\nu$ can be defined by 
\begin{align}\label{eq:hcnt_nu}
\nu = \dfrac{h}{r_0 -r_1} = \tan \left( \theta_\rc \right)
\end{align}
restricting $r_1 < r_0$. The case $r_1=r_0$ resulting in a hypercylinder with aspect ratio $\nu^\prime=h/(2 r_0)$ is discussed in section \ref{sec:hcyl}.
In the limits $\nu \to 0$ and $h \to 0$ a spherical hyperplate with radius $r_0 =r_1$ results while for the limit $r_1 / r_0 \to  0$ a hypercone with radius $r_0$ and aspect ratio $\nu$ is obtained. 

Again, $\theta_\rc$ is the critical angle with 
\begin{align}
\cos \left(\theta_\rc \right) = \dfrac{1}{\sqrt{1+\nu^2}}\,,
\end{align}
identical to hyperdoublecones and hypercones. 
\begin{figure}[ht]
\centering
\includegraphics[width=0.5\textwidth]{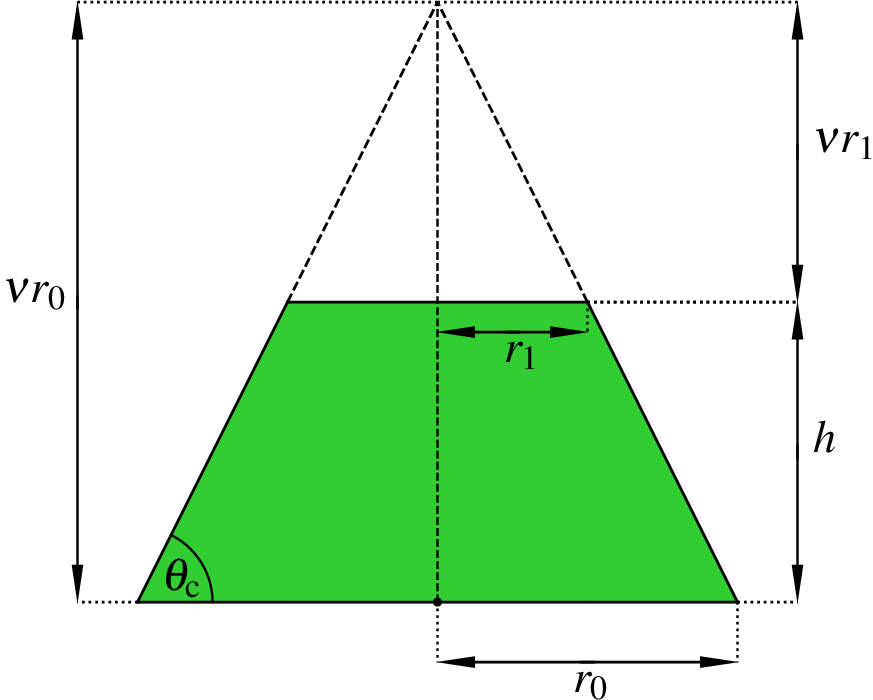}
\caption{Two-dimensional section of a truncated hypercone with radii $r_0$, $r_1$, height $h=r_0$,  aspect ratio $\nu=2$, and critical angle $\theta_\rc$.}\label{fig:hct}
\end{figure} 
The volume $V_\rP$ of a truncated hypercone can with Eq. \eqref{eq:VP} be written as 
\begin{align}
V_\rP &= \dfrac{4}{3} \pi \int\limits_0^h \left(r_0 - \dfrac{r_0 -r_1}{h}\,z \right)^3\,\rd z = \dfrac{h}{3} \pi \left(r_0^3 + r_0^2r_1 + r_0r_1^2 + r_1^3 \right)
\end{align}
which can be extended to 
\begin{align}
V_\rP = \kappa_{D-1} \dfrac{h}{D} \dfrac{r_1^D - r_0^D}{r_1 -r_0}
\end{align}
for $D$-dimensional truncated cones in Euclidean spaces $\mathbb{R}^D$. 
Its lateral surface area $S_\rP^\prime$ can be written as
\begin{align}
S_\rP^\prime 
&= \dfrac{4}{3} \pi \left(r_0^2 + r_0 r_1 + r_1^2 \right) \sqrt{h^2+(r_0-r_1)^2}
\end{align}
[Eq. \eqref{eq:MP}] and with the contributions  
\begin{align}
S_\rP^{\prime \prime} = \dfrac{4}{3} \pi r_0^3
\end{align}
and 
\begin{align}
S_\rP^{\prime \prime \prime} = \dfrac{4}{3} \pi r_1^3
\end{align}
at the singularities, the total surface area $S_\rP=S_\rP^\prime + S_\rP^{\prime\prime} + S_\rP^{\prime\prime\prime}$ as
\begin{align}
S_\rP = \dfrac{4}{3} \pi \left[ r_0^3 + r_1^3 + \left(r_0^2 + r_0 r_1 + r_1^2 \right) \sqrt{h^2 + \left(r_0 -r_1 \right)^2} \, \right]
\end{align}
is obtained. This can be generalized to $\mathbb{R}^D$ as 
\begin{align}
S_\rP = \kappa_{D-1} \left[r_0^{D-1} + r_1^{D-1} + \dfrac{r_1^{D-1}-r_0^{D-1}}{r_1-r_0} \sqrt{h^2 + \left(r_0 -r_1 \right)^2} \,\right]
\end{align}
using the common relation
\begin{align}
\kappa_{D} = \dfrac{\beta_{D}}{D}
\end{align}
of the volume $\kappa_{D}$ to the surface area $\beta_{D}$ for a $D$-dimensional unit sphere [Eqs. \eqref{eq:beta_D} and \eqref{eq:kappa}]. 
Combining the meridian curve [Eq. \eqref{eq:meridian_truncated_cone}] with Eqs. \eqref{eq:R1R2} and \eqref{eq:R3} 
\begin{align}
R_1(z) = R_2(z) &= \left(r_0 - \dfrac{r_0 -r_1}{h}\,z \right) \dfrac{\sqrt{h^2 + ( r_0 -r_1)^2}}{h}
\end{align} 
and
\begin{align}
\vert R_3\vert =  \infty
\end{align} 
result as principal radii of curvature.
Herewith, the contribution of the part with continuous surface curvature to the second quermassintegral of a truncated hypercone can be written as
\begin{align}
W_2^\prime &= \dfrac{1}{3} \pi h \left(r_0 + r_1 \right)
\end{align}
using Eq. \eqref{eq:W2_an}. 
For the contribution at the base singularity, similar to hypercones
\begin{align}
W_2^{\prime \prime} &= \dfrac{1}{2} \, \dfrac{1}{3} \pi^2 r_0^2 + \dfrac{1}{2} \, \dfrac{2}{3} \pi r_0^2 \arcsin \left[\cos \left(\theta_\rc \right) \right] \nonumber \\
 &= \dfrac{1}{6} \pi r_0^2 \left[\pi + 2 \arcsin\left(\dfrac{r_0 -r_1}{\sqrt{h^2+ ( r_0 - r_1 )^2}} \right) \right]
\end{align}
results, using half of a spherical hyperplate and half of a hyperdoublecone. 
For the contribution at the top singularity to the second quermassintegral 
\begin{align}
W_2^{\prime \prime \prime} &= \dfrac{1}{2} \, \dfrac{1}{3} \pi^2 r_1^2 - \dfrac{1}{2} \, \dfrac{2}{3} \pi r_1^2 \arcsin \left[\cos \left(\theta_\rc \right) \right] \nonumber \\
&= \dfrac{1}{6} \pi r_1^2 \left[ \pi- 2\arcsin\left(\dfrac{r_0-r_1}{\sqrt{h^2+(r_0-r_1)^2}} \right)\right] 
\end{align}
is obtained in analogy to three-dimensional truncated cones \cite{Herold2017}. 
In total,  
\begin{align}
W_2 = W_2^\prime + W_2^{\prime \prime} + W_2^{\prime \prime \prime} =& \dfrac{1}{3} \pi \left(r_0+r_1 \right) \left\{h + \left(r_0 -r_1 \right) \arcsin\left[\dfrac{r_0-r_1}{\sqrt{h^2+(r_0-r_1)^2}} \right] \right\} \nonumber \\ & \quad +\dfrac{1}{6} \pi^2 \left(r_0^2+r_1^2 \right) 
\end{align}
results as second quermassintegral $W_2$ for truncated hypercones in $\mathbb{R}^4$. 

Finally, the mean radius of curvature $\tilde{R}_\rP$ of a truncated hypercone consists of the contribution of the continuous part 
\begin{align}
\tilde{R}_\rP^\prime = \dfrac{2}{3\pi} \dfrac{h^2}{\sqrt{h^2 + (r_0-r_1)^2}}
\end{align}
[Eq. \eqref{eq:RP_an}] and the contributions of both singularities    
\begin{align}
\tilde{R}_\rP^{\prime \prime} &= \dfrac{1}{2} \, \dfrac{8}{3 \pi} r_0 + \dfrac{1}{2} \, \dfrac{8}{3 \pi} r_0 \cos \left(\theta_\rc \right) = \dfrac{4}{3 \pi} r_0 \left(1+ \dfrac{r_0-r_1}{\sqrt{h^2+(r_0-r_1)^2}} \right)
\end{align}
and
\begin{align}
\tilde{R}_\rP^{\prime \prime \prime} &= \dfrac{1}{2} \, \dfrac{8}{3 \pi} r_1 - \dfrac{1}{2} \, \dfrac{8}{3 \pi} r_1 \cos \left(\theta_\rc \right) = \dfrac{4}{3 \pi} r_1 \left(1-\dfrac{r_0-r_1}{\sqrt{h^2+(r_0-r_1)^2}} \right)\,.
\end{align} 
Hence, for the mean radius of curvature of a truncated hypercone 
\begin{align}
\tilde{R}_\rP &= \tilde{R}_\rP^{\prime} + \tilde{R}_\rP^{\prime \prime} +\tilde{R}_\rP^{\prime \prime \prime} = \dfrac{2}{3 \pi} \left[2 (r_0+r_1) + \dfrac{h^2 + 2 (r_0-r_1)^2}{\sqrt{h^2+(r_0-r_1)^2}} \right]
\end{align}
is obtained.


\section{Results}\label{sec4}

Using the geometric measures volume $V_\rP$, total surface area $S_\rP$, second quermassintegral $W_2$, and mean radius of curvature $\tilde{R}_\rP$, the second virial coefficient $B_2$, representing the excluded volume per particle for the four-dimensional, convex solids of revolution provided in Sec. \ref{eq_sec:measures_4d_specific}, can be obtained using Eq. \eqref{eq:B2(K)}. 
To investigate the influence of the specific particle shape in dependence on the aspect ratio $\nu$, these virial coefficients can be normalized to the particle volume $V_\rP$ using 
\begin{align}\label{eq:B2red_an}
B_2^\ast = \dfrac{B_2}{V_\rP} = 1+ \dfrac{S_\rP \tilde{R}_\rP}{V_\rP} + \dfrac{6}{\pi^2} \dfrac{W_2^2}{V_\rP}
\end{align}
as reduced second virial coefficients $B_2^\ast$. 
The resulting reduced virial coefficients, independent of the particle size, allow an expansion of the real gas factor $Z$ in powers of the packing fraction $\eta$ [Eq. \eqref{eq:real_gas_factor_volume_fraction}] 
and  thus a comparison of different geometric shapes.

\subsection{Solids with inversion symmetry} 

In Fig. \ref{fig:b2_nu}, the second virial coefficients $B_2$ of  hyperellipsoids, hyperspherocylinders, hyperspindles, hypercylinders, hyperdoublecones, and hyperlenses as four-dimensional, uniaxial solids of revolution with inversion symmetry are shown. While hyperspherocylinders and hyperspindles are strictly prolate geometries with $\nu \geq 1$,  hyperlenses are strictly oblate geometries with $0 \leq \nu \leq 1$.  
Hyperellipsoids, hypercylinders, and hyperdoublecones can be prolate or oblate with $0\le\nu < \infty$.

\begin{figure}[h]
\centering
\includegraphics[width=0.49\textwidth,valign=t]{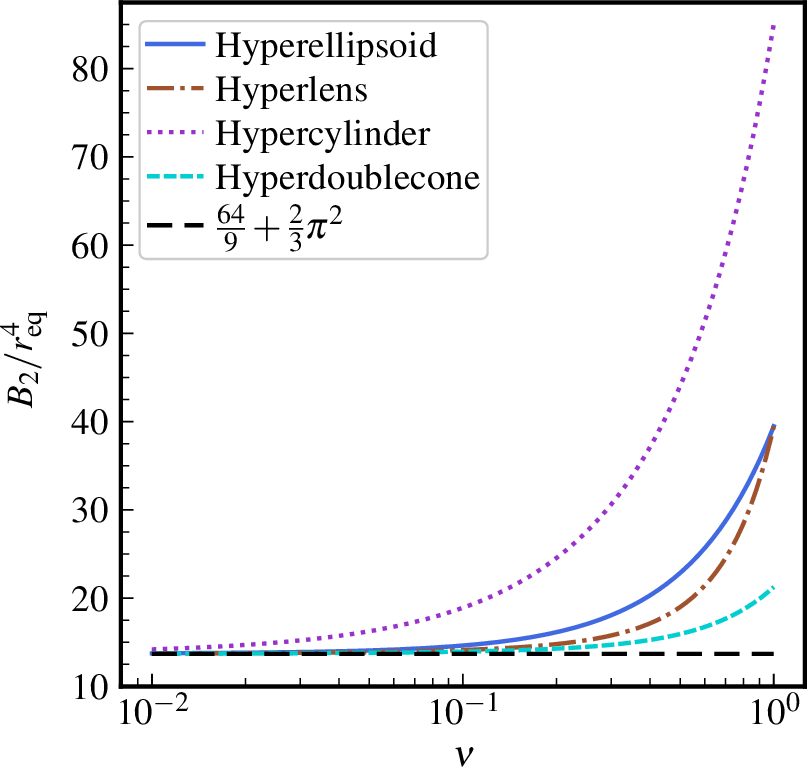}\hfill
\includegraphics[width=0.49\textwidth,valign=t]{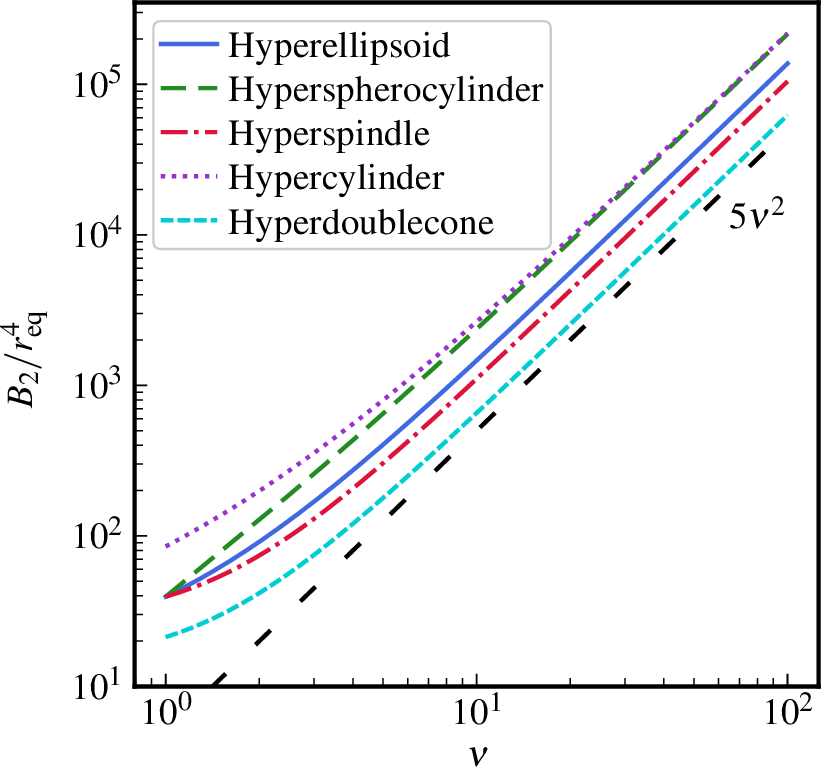}
\caption{Second virial coefficient $B_2/r_\req^4$ of oblate (lhs) and prolate (rhs) uniaxial solids of revolution with inversion symmetry in dependence on the aspect ratio $\nu$. For $\nu\to 0$, the common limit
$B_2/r_\req^4=64/9+2\pi^2/3$ of a spherical hyperplate is approached as visualized by the horizontal line. To indicate the common slope at infinite aspect ratio, $5\nu^2$ is depicted as a guide to the eye.}\label{fig:b2_nu}
\end{figure}

As visible in Fig. \ref{fig:b2_nu}, hyperellipsoids, hyperspherocylinders, hyperspindles, and hyperlenses approach in limit $\nu \to 1$  hyperspheres with radius $r_\req$ and 
\begin{align}
B_2^{ ( \mathrm{hsph} ) } = 4 \pi^2 r_\req^4
\end{align}
as their second virial coefficient. 

For oblate geometries, with decreasing aspect ratio all geometries approach in the limit 
\begin{align}
\lim\limits_{\nu \to 0} \dfrac{B_2^{(\mathrm{hdcn})}}{r_\req^4} &= 
\lim\limits_{\nu \to 0} \dfrac{B_2^{(\mathrm{hlen})}}{r_\req^4} =
\lim\limits_{\nu \to 0} \dfrac{B_2^{(\mathrm{hell})}}{r_\req^4} =
\lim\limits_{\nu \to 0} \dfrac{B_2^{(\mathrm{hcyl})}}{r_\req^4} 
= \dfrac{64}{9} + \dfrac{2}{3} \pi^2
\end{align}
spherical hyperplates with radius $r_\req$ and 
\begin{align}
B_2^{(\mathrm{shyp})} = \left(\dfrac{64}{9} + \dfrac{2}{3} \pi^2 \right) r_\req^4
\end{align}
as their second virial coefficient.

In the limit $\nu \to \infty$, hard hyperneedles result where hyperspherocylinders approach hypercylinders since the contribution of the capping hyperspheres becomes negligible. 
For sufficiently large aspect rations, the proportionality $B_2(\nu \gg 1) \propto \nu^2$ is observed as indicated by the black line $5\nu^2$. 
The limits of $B_2/(\nu^2 r_\req^4)$ with 
\begin{subequations}
\begin{align}
\lim\limits_{\nu \to \infty} \dfrac{B_2^{(\mathrm{hcyl})}}{\nu^2 r_\req^4} &= \dfrac{64}{3} \,, \\ 
\lim\limits_{\nu \to \infty} \dfrac{B_2^{(\mathrm{hscy})}}{\nu^2 r_\req^4} &= \dfrac{64}{3} \,, \\
\lim\limits_{\nu \to \infty} \dfrac{B_2^{(\mathrm{hell})}}{\nu^2 r_\req^4} &=  \dfrac{64}{9} + \dfrac{2}{3} \pi^2 \,, \\
\lim\limits_{\nu \to \infty} \dfrac{B_2^{(\mathrm{hspi})}}{\nu^2 r_\req^4} &= \dfrac{1408}{135}  \,, \\
\lim\limits_{\nu \to \infty} \dfrac{B_2^{(\mathrm{hdcn})}}{\nu^2 r_\req^4} &= \dfrac{56}{9}
\end{align}
\end{subequations}
decrease from hypercylinders (hcyl) and hyperspherocylinders (hscy) with identical value over hyperellipsoids (hell) and hyperspindles (hspi) to hyperdoublecones (hdcn).
This sequence is in accordance to the relation 
\begin{align}
K^{(\mathrm{hdcn})}\subset K^{(\mathrm{hspi})}\subset K^{(\mathrm{hell})}\subset K^{(\mathrm{hscy})}\subset K^{(\mathrm{hcyl})}
\end{align}
for convex sets $K\equiv K(\nu,r_\req)$ of hyperdoublecones, hyperspindles, hyperellipsoids, hyperspherocylinders, and hypercylinders with identical aspect ratio $\nu$ and equatorial radius $r_\req$.
This behavior is identical to their three-dimensional analogues \cite{Herold2017}.

\subsubsection{Prolate geometries}\label{sec:res_pro}

To analyze the influence of the detailed particle geometry, reduced virial coefficients $B_2^\ast$ are compared in Fig. \ref{fig:b2red_prolate} (lhs) in dependence on their aspect ratio $\nu$. 
For a better display of relative deviances, the ratios of reduced second virial coefficients $B_2^\ast(\nu)$ to the reduced second virial coefficient $B_2^{\ast\,\rm (hell)}(\nu)$ of hyperellipsoids are shown in Fig. \ref{fig:b2red_prolate} (rhs).  
In the limit $\nu \to 1$, hyperellipsoids, hyperspherocylinders, and hyperspindles approach hyperspheres with a reduced second virial coefficient of $B_2^{\ast\,(\mathrm{hsph})} = 8$. 
For arbitrary dimensions $D$, the reduced second virial coefficient of $D$-dimensional spheres is 
\begin{align}\label{eq:B2red_D_sphere}
B_2^{\ast\,(\mathrm{hsph})} = 2^{D-1}
\end{align}
as commonly known \cite{Urrutia2022}.

\begin{figure}[ht]
\centering
\includegraphics[width=0.475\textwidth,valign=t]{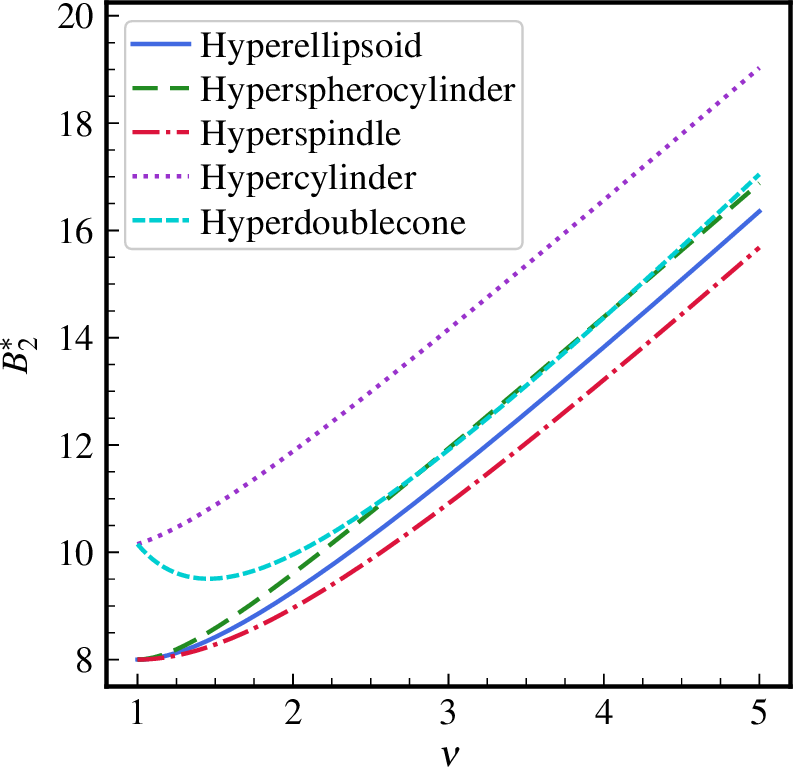} \hfill
\includegraphics[width=0.5\textwidth,valign=t]{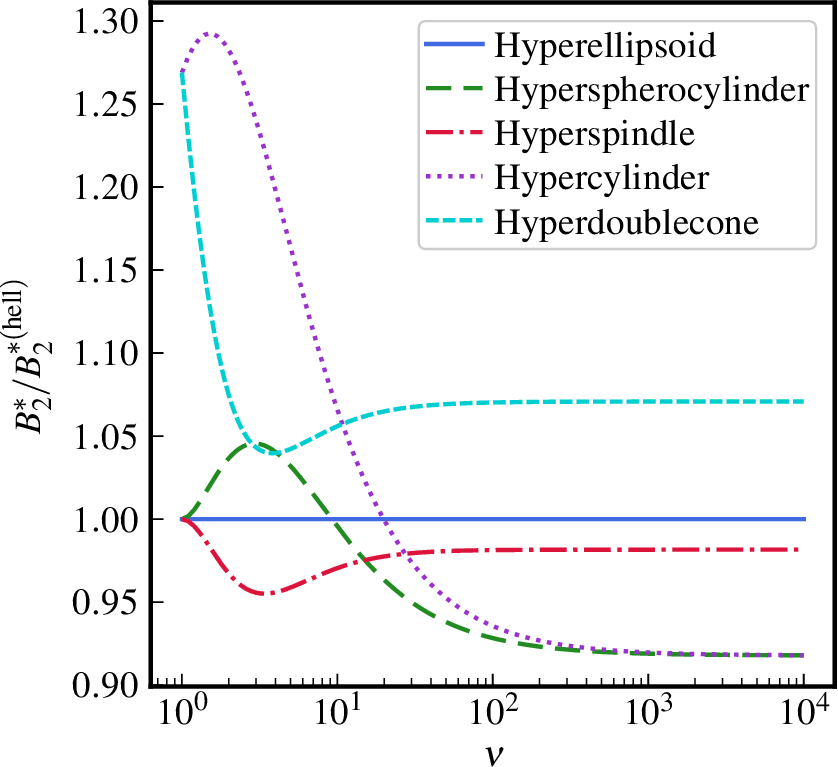}
\caption{Reduced second virial coefficients $B_2^\ast$ of prolate solids of revolution with inversion symmetry (lhs) and these coefficients normalized to the reduced second virial coefficients of hyperellipsoids $B_2^{\ast\,(\mathrm{hell})}$ (rhs) in dependence on the aspect ratio $\nu$.}\label{fig:b2red_prolate}
\end{figure}

With increasing aspect ratio $\nu$, the proportionality $B_2^\ast (\nu \gg 1) \propto \nu$ is observed with the limits 
\begin{subequations}\label{eq:b2_red_lim_pro}
\begin{align}
\lim\limits_{\nu \to \infty} \dfrac{B_2^{\ast\,(\mathrm{hcyl})}}{\nu} &= \dfrac{8}{\pi} \,,\\ 
\lim\limits_{\nu \to \infty} \dfrac{B_2^{\ast\,(\mathrm{hscy})}}{\nu} &= \dfrac{8}{\pi} \,,\\
\lim\limits_{\nu \to \infty} \dfrac{B_2^{\ast\,(\mathrm{hspi})}}{\nu} &= \dfrac{77}{9 \pi} \,,\\
\lim\limits_{\nu \to \infty} \dfrac{B_2^{\ast\,(\mathrm{hell})}}{\nu} &= \dfrac{4}{3} + \dfrac{128}{9 \pi^2}\,, \label{eq:b2_red_lim_pro_hell} \\
\lim\limits_{\nu \to \infty} \dfrac{B_2^{\ast\,(\mathrm{hdcn})}}{\nu} &= \dfrac{28}{3\pi} 
\end{align}
\end{subequations}
in increasing order [Fig. \ref{fig:b2red_prolate} (rhs)].
Again, for large aspect ratios hyperspherocylinders approach hypercylinders with a negligible contribution of the capping hyperspheres. 
The limit for four-dimensional hyperspherocylinders is already reported in a previous work \cite{Kulossa2022}. 

Particularly interesting is the limiting case $\nu \to 1$ for hypercylinders and hyperdoublecones: 
For four-dimensional hypercylinders and hyperdoublecones with aspect ratio $\nu=1$, the reduced second virial coefficient 
\begin{align}
\left[B_2^{\ast\,(\mathrm{hcyl})}(\nu=1)\right] = \left[B_2^{\ast\,(\mathrm{hdcn})}(\nu=1)\right] = 3 + \dfrac{20}{\pi} + \dfrac{\pi}{4} 
\end{align}
is identical and thus the rotation-averaged excluded volume $V_\rex$ for particles with the same volume $V_\rP$ is identical, too.  
For the volumes of $D$-dimensional cylinders and doublecones in $\mathbb{R}^D$, in general
\begin{align}
V_\rP^{(\mathrm{cyl})}(\nu) = D V_\rP^{(\mathrm{dcn})}(\nu) = 2  \kappa_{D-1} \nu r_\req^D
\end{align}
is obtained [Eqs. \eqref{eq:V_P^D}, \eqref{eq:Hc_rz}, and \eqref{eq:Hdcn_rz}]. 
The reduced second virial coefficients $B_2^\ast$ of four-dimensional hypercylinders and hyperdoublecones can also be compared to $B_2^\ast$ of their analogue geometries in lower-dimensional Euclidean spaces: 
In $\mathbb{R}^2$, for the meridian curve of a hypercylinder a rectangle and for the meridian curve of a hyperdoublecone a rhombus with half the area of the corresponding rectangle are obtained.
In the limit $\nu \to 1$, hard squares result for both geometries and therefore also the same reduced second virial coefficient $B_2^\ast = 1 + 4 / \pi$ results \cite{Kulossa2023b}. 

In $\mathbb{R}^3$, however, the reduced second virial coefficients of cylinders and doublecones with aspect ratio $\nu=1$ differ \cite{Herold2017}. 
This results lead to the conjecture that the rotation-averaged excluded volumes of $D$-dimensional, uniaxial cylinders and doublecones with aspect ratio $\nu=1$ and identical volumes $V_\rP^{\rm (cyl)}=V_\rP^{\rm (dcn)}$ are identical in even-dimensional Euclidean spaces and different in odd-dimensional spaces. The proof is beyond the scope of this work and will be topic of future work.

Comparing Fig. \ref{fig:b2red_prolate} (rhs) with the analogue three-dimensional results  (Fig. 7 in \cite{Herold2017}), characteristic differences are observed between $\mathbb{R}^3$ and $\mathbb{R}^4$:  
For large aspect ratios $\nu \gg 1$, the ratios of the reduced second virial coefficients $B_2^\ast/B_2^{\ast (\rm hell)}$ are more similar in $\mathbb{R}^4$ than in $\mathbb{R}^3$ due to the additional symmetry of the solid. 
The reduced second virial coefficients of hyperspherocylinders and hypercylinders in $\mathbb{R}^4$ normalized to those of hyperellipsoids similarly depend on the aspect ratio $\nu$  as the reduced second virial coefficients of spherocylinders and cylinders normalized to ellipsoids in $\mathbb{R}^3$. However, for solids with apical singularities, deviations are observed. 
Opposite to spindles in $\mathbb{R}^3$, for hyperspindles in $\mathbb{R}^4$ the ratio $\left[B_2^\ast(\nu)/B_2^{\ast (\rm hell)}(\nu)\right] <1$ results for $\nu \to \infty$. 
Also for hyperdoublecones with increasing aspect ratio $\nu$, the ratio $B_2^\ast(\nu)/B_2^{\ast (\rm hell)}(\nu)$ is significantly smaller than for doublecones in $\mathbb{R}^3$. Especially, the reduced second virial coefficients of geometries with apical singularities differ much less in $\mathbb{R}^4$ at large aspect ratios $\nu \gg 1$ from those without apical singularities  than in $\mathbb{R}^3$.

\subsubsection{Oblate geometries}

While hyperellipsoids of revolution, hypercylinders, and hyperdoublecones exist in addition to prolate aspect ratios $\nu\ge 1$ also as oblate solids of revolution with $0\le \nu\le 1$, hyperlenses are restricted to oblate aspect ratios. In Fig. \ref{fig:b2red_oblate} (lhs), reduced second virial coefficients $B_2^\ast$ of these geometries are shown in dependence on their inverse aspect ratio $\nu^{-1}$. 

\begin{figure}[ht]
\centering
\includegraphics[width=0.475\textwidth,valign=t]{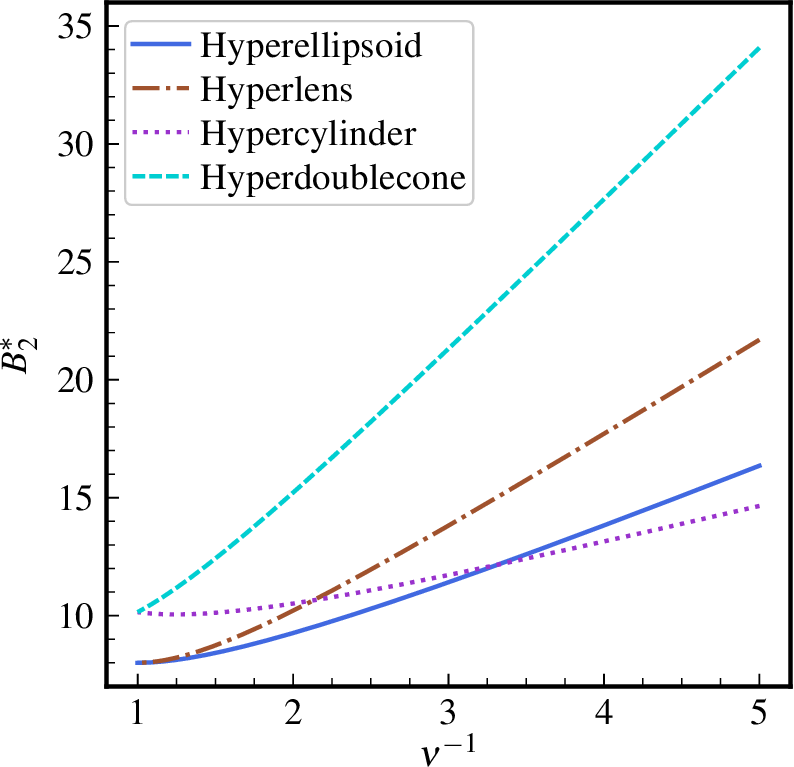} \hfill
\includegraphics[width=0.5\textwidth,valign=t]{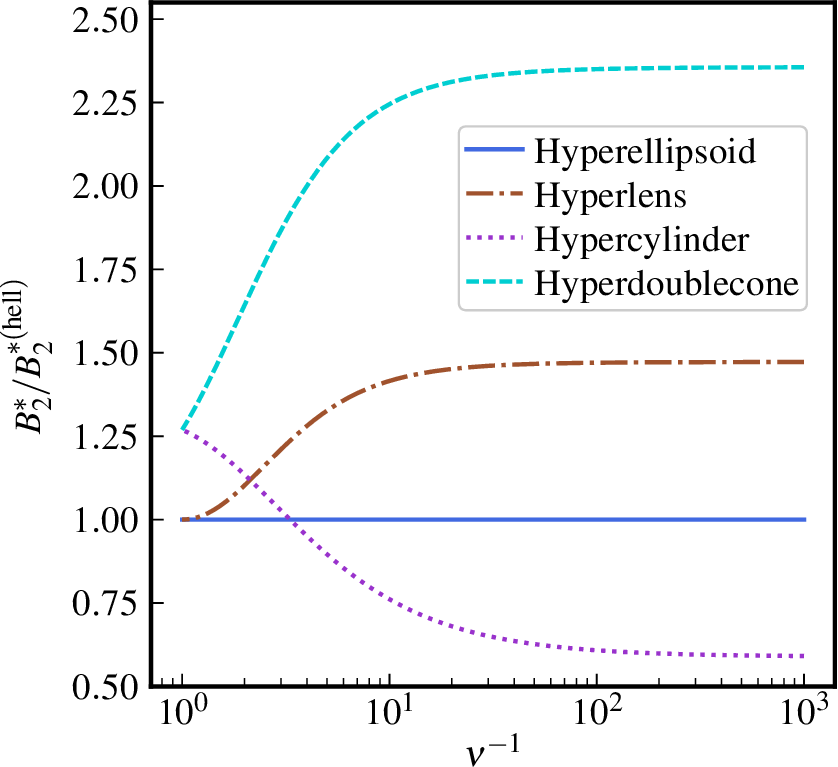}
\caption{Reduced second virial coefficients $B_2^\ast$ of oblate solids of revolution with inversion symmetry (lhs) and these coefficients normalized to the reduced second virial coefficients of hyperellipsoids $B_2^{\ast\,(\mathrm{hell})}$ (rhs) in dependence on the inverse aspect ratio $\nu^{-1}$.}\label{fig:b2red_oblate}
\end{figure}

In the limit $\nu^{-1} \to 1$, hyperlenses and hyperellipsoids approach hyperspheres with $B_2^\ast =8$. 
With decreasing aspect ratio $\nu$, the proportionality  $B_2^\ast(\nu \ll 1) \propto \nu^{-1}$ is observed with the limits 
\begin{subequations}\label{eq:b2_red_lim_obl}
\begin{align}
\lim\limits_{\nu \to 0} \left[\nu B_2^{\ast\,(\mathrm{hcyl})}\right] &= \dfrac{\pi}{4} + \dfrac{8}{3 \pi} \,,\\ 
\lim\limits_{\nu \to 0} \left[\nu B_2^{\ast\,(\mathrm{hell})}\right] &= \dfrac{4}{3} + \dfrac{128}{9 \pi^2} \,, \label{eq:b2_red_lim_obl_hell}\\
\lim\limits_{\nu \to 0} \left[\nu B_2^{\ast\,(\mathrm{hlen})}\right] &= \dfrac{5}{8} \pi + \dfrac{20}{3 \pi} \,,\\
\lim\limits_{\nu \to 0} \left[\nu B_2^{\ast\,(\mathrm{hdcn})}\right] &= \pi + \dfrac{32}{3 \pi} 
\end{align}
\end{subequations}
in increasing order. Note that for hyperellipsoids of revolution the relation
\begin{align}
\label{eq:hell_symmetry_0_infty}
\lim\limits_{\nu\to 0} \left[\nu B_2^\ast(\nu )\right] = \lim\limits_{\nu\to \infty} \dfrac{B_2^\ast(\nu )}{\nu}
\end{align}
is fulfilled [Eqs. \eqref{eq:b2_red_lim_pro_hell} and \eqref{eq:b2_red_lim_obl_hell}].

Despite hyperlenses and oblate hyperellipsoids are similar particle shapes, the reduced second virial coefficients of hyperlenses exceed those of the corresponding hyperellipsoids of revolution for all aspect ratios $\nu < 1$, indicating the influence of the detailed particle shape beyond size and aspect ratio.
In $\mathbb{R}^4$, for small aspect ratios $0<\nu\ll 1$, larger differences between the reduced second virial coefficients of these geometries are observed than between their analogues in $\mathbb{R}^3$ [Fig. 11 in \cite{Herold2017} and Fig. \ref{fig:b2red_oblate} (rhs)] which is opposite to the comparison of prolate geometries between $\mathbb{R}^3$ and $\mathbb{R}^4$.

\subsubsection{Comparison of shapes existing both as prolate and oblate solids of revolutions}\label{sec:res_pro_obl}

Hyperellipsoids, hypercylinders, and hyperdoublecones exist both as prolate and oblate solids of revolution with aspect ratios \mbox{$0 \leq \nu < \infty$}. Since in addition to the former geometries with inversion symmetry also hypercones (discussed later in Sec. \ref{sec:cones}) exist as prolate and oblate solids of revolution, they are for completeness included in this comparison.
In Fig. \ref{fig:b2_sym}, reduced second virial coefficients $B_2^\ast$ of these geometries are displayed as well as reduced second virial coefficients normalized to those of hyperellipsoids of revolution in dependence of their aspect ratio $\nu$.

\begin{figure}[ht]
\centering
\includegraphics[width=0.49\textwidth,valign=t]{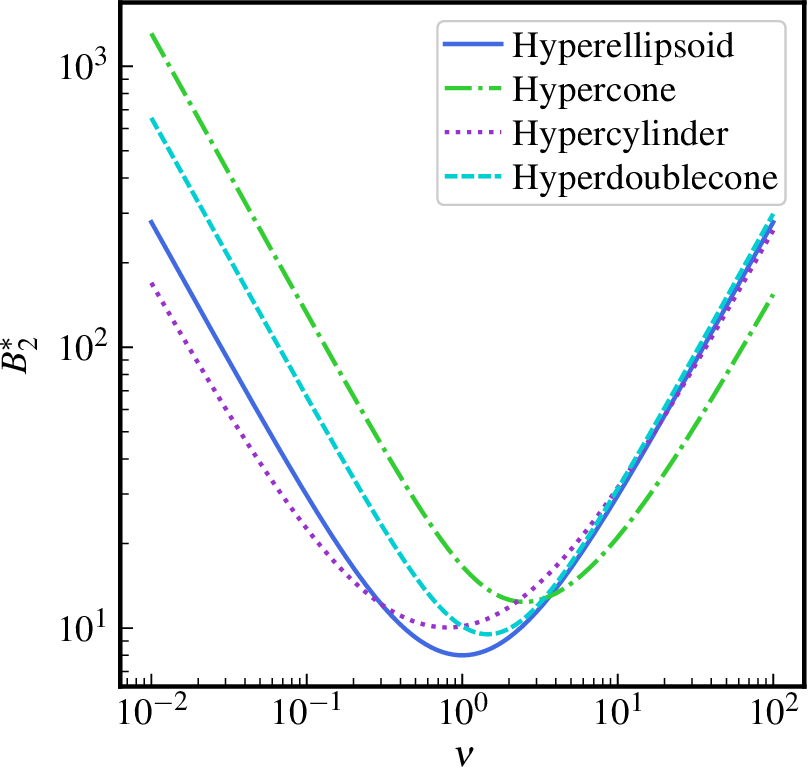} \hfill
\includegraphics[width=0.48\textwidth,valign=t]{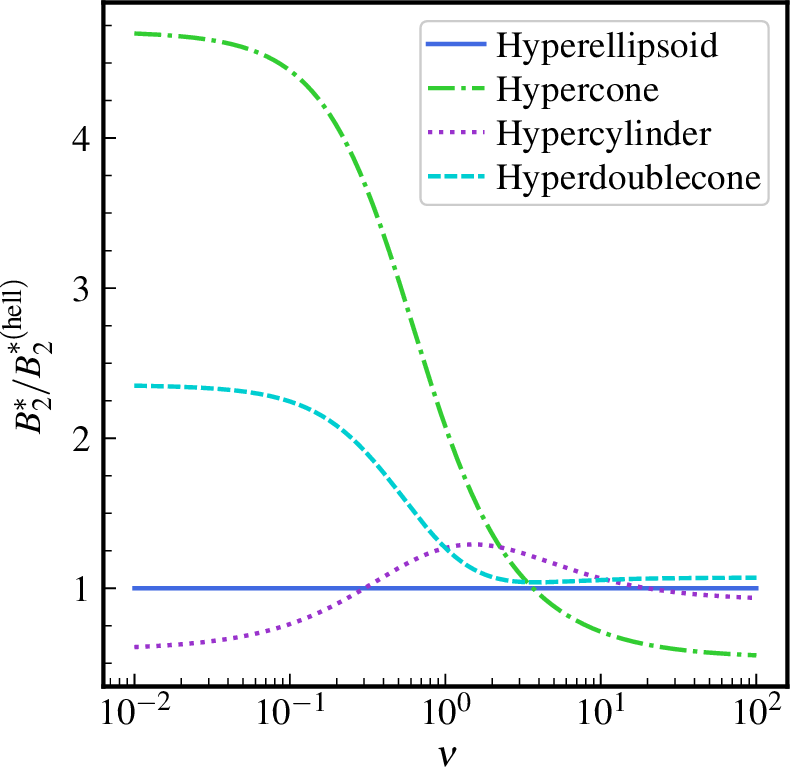}
\caption{Reduced second virial coefficients $B_2^\ast$ for shapes existing as prolate and oblate solids of revolution (lhs) and their reduced second virial coefficients normalized to those of hyperellipsoids $B_2^\ast / B_2^{\ast\,(\mathrm{hell})}$ (rhs) in dependence on the aspect ratio $\nu$.}\label{fig:b2_sym}
\end{figure} 

The reduced second virial coefficient $B_2^\ast$ for each geometry traverses a minimum $B_{2, \rmin}^\ast$ at $\nu_{\rm min}$ as summarized in Table \ref{tab:nu_min}. 
As expected, hyperellipsoids approach for $\nu \to 1$ a hypersphere with the smallest $B_2^\ast$ possible while for different geometries the minima are located at $\nu_{\min}\neq 1$.

\begin{table}[ht]
\caption{\label{tab:nu_min} Minima of reduced second virial coefficients $B_{2, \rmin}^\ast$ and corresponding aspect ratios $\nu_{\rm min}$.}
\begin{center}
\begin{tabular}{lcc}
\toprule 
Geometry & $\nu_{\rm min}$ & $B_{2, \rmin}^\ast(\nu_{\rm min})$\\ 
\midrule 
 Hyperellipsoid 		& $1$ & $8$	\\ 
 Hyperdoublecone 		& $1.456\,428\,\dots$ & $\phantom{1}9.506\,809\,\dots$	\\ 
 Hypercylinder 			& $0.801\,098\,\dots$ & $10.050\,852\,\dots$ \\ 
 Hypercone 					& $2.502\,640\,\dots$ & $12.394\,541\,\dots$ \\ 
\bottomrule
\end{tabular}
\end{center}
\end{table} 

In addition to the symmetry relation for reduced second virial coefficients of hyperellipsoids of revolution in the limits $\nu\to 0$ and $\nu\to \infty$ [Eq. \eqref{eq:hell_symmetry_0_infty}], a general
symmetry relation 
\begin{align}
B_2^{\ast\, (\mathrm{hell})} \left(\nu \right) = B_2^{\ast\, (\mathrm{hell})} (\nu^{-1})
\end{align}
can be identified in $\mathbb{R}^4$ which also exists for ellipsoids of revolution in $\mathbb{R}^3$ \cite{Herold2017}. This symmetry relation exists for $D$-dimensional ellipsoids of revolution in arbitrary dimensional Euclidean spaces $\mathbb{R}^D$ as proven in the following.
\begin{theorem}\label{lemma:hell}
For $D$-dimensional, uniaxial ellipsoids of revolution with aspect ratio $\nu$, the relation
\begin{align}
B_2^{\ast\,(\mathrm{ell}) }(D, \nu) &= B_2^{\ast\,(\mathrm{ell}) }(D, \nu^{-1})
\end{align}
is fulfilled for the reduced second virial coefficient $B_2^{\ast {\rm (ell)}}$ in Euclidean spaces $\mathbb{R}^D$ with arbitrary dimension $D$. 
\end{theorem}
\begin{proof}
Starting from the Brunn-Minkowski theorem
\begin{align*}
V_\rex(K) = \dfrac{1}{\kappa_D} \sum\limits_{i=0}^D \begin{pmatrix}
D \\ i\\
\end{pmatrix}  W_i(K) W_{D-i}(K)
\end{align*}
[Eq. \eqref{eq:Vex_K}] for the rotation-averaged excluded volume of two identical convex particles $K$, their reduced second virial coefficient $B_2^\ast (K)$ can be written as 
\begin{align}
B_2^\ast (K) = \dfrac{1}{2 W_0(K) \kappa_D} \sum\limits_{i=0}^D \begin{pmatrix}
D \\ i\\
\end{pmatrix}  W_i(K) W_{D-i}(K)
\end{align}
with $B_2^\ast = B_2 / V_\rP$ [Eq. \eqref{eq:B2red_an}], $V_\rex(K) = 2 B_2(K)$ [Eq. \eqref{eq:B2Vex}], and $V_\rP(K) = W_0(K)$ [Eq. \eqref{eq:W_0}]. 
For $D$-dimensional, uniaxial ellipsoids of revolution, the quermassintegral $W_i$ can be determined by
\begin{align*}
W_i = \kappa_D \nu^{i+1} r_\req^{D-i} \,_2\mathcal{F}_1\left(\dfrac{D+1}{2},\dfrac{i}{2},\dfrac{D}{2};1-\nu^2 \right)
\end{align*}
[Eq. \eqref{eq:hell_Wi}]. 
With 
\begin{align}
W_0 = \kappa_D \nu r_\req^{D}\,,
\end{align}
the reduced second virial coefficient $B_2^\ast$ for a $D$-dimensional, uniaxial ellipsoid of revolution reads as 
\begin{align}
B_2^{\ast\,(\mathrm{ell}) }(D, \nu) = \dfrac{1}{2\kappa_D^2 \nu r_\req^D} &\sum\limits_{i=0}^D \left[ \begin{pmatrix}
D \\ i\\
\end{pmatrix}  \kappa_D^2 \nu^{D+2} r_\req^D \,_2\mathcal{F}_1 \left(\dfrac{D+1}{2},\dfrac{i}{2},\dfrac{D}{2};1-\nu^2 \right) \right. \nonumber \\
&\qquad \left. \,_2\mathcal{F}_1 \left(\dfrac{D+1}{2},\dfrac{D-i}{2},\dfrac{D}{2};1-\nu^2 \right) \right]
\end{align}
which can be rewritten as
\begin{align}\label{b2red_hell_d_nu}
B_2^{\ast\,(\mathrm{ell}) }(D, \nu) = \dfrac{1}{2} \nu^{D+1} &\sum\limits_{i=0}^D \left[ \begin{pmatrix}
D \\ i\\
\end{pmatrix} \,_2\mathcal{F}_1 \left(\dfrac{D+1}{2},\dfrac{i}{2},\dfrac{D}{2};1-\nu^2 \right) \right. \nonumber \\
&\qquad \left. \,_2\mathcal{F}_1 \left(\dfrac{D+1}{2},\dfrac{D-i}{2},\dfrac{D}{2};1-\nu^2 \right) \right] \,.
\end{align}
Analogously, for aspect ratios $\nu^{-1}$,
\begin{align}\label{b2red_hell_d_nu_inv_1}
B_2^{\ast\,(\mathrm{ell}) }(D, \nu^{-1}) = \dfrac{1}{2}  \nu^{-(D+1)} &\sum\limits_{i=0}^D \left[ \begin{pmatrix}
D \\ i\\
\end{pmatrix} \,_2\mathcal{F}_1 \left(\dfrac{D+1}{2},\dfrac{i}{2},\dfrac{D}{2};1-\dfrac{1}{\nu^2} \right) \right. \nonumber \\
&\qquad \left. \,_2\mathcal{F}_1 \left(\dfrac{D+1}{2},\dfrac{D-i}{2},\dfrac{D}{2};1-\dfrac{1}{\nu^2} \right) \right]
\end{align}
is obtained. 
Using the Pfaff transformation \cite{NIST:DLMF} 
\begin{subequations}
\begin{align}
\,_2\mathcal{F}_1 \left(a,b,c;z \right) &= \left(1-z \right)^{-a} \,_2\mathcal{F}_1 \left(a,c-b,c;\dfrac{z}{z-1} \right) \\
&=\left(1-z \right)^{-b} \,_2\mathcal{F}_1 \left(c-a,b,c;\dfrac{z}{z-1} \right)
\end{align}
\end{subequations}
the relations 
\begin{align}\label{eq:hell_pfaff1}
\,_2\mathcal{F}_1 \left(\dfrac{D+1}{2},\dfrac{i}{2},\dfrac{D}{2};1-\nu^2 \right) = \nu^{-(D+1)} \,_2\mathcal{F}_1 \left(\dfrac{D+1}{2},\dfrac{D-i}{2},\dfrac{D}{2};1-\dfrac{1}{\nu^2} \right) 
\end{align}
and
\begin{align}\label{eq:hell_pfaff2}
\,_2\mathcal{F}_1 \left(\dfrac{D+1}{2},\dfrac{D-i}{2},\dfrac{D}{2};1-\nu^2 \right) = \nu^{-(D+1)} \,_2\mathcal{F}_1 \left(\dfrac{D+1}{2},\dfrac{i}{2},\dfrac{D}{2};1-\dfrac{1}{\nu^2} \right) 
\end{align}
result. 
Using Eqs. \eqref{eq:hell_pfaff1} and \eqref{eq:hell_pfaff2}, Eq. \eqref{b2red_hell_d_nu_inv_1} can be written as 
\begin{align}\label{b2red_hell_d_nu_inv}
B_2^{\ast\,(\mathrm{ell}) }(D, \nu^{-1}) = \dfrac{1}{2} \nu^{D+1} &\sum\limits_{i=0}^D \left[ \begin{pmatrix}
D \\ i\\
\end{pmatrix} \,_2\mathcal{F}_1 \left(\dfrac{D+1}{2},\dfrac{i}{2},\dfrac{D}{2};1-\nu^2 \right) \right. \nonumber \\
&\qquad \left. \,_2\mathcal{F}_1 \left(\dfrac{D+1}{2},\dfrac{D-i}{2},\dfrac{D}{2};1-\nu^2 \right)\right] \,,
\end{align}
equal to Eq. \eqref{b2red_hell_d_nu}. 
Hence, the theorem 
\begin{align*}
B_2^{\ast\,(\mathrm{ell}) }(D, \nu) = B_2^{\ast\,(\mathrm{ell}) }(D, \nu^{-1})
\end{align*}
results. This completes the proof. 
\end{proof}

\subsection{Solids without inversion symmetry}

Hyperspherical caps, hypercones, and truncated hypercones are convex, uniaxial solids of revolution without inversion symmetry for which
second virial coefficients $B_2$ and reduced second virial coefficients $B_2^\ast$  can be obtained analytically from the geometric measures 
$V_\rP$, $S_\rP$, $W_2$, and $\tilde{R}_\rP$ [Eq. \eqref{eq:B2red_an}].

\subsubsection{Hyperspherical caps}

The second virial coefficients of hyperspherical caps can either be expressed in dependence on the critical angle $\theta_\rc$ or the aspect ratio $\nu$. The latter quantities are related by 
\begin{align}
\nu = \left\{\begin{array}{ccc}
\left[1- \cos\left( \dfrac{\theta_\rc}{2} \right) \right]\left[2 \sin \left(\dfrac{\theta_\rc}{2} \right)\right]^{-1} & : & 0 \leq \theta_\rc \leq \pi \\
\\
\dfrac{1}{2} \left[1-\cos \left(\dfrac{\theta_\rc}{2} \right) \right] 	& : & \pi \leq \theta_\rc \leq 2 \pi
\end{array} \right.
\end{align} 
[Eqs. \eqref{eq:hcap_h}, \eqref{eq:hcap_req}, and  \eqref{eq:hcap_nu}]. 
\begin{figure}[ht]
\centering
\includegraphics[width=0.475\textwidth ,valign=t]{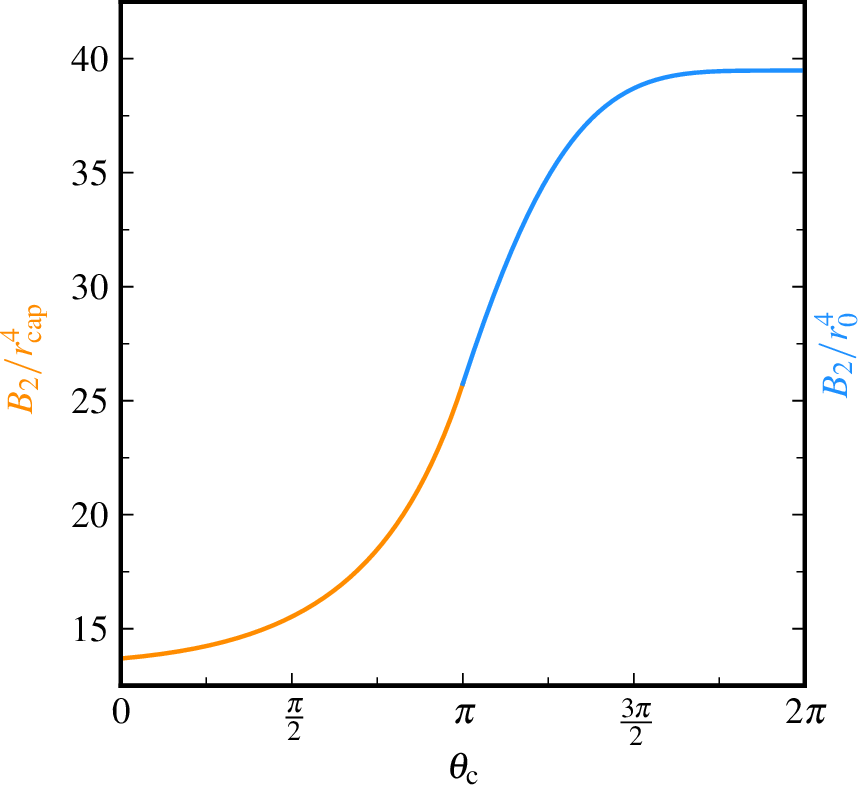} \hfill
\includegraphics[width=0.475\textwidth ,valign=t]{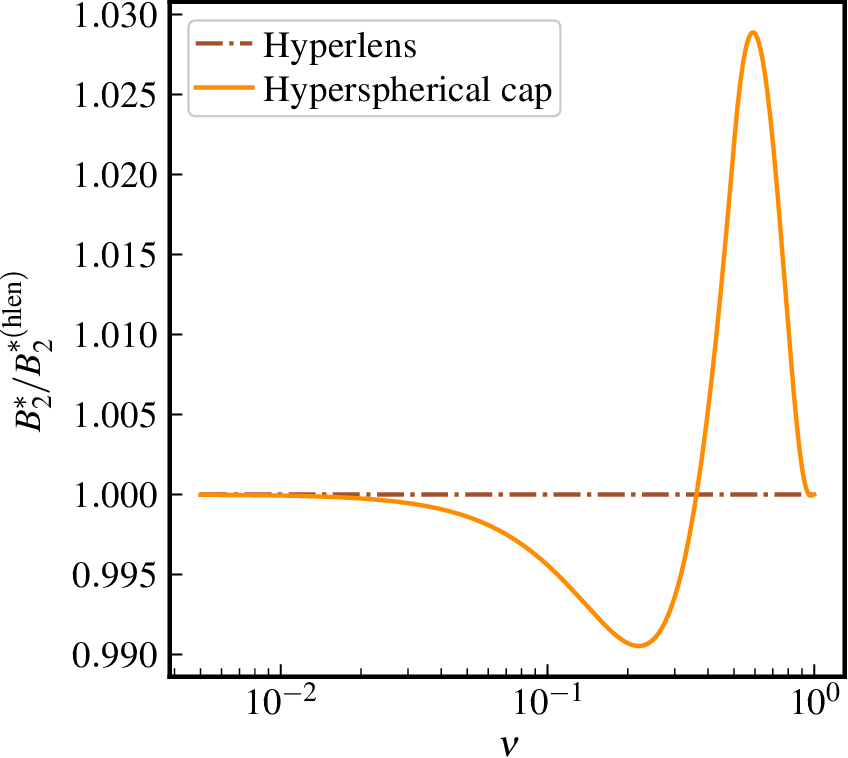}
\caption{Second virial coefficient $B_2$ of hyperspherical caps with radius of the generating hypersphere $r_0$ and radius of the spherical singularity $r_\rcap$ in dependence on the critical angle $\theta_\rc$ (lhs) and reduced second virial coefficients $B_2^\ast(\nu) / B_2^{\ast\,(\mathrm{hlen})}(\nu)$ normalized to those of hyperlenses in dependence on their aspect ratio $\nu$ (rhs). }\label{fig:b2_hcap}
\end{figure}

In Fig. \ref{fig:b2_hcap} (lhs), the second virial coefficients $B_2$ of hyperspherical caps are shown in dependence on the critical angle $\theta_\rc$. 
For critical angles $\theta_\rc \to 0$, a spherical hyperplate with radius $r_\rcap$ and second virial coefficient $B_2 = (\frac{64}{9} + \frac{2}{3} \pi^2) r_{\rm cap}^4$ results, while for $\theta_\rc \to 2 \pi$, a hypersphere with radius $r_0$ and second virial coefficient $B_2 = 4 \pi^2 r_0^4$ is obtained. 
For $\theta_\rc = \pi$, simply $r_\rcap = r_0 = h$ with second virial coefficient $B_2 = \left( \frac{43}{24} \pi^2 + 2 \pi + \frac{16}{9}\right) r_0^4$ results. 

The ratio of reduced second virial coefficients of hyperspherical caps and hyperlenses $B_2^{\ast}(\nu)/B_2^{\ast\,(\mathrm{hlen})}(\nu)$ is displayed in Fig. \ref{fig:b2_hcap} (rhs) in dependence on the aspect ratio $\nu$.
The reduced second virial coefficients are nearly identical with a maximum ratio at $\nu\approx 0.59$, a minimum ratio at $\nu\approx 0.22$ and a second, less pronounced minimum ratio at $\nu\approx 0.95$. 
In the limit $\nu \to 0$, both geometries approach a spherical hyperplate.

\subsubsection{Hypercones \& truncated hypercones}\label{sec:cones}

The second virial coefficients $B_2$ and reduced second virial coefficients $B_2^\ast$ of truncated hypercones are depicted in Fig. \ref{fig:b2_hcnt} in dependence on the aspect ratio $\nu$ for selected ratios $r_1/r_0$. 
In the limit $r_1/r_0 \to 0$, truncated hypercones approach hypercones. For the second virial coefficient $B_2^{(\mathrm{hcon})}$ of hypercones, in the limit 
\begin{align}
\lim\limits_{\nu \to \infty} \dfrac{B_2^{(\mathrm{hcon})}}{\nu^2 r_0^4} = \dfrac{14}{9}
\end{align} 
the proportionality $B_2 (\nu \gg1) \propto \nu^2$ is observed while for the reduced second virial coefficient in this limit
\begin{align}
\lim\limits_{\nu \to \infty} \dfrac{B_2^{\ast\,(\mathrm{hcon})}}{\nu} = \dfrac{14}{3 \pi}
\end{align}
the proportionality $B_2^\ast( \nu \gg 1) \propto \nu$ results. 

\begin{figure}[ht]
\centering
\includegraphics[width=0.48\textwidth ,valign=t]{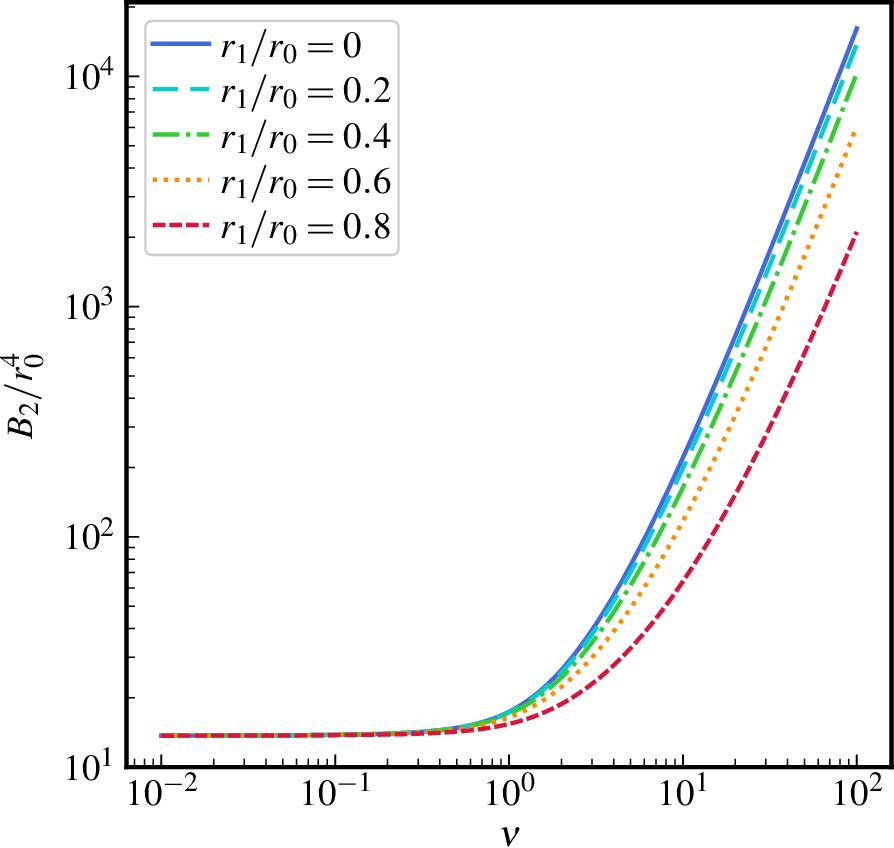}\hfill
\includegraphics[width=0.48\textwidth ,valign=t]{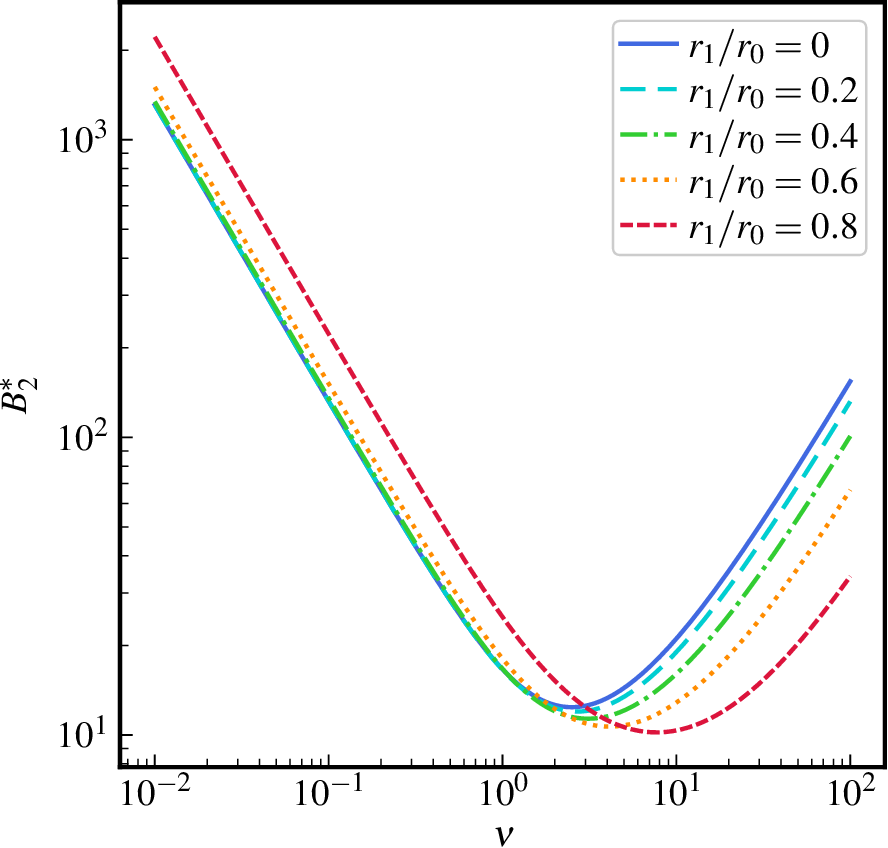}
\caption{Second virial coefficient $B_2$ (lhs) and reduced second virial coefficient $B_2^\ast$ (rhs) of truncated hypercones in dependence on the aspect ratio $\nu$ for selected ratios of radii $r_1/r_0$. In the limit $r_1/r_0 \to 0$, a hypercone results.}\label{fig:b2_hcnt}
\end{figure}

The second virial coefficients of both, hypercones and truncated hypercones, approach in the limit  
\begin{align}
\lim\limits_{\nu \to 0} \dfrac{B_2}{r_0^4} = \dfrac{64}{9} + \dfrac{2}{3} \pi^2  
\end{align}
the expected value of a spherical hyperplate. 
In the limit $r_1 / r_0 \to 0$, approaching hypercones with the reduced second virial coefficient $B_2^{\ast\,(\mathrm{hcon})}$, in the limit 
\begin{align}
\lim\limits_{\nu \to 0} \left[ \nu B_2^{\ast\,(\mathrm{hcon})} \right] = 2 \pi + \dfrac{64}{3 \pi}\,,
\end{align}
the proportionality $B_2^{\ast\,(\mathrm{hcon})} (\nu \ll 1) \propto \nu^{-1}$ is observed for vanishing aspect ratios $\nu$. 

With increasing aspect ratio $\nu$, the second virial coefficient $B_2$ of truncated hypercones decreases with increasing ratio of radii $r_1/r_0$. 
The minima of reduced second virial coefficients $B_{2}^\ast$ with respect to the aspect ratio $\nu$ decrease with rising ratio of radii $r_1/r_0$ and shift to larger aspect ratios.

\subsection{Excluded volume of $D$-dimensional spherocylinders}

Using the general expression for quermassintegrals of $D$-dimensional, uniaxial spherocylinders in Euclidean spaces $\mathbb{R}^D$ 
\begin{align*}
W_i = \left[\kappa_D + 2 \left(\nu-1\right) \dfrac{ D-i}{D} \kappa_{D-1} \right] r_\req^{D-i} 
\end{align*}
[Eq. \eqref{eq:Hsc_Wi_D}] and the Brunn-Minkowski theorem [Eq. \eqref{eq:Vex_K}], the excluded volume of identical particles can be written as
\begin{align}
V_\rex^{(\mathrm{scyl})} &= 2^D V_\rP + 2^{D-2} \left(S_\rP \tilde{R}_\rP - D V_\rP \right) 
\end{align}
for any dimension $D$, depending only on quermassintegrals of orders $i=0$, $i=1$ and $i=D-1$.

\begin{theorem}\label{lemma:hscy}
For identical, uniaxial $D$-dimensional spherocylinders with aspect ratio $\nu$ and radius $r_\req$, the rotation-averaged excluded volume $V_\rex^{(\mathrm{scyl})}$ can be written as
\begin{align}\label{eq:lemma_hscy}
V_\rex^{(\mathrm{scyl})} 
&= \left(2 r_\req \right)^D \left[ \kappa_D + 2 \left(\nu -1 \right) \kappa_{D-1} + \left(\nu -1 \right)^2 \,\dfrac{D-1}{D}\, \dfrac{\kappa_{D-1}^2}{\kappa_D} \right] \\
&= 2^D V_\rP + 2^{D-2} \left(S_\rP \tilde{R}_\rP - D V_\rP \right)  \nonumber
\end{align}
in Euclidean spaces $\mathbb{R}^D$ with arbitrary dimension $D$.
\end{theorem}

\begin{proof}
Starting from the Brunn-Minkowski theorem
\begin{align*}
V_\rex(K) = \dfrac{1}{\kappa_D} \sum\limits_{i=0}^D \begin{pmatrix}
D \\ i\\
\end{pmatrix}  W_i(K) W_{D-i}(K)
\end{align*}
[Eq. \eqref{eq:Vex_K}] for the rotation-averaged excluded volume between any two identical convex solids $K$, the expression
\begin{align}
V_\rex^{(\mathrm{scyl})} = \dfrac{1}{\kappa_D} &\sum\limits_{i=0}^D \begin{pmatrix}
D \\ i\\
\end{pmatrix} \left\{ \left[\kappa_D r_\req^{D-i} + 2 \left(\nu-1 \right) \dfrac{D-i}{D} \kappa_{D-1} r_\req^{D-i} \right] \right. \nonumber \\
&\qquad \left. \left[\kappa_D r_\req^i + 2 \left(\nu -1 \right) \dfrac{i}{D} \kappa_{D-1} r_\req^i \right] \right\}
\end{align}
results for $D$-dimensional spherocylinders using Eq. \eqref{eq:Hsc_Wi_D}. 
This can be rewritten as 
\begin{align}
V_\rex^{(\mathrm{scyl})} &= \sum\limits_{i=0}^D \begin{pmatrix}
D \\ i\\
\end{pmatrix} \kappa_D r_\req^D  +\sum\limits_{i=0}^D \begin{pmatrix}
D \\ i\\
\end{pmatrix} 2 \left( \nu -1 \right) \kappa_{D-1} r_\req^D  +  \nonumber \\
&\quad  \sum\limits_{i=0}^D \begin{pmatrix}
D \\ i\\
\end{pmatrix} 4 \left( \nu -1 \right)^2 \dfrac{(D-i)i}{D^2} \dfrac{\kappa_{D-1}^2}{\kappa_D} r_\req^D  \,.
\end{align}
Using the binomial theorem 
\begin{align}
\sum\limits_{i=0}^D \begin{pmatrix}
D \\ i\\
\end{pmatrix}  =2^D \,,
\end{align}
for the excluded volume between identical, $D$-dimensional spherocylinders 
\begin{align}
V_\rex^{(\mathrm{scyl})} =\left(2 r_\req \right)^D \left[\kappa_D + 2 \left(\nu -1\right) \kappa_{D-1} \right] + \dfrac{\left(\nu-1 \right)^2}{D} \dfrac{\kappa_{D-1}^2}{\kappa_D} r_\req^D\, \sum\limits_{i=0}^D \begin{pmatrix}
D \\ i\\
\end{pmatrix} 4 \dfrac{(D-i)\,i}{D}\,,
\end{align}
is obtained.  With 
\begin{align}
\sum\limits_{i=0}^D \begin{pmatrix}
D \\ i\\
\end{pmatrix} 4 \dfrac{(D-i)\,i}{D} = 2^D \left(D-1 \right)
\end{align}
immediately
\begin{align}\label{eq:Hscy_Vex_D_nu}
V_\rex^{(\mathrm{scyl})} = \left(2 r_\req \right)^D \left[ \kappa_D + 2 \left(\nu -1 \right) \kappa_{D-1} + \left(\nu -1 \right)^2 \,\dfrac{D-1}{D}\, \dfrac{\kappa_{D-1}^2}{\kappa_D} \right]
\end{align}
results. Using the representation of geometric measures $V_\rP$, $S_\rP$, and $\tilde R_\rP$ via quermassintegrals
\begin{subequations}
\begin{align}
V_\rP &= W_0 = \kappa_D r_\req^D + 2 \left(\nu-1 \right) \kappa_{D-1} r_\req^D \,, \label{eq:hsc_VP_W0}\\
S_\rP &= D \,W_1 = D \kappa_D r_\req^{D-1} + 2 \left(\nu -1 \right) \left(D-1 \right) \kappa_{D-1} r_\req^{D-1} \,,\\
\tilde{R}_\rP &= \dfrac{1}{\kappa_D} W_{D-1} = r_\req + 2 \left(\nu-1 \right) \dfrac{1}{D} \dfrac{\kappa_{D-1}}{\kappa_D} r_\req\,,
\end{align}
\end{subequations}
the relation
\begin{align}\label{eq:hsc_SR}
S_\rP \tilde{R}_\rP - D V_\rP = 4 \left(\nu-1 \right)^2 \dfrac{D-1}{D} \dfrac{\kappa_{D-1}^2}{\kappa_D} r_\req^D
\end{align}
is obtained for $D$-dimensional spherocylinders [Eqs. \eqref{eq:W} and \eqref{eq:Hsc_Wi_D}]. With
Eqs. \eqref{eq:hsc_VP_W0} and \eqref{eq:hsc_SR}, the excluded volume can be rewritten as
\begin{align*}
V_\rex^{(\mathrm{scyl})} &= 2^D V_\rP + 2^{D-2} \left(S_\rP \tilde{R}_\rP - D V_\rP \right)
\end{align*}
[Eq. \eqref{eq:Hscy_Vex_D_nu}]. This completes the proof.  
\end{proof}

As a unique feature of four-dimensional hyperspherocylinders, their excluded volume $V_\rex^{(\mathrm{hscy})}$ [Eq. \eqref{eq:lemma_hscy}] does not depend on the
particle volume $V_\rP$. Using Eq. \eqref{eq:B2Vex}, the second virial coefficient of $D$-dimensional, uniaxial spherocylinders reads as
\begin{align}
B_2^{(\mathrm{scyl})} = 2^{D-1} V_\rP + 2^{D-3} \left(S_\rP \tilde{R}_\rP - D V_\rP \right)
\end{align}
and the reduced second virial coefficient $B_2^\ast$ as 
\begin{align}
B_2^{\ast \,(\mathrm{scyl})} = 2^{D-1} + 2^{D-3} \left(\dfrac{S_\rP \tilde{R}_\rP}{V_\rP}-D \right)
\end{align}
[Eq. \eqref{eq:B2red_an}].

The influence of the aspect ratio $\nu$ is visible from the reformulation 
\begin{align}
B_2^{\ast \,(\mathrm{scyl})} = 2^{D-1} \left[1+ \dfrac{D-1}{D} \dfrac{\kappa_{D-1}}{\kappa_D} \dfrac{\left(\nu-1 \right)^2}{2 \left(\nu-1 \right)+ \kappa_D\kappa_{D-1}^{-1}} \right]
\end{align}
with the limits 
\begin{align*}
\lim\limits_{\nu \to 1} B_2^{\ast \,(\mathrm{scyl})}(D, \nu) = 2^{D-1}
\end{align*}
for a $D$-dimensional sphere [Eq. \eqref{eq:B2red_D_sphere}] and  
\begin{align}\label{eq:lim_needle_D}
\lim\limits_{\nu \to \infty} \dfrac{B_2^{\ast \,(\mathrm{scyl})} (D,\nu)}{\nu} = 2^D \dfrac{D-1}{8 \sqrt{\pi}} \dfrac{\Gamma\left(D/2 \right)}{\Gamma\left((D+1)/2 \right)}
\end{align}
for an infinitely long $D$-dimensional spherocylinder. In the latter limit, the proportionality $B_2^\ast\propto \nu$ is obtained in accordance to the results in dimensions $2\le D\le 4$
reported in the literature \cite{Kulossa2023b,Herold2017,Kulossa2022}. 
Since in the limit $\nu \to \infty$ the contributions of the capping $D$-dimensional hemispheres are negligible (see Sec. \ref{sec:res_pro}), the limits 
\begin{align}
\lim\limits_{\nu \to \infty} \dfrac{B_2^{\ast \,(\mathrm{scyl})} (D,\nu)}{\nu} = \lim\limits_{\nu \to \infty} \dfrac{B_2^{\ast \,(\mathrm{cyl})} (D,\nu)}{\nu}
\end{align}
are identical for $D$-dimensional spherocylinders (scyl) and $D$-dimensional cylinders (cyl).  
In Table \ref{tab:hyperneedle} the results for these limits are provided for Euclidean spaces $\mathbb{R}^D$ with $D\leq 16$.

\begin{table}[ht]
\caption{\label{tab:hyperneedle} Reduced second virial coefficients $B_2^\ast$ for $D$-dimensional spherocylinders with aspect ratio $\nu$ in the limit $\nu \to \infty$ [Eq. \eqref{eq:lim_needle_D}].}
\begin{center}
\begin{tabular}{rcrc}
\toprule 
$D$ & $\lim\limits_{\nu  \to \infty} \dfrac{B_2^{\ast \,(\mathrm{scyl})}(D,\nu)}{\nu}$ & $D$ & $\lim\limits_{\nu  \to \infty} \dfrac{B_2^{\ast \,(\mathrm{scyl})}(D,\nu)}{\nu}$ \\ 
\midrule \vspace*{3mm}
  1 & $0$ 		&  2 & $\dfrac{1}{\pi}$ 		 \\ \vspace*{3mm}  
  3 & $1$ 		&  4 & $\dfrac{2^{3}}{\pi}$ 		 \\ \vspace*{3mm} 
  5 & $6$ 		&  6 & $\dfrac{2^{7}}{3\pi}$ 	 \\ \vspace*{3mm} 
  7 & $30$ 		&  8 & $\dfrac{2^{11}}{10\pi}$ 	 \\ \vspace*{3mm}
  9 & $140$ 	& 10 & $\dfrac{2^{15}}{35\pi}$ \\ \vspace*{3mm}
 11 & $630$   & 12 & $\dfrac{2^{19}}{126\pi}$ \\ \vspace*{3mm}
 13 & $2\,772$ 	& 14 & $\dfrac{2^{23}}{462\pi}$ \\ \vspace*{3mm}
 15 & $12\,012$	& 16 & $\dfrac{2^{27}}{1\,716\pi}$ \\ 
\bottomrule
\end{tabular}
\end{center}
\end{table}

Instead of an infinitely long $D$-dimensional spherocylinder with radius $r_\req$, a rod with length $l$ and radius $r_\req=0$ is an alternative description of a hard needle in $\mathbb{R}^D$ \cite{Kulossa2023b}.  
The quermassintegrals $W_i$ read as 
\begin{align}
W_i = \left\{\begin{array}{lcl}
0 & : & i \leq D-2 \vspace*{2mm}\\
\dfrac{\kappa_{D-1}}{D} \,l & : & i = D-1 \vspace*{2mm}\\
\kappa_D & : & i=D 
\end{array} \right.
\end{align}
for such line segments \cite{Santalo2004,Torquato2022}. 
With Eq. \eqref{eq:Vex_K}, in $\mathbb{R}^1$ the excluded volume  $V_\rex = 2 l$ and in $\mathbb{R}^2$ the rotation-averaged excluded volume $V_\rex=2l^2/\pi$ of hard rods are obtained. 
In Euclidean spaces $\mathbb{R}^D$ with $D\geq 3$, the rotation-averaged excluded volume $V_\rex=0$ of infinitely thin rods with length $l$ vanishes \cite{Torquato2022}.

\section{Summary and Outlook} 

In this work, analytical expressions for the geometric measures volume $V_\rP$, surface area $S_\rP$, second quermassintegrals $W_2$ and mean radius of curvature $\tilde{R}_\rP$ are derived for selected convex solids of revolution in $\mathbb{R}^4$. 
The results are summarized in Tables \ref{tab:4d_1}, \ref{tab:4d_2}, and \ref{tab:4d_3}. 
Using Eqs. \eqref{eq:W}, analytical expressions for so far unknown quermassintegrals $W_i$ are obtained. 
Employing these quantities, with Eq. \eqref{eq:intrinsic_volumes} the intrinsic volumes $\upsilon_i$ of these convex solids 
of revolution are analytically accessible, too.

In addition to hyperspheres, hyperellipsoids of revolution, hyperspherocylinders, hypercylinders, and spherical hyperplates with already known quermassintegrals (Table \ref{tab:4d_1}), 
analytical expression for so far unknown quermassintegrals of hyperspindles, hyperlenses, hyperdoublecones, hypercones, truncated hypercones, and hyperspherical caps are
provided (Tables \ref{tab:4d_2} and \ref{tab:4d_3}). The latter geometries possess removable singularities in their surface curvature. While apical, zero-dimensional singularities do not contribute to
quermassintegrals, higher-dimensional singularities contribute as exemplarily shown in detail for hyperlenses. This contribution, in general depends on the critical angle enclosed within the $\epsilon$-vicinity
of the respective singularity.  For uniaxial solids of revolution with infinitely large aspect ratios in $\mathbb{R}^4$, a general proportionality $B_2^\ast \propto \nu$ of reduced second virial coefficients to the aspect ratio arises whereas for infinitely thin oblate geometries, the general proportionality $B_2^\ast\propto\nu^{-1}$ results. 

With known principal radii of curvature, the quermassintegrals are accessible for arbitrary convex solids in $\mathbb{R}^4$ [Eqs. \eqref{eq:Wi_K}]. The generalization to dimensions $D>4$ for geometries with continuous surfaces curvature [Eq. \eqref{eq:W_i_continuous}] is straight forward. Possible contributions of singularities can
be determined analogously as demonstrated in $\mathbb{R}^3$ \cite{Herold2017} and $\mathbb{R}^4$ (this work).

For selected geometries, general expressions for their quermassintegrals are compiled in Table \ref{tab:Wi}. Using such general expressions, the parity $B_2^\ast(\nu)=B_2^\ast(\nu^{-1})$ for uniaxial $D$-dimensional ellipsoids is proven. Additionally, 
for $D$-dimensional, uniaxial spherocylinders the dependence of the reduced second virial coefficient $B_2^\ast$ on at most three quermassintegrals $W_0$, $W_1$, and $W_{D-1}$ is proven, too. 

A remaining task is the derivation of a general expression for quermassintegrals of $D$-dimensional doublecones in $\mathbb{R}^D$ needed to prove or disprove the conjecture that in 
even dimensions $D$ reduced second virial coefficients of cylinders and doublecones are identical at aspect ratio $\nu=1$.

\newpage

\begin{table}[h!] 
\caption{\label{tab:4d_1} Geometric measures of four-dimensional hyperspheres, hyperellipsoids of revolution, hyperspherocylinders, hypercylinders, and spherical hyperplates.}
\begin{center}
\begin{tabular}{lll}
\toprule 
Hypersphere & $V_\rP=\dfrac{1}{2} \pi^2 r_0^4$  & $S_\rP=2 \pi^2 r_0^3$\\
\\
						& $W_2=\dfrac{1}{2}\pi^2 r_0^2$ & $\tilde{R}_\rP = r_0$  \\
\midrule
Hyperellipsoid & $V_\rP=\dfrac{1}{2} \pi^2 \nu r_\req^4$  & $W_2=\dfrac{1}{3} \pi^2 r_\req^2 \left( \dfrac{1-\nu^3}{1-\nu^2}\right)$ \\
\\
& $S_\rP= 2 \pi^2 \nu^2 r_\req^3 \, _2\mathcal{F}_1\left(\dfrac{5}{2},\dfrac{1}{2},2;1-\nu^2\right)$ \\
\\
& $\tilde{R}_\rP = \nu^4 r_\req \, _2\mathcal{F}_1\left(\dfrac{5}{2},\dfrac{3}{2},2;1-\nu^2\right)$ \\
\midrule
Hypersphero- & $V_\rP=\left[\dfrac{1}{2}\pi^2+\dfrac{8}{3} \left(\nu-1\right) \pi \right] r_\req^4$  & $S_\rP= \left[2\pi^2 + 8 \left(\nu-1\right)\pi \right] r_\req^3$\\
cylinder \\
 &  $W_2= \left[\dfrac{1}{2} \pi^2 + \dfrac{4}{3} \left(\nu-1 \right) \pi \right] r_\req^2$ & $\tilde{R}_\rP = \left[1+\dfrac{4}{3 \pi} \left(\nu-1\right) \right] r_\req$  \\
\midrule
Hypercylinder & $V_\rP =\dfrac{8}{3} \pi \nu r_\req^4$  & $S_\rP = \dfrac{8}{3} \pi \left(3 \nu+ 1 \right) r_\req^3$\\
 \\
& $W_2=\dfrac{1}{3} \pi \left(4\nu+\pi\right) r_\req^2$ & $\tilde{R}_\rP = \dfrac{4}{3 \pi} \left(\nu+2\right) r_\req$  \\
\midrule
Spherical  & $V_\rP =0$ & $S_\rP =\dfrac{8}{3} \pi r_0^3$ \\
hyperplate \\
			 & $W_2=\dfrac{1}{3} \pi^2 r_0^2$ & $\tilde{R}_\rP =\dfrac{8}{3 \pi} r_0$ \\
\bottomrule
\end{tabular} 
\end{center}
\end{table}

\newpage 

\begin{table}[h!]
\caption{\label{tab:4d_2} Geometric measures of four-dimensional hyperspindles, hyperlenses, and hyperspherical caps.}
\begin{center}
\begin{tabular}{lll}
\toprule  
Hyperspindle & \multicolumn{2}{l}{$V_\rP =\dfrac{5\pi}{16}r_\req^4 \left[ \left(\nu^4-\dfrac{6}{5}\nu^2+1 \right) \left(\nu^2+1 \right)^2 \arcsin\left(\dfrac{2\nu}{\nu^2+1} \right) - 2 \nu^7 \right.$ \vspace*{2mm}} \\
& \multicolumn{2}{l}{ \qquad \quad $\left. -\dfrac{14}{15}\nu^5  + \dfrac{14}{15}\nu^3 +2\nu \right]$} \\
\\
& \multicolumn{2}{l}{$S_\rP =\dfrac{1}{2} \pi r_\req^3 \left[\left(3 \nu^6 + \nu^4 + \nu^2 +3 \right) \arcsin \left(\dfrac{2 \nu}{\nu^2 +1} \right) -6\nu^5 +6\nu \right]$} \\
\\
& \multicolumn{2}{l}{$W_2 = \dfrac{5 \pi}{12} r_\req^2 \left[\left(\nu^4 + \dfrac{2}{5} \nu^2 +1 \right) \arcsin \left(\dfrac{2\nu}{\nu^2+1} \right) - 2\nu^3 + 2 \nu \right]$} \\ 
\\
& \multicolumn{2}{l}{$\tilde{R}_\rP=\dfrac{r_\req}{\pi}\left[\left(\nu^2 +1 \right)\arcsin\left(\dfrac{2\nu}{\nu^2 +1} \right)-\dfrac{2}{3}\nu\,\dfrac{\nu^2-1}{\nu^2+1} \right]$} \\
\midrule 
Hyperlens & \multicolumn{2}{l}{$V_\rP=\dfrac{\pi r_\req^4}{32 \nu^4}\left\{\left(1+\nu^2\right)^4 \left[\pi-2\arcsin	\left(\dfrac{1-\nu^2}{1+\nu^2} \right) \right] - 4\nu \left[1-\nu^6 \phantom{\dfrac{1}{1}}\right. \right. $ \vspace*{2mm}} \\
& \multicolumn{2}{l}{\qquad \quad $\left. \left. + \dfrac{11}{3} \nu^2 \left(1-\nu^2 \right) \right] \right\} $ }\\
\\
& \multicolumn{2}{l}{$S_\rP=\dfrac{\pi\left(1+\nu^2 \right)r_\req^3}{4\nu^3}\left\{\left(1+\nu^2 \right)^2 \left[\pi - 2 \arcsin \left(\dfrac{1-\nu^2}{1+\nu^2} \right) \right] \right.$ \vspace*{2mm}} \\
& \multicolumn{2}{l}{\qquad \quad $ \left. \phantom{\dfrac{1}{1}} -4\nu\left(1-\nu^2\right) \right\}$}\\
\\
& \multicolumn{2}{l}{$W_2=\dfrac{\pi r_\req^2}{8 \nu^2}\left[\left(1+\nu^2\right)^2 \pi -2 \left(\nu^4-\dfrac{2}{3}\nu^2 +1 \right) \arcsin\left(\dfrac{1-\nu^2}{1+\nu^2} \right)\right. $ \vspace*{2mm}} \\ 
& \multicolumn{2}{l}{\qquad \quad $ \left. \phantom{\dfrac{1}{1}} - 4 \nu \left(1-\nu^2 \right) \right]$}\\
\\
& \multicolumn{2}{l}{$\tilde{R}_\rP=\dfrac{r_\req}{2 \pi \nu}\left\{\left(1+\nu^2 \right) \left[\pi -2 \arcsin\left(\dfrac{1-\nu^2}{1+\nu^2} \right) \right] + \dfrac{4 \nu}{3} \dfrac{1-\nu^2}{1+\nu^2} \right\}$} \\
\midrule
Hyperspherical & \multicolumn{2}{l}{$V_\rP = \dfrac{\pi^2}{4} r_0^4 \left[1- \dfrac{2}{\pi} \arctan\left(\dfrac{r_0 -h}{r_\rcap} \right) - \dfrac{2 r_\rcap}{3 \pi r_0^4} \left(r_0-h \right) \left(3r_0^2 +2 r_\rcap^2 \right) \right]$}  \\
 cap  \\
 & \multicolumn{2}{l}{$S_\rP =\pi^2 r_0^3 \left[1-\dfrac{2}{\pi} \arctan \left(\dfrac{r_0-h}{r_\rcap} \right) - \dfrac{2 r_\rcap}{\pi r_0^2} \left(r_0 -h \right) \right] + \dfrac{4}{3} \pi r_\rcap^3$}  \\ 
 \\
 & \multicolumn{2}{l}{$W_2 = \dfrac{\pi^2}{4} r_0^2 \left[1-\dfrac{2}{\pi} \arctan \left(\dfrac{r_0-h}{r_\rcap} \right) - \dfrac{2 r_\rcap}{\pi r_0^2} \left(r_0 -h \right) \right] $ \vspace*{2mm} }\\
 & \multicolumn{2}{l}{\qquad \quad $+ \dfrac{\pi}{3} r_\rcap^2  \left[\dfrac{\pi}{2}+ \arcsin\left( 1- \dfrac{h}{r_0} \right) \right]$}\\
\\
& \multicolumn{2}{l}{$\tilde{R}_\rP = \dfrac{r_0}{2} \left[1-\dfrac{2}{\pi} \arctan \left(\dfrac{r_0-h}{r_\rcap} \right) - \dfrac{2 r_\rcap}{\pi r_0^2} \left(r_0 -h \right) \right] $ \vspace*{2mm}}  \\
& \multicolumn{2}{l}{\qquad \quad $+ \dfrac{4}{3 \pi} r_\rcap \left(2-\dfrac{h}{r_0} \right) $ }\\
\bottomrule
\end{tabular}
\end{center}
\end{table}

\newpage

\begin{table}[h!]
\caption{\label{tab:4d_3} Geometric measures of four-dimensional hyperdoublecones, hypercones, and truncated hypercones.}
\begin{center}
\begin{tabular}{lll}
\toprule  
Hyperdoublecone & $V_\rP=\dfrac{2}{3}\pi \nu r_\req^4$ & $S_\rP =\dfrac{8}{3}\pi r_\req^3 \sqrt{1+\nu^2} $ \\
\\
& \multicolumn{2}{l}{$W_2 = \dfrac{2}{3} \pi r_\req^2 \left[\nu+ \arcsin\left(\dfrac{1}{\sqrt{1+\nu^2}} \right)\right]$} \\ 
\\
& $\tilde{R}_\rP =\dfrac{4}{3 \pi}  r_\req \dfrac{\nu^2+2}{\sqrt{1+\nu^2}}$ \\
\midrule 
Hypercone & $V_\rP =\dfrac{1}{3} \pi \nu r_0^4$ & $S_\rP =\dfrac{4}{3} \pi r_0^3 \left(\sqrt{1+\nu^2}  +1 \right)$ \\
\\
& \multicolumn{2}{l}{$W_2 = \dfrac{1}{6} \pi r_0^2 \left[2\nu +\pi + 2\arcsin\left(\dfrac{1}{\sqrt{1+\nu^2}}\right) \right]$} \\
\\
& $\tilde{R}_\rP = \dfrac{2}{3 \pi}  r_0 \left(2+\dfrac{\nu^2+2}{\sqrt{1+\nu^2}} \right)$  \\
\\
\midrule
Truncated  & \multicolumn{2}{l}{$V_\rP=\dfrac{h}{3} \pi \left(r_0^3 + r_0^2r_1 + r_0r_1^2 +r_1^3 \right)$} \\
hypercone \\
& \multicolumn{2}{l}{$S_\rP =\dfrac{4}{3} \pi \left[r_0^3 + r_1^3 + \left(r_0^2 +r_0r_1 +r_1^2 \right) \sqrt{h^2+\left(r_0-r_1\right)^2}\,  \right]$}  \\
\\ 
& \multicolumn{2}{l}{$W_2 = \dfrac{1}{3} \pi \left(r_0+r_1 \right) \left\{h + \left(r_0 -r_1 \right) \arcsin\left[\dfrac{r_0-r_1}{\sqrt{h^2+(r_0-r_1)^2}} \right] \right\} $ \vspace*{2mm}} \\
& \multicolumn{2}{l}{\qquad \quad $ +\dfrac{1}{6} \pi^2 \left(r_0^2+r_1^2 \right)$}\\
\\
& \multicolumn{2}{l}{$\tilde{R}_\rP=\dfrac{2}{3 \pi} \left[2 \left(r_0 + r_1 \right) + \dfrac{h^2+2(r_0 -r_1)^2}{\sqrt{h^2 + (r_0-r_1)^2}} \right]$} \\
\bottomrule
\end{tabular}
\end{center}
\end{table}

\newpage

\begin{table}[h!]
\caption{\label{tab:Wi} Quermassintegrals $W_i$ of $D$-dimensional, uniaxial solids of revolution.}
\begin{center}
\begin{tabular}{ll}
\toprule 
Sphere & $W_i = \kappa_D r_0^{D-i}$ \\ 
\midrule 
Ellipsoid 	 & $W_i =  \kappa_D \nu^{i+1} r_\req^{D-i} \,_2\mathcal{F}_1\left(\dfrac{D+1}{2},\dfrac{i}{2}, \dfrac{D}{2}; 1-\nu^2 \right)$	\\ 
\midrule 
Spherocylinder & $W_i = \left[ \kappa_D + 2 \left(\nu-1 \right)\, \dfrac{D-i}{D}\, \kappa_{D-1} \right] r_\req^{D-i}$ \\ 
\midrule
Cylinder 		 & $W_i = \left\{\begin{array}{lcc}
2  \kappa_{D-1} \nu  r_\req^D & : & i=0 \vspace*{2mm} \\
\dfrac{2}{D} \kappa_{D-1} \left[\left(D-i \right) \nu + \dfrac{\beta_{i}}{2 \kappa_{i-1}}  \right] r_\req^{D-i} & : & i \geq 1\\
\end{array} \right.$ \\ 
\midrule 
Spherical cap\\
$0 \leq h \leq r_0:$ & $W_0 = \dfrac{\kappa_D}{2} r_0^D \mathcal{I}_{r_\rcap^2/r_0^2}\left(\dfrac{D+1}{2},\dfrac{1}{2} \right)$\\
\\
&$W_1 = \dfrac{\kappa_D}{2}  r_0^{D-1} \mathcal{I}_{r_\rcap^2/r_0^2}\left(\dfrac{D-1}{2},\dfrac{1}{2}\right) + \dfrac{\kappa_{D-1}}{D} r_\rcap^{D-1} $ \vspace*{2mm}\\
\\
$r_0 \leq h \leq 2 r_0:$ & $W_0 = \dfrac{\kappa_D}{2} r_0^D \left[ 2-  \mathcal{I}_{r_\rcap^2/r_0^2}\left(\dfrac{D+1}{2},\dfrac{1}{2} \right)\right]$ \vspace*{2mm}\\
& $ W_1 = \dfrac{\kappa_D}{2}  r_0^{D-1} \left[2- \mathcal{I}_{r_\rcap^2/r_0^2}\left(\dfrac{D-1}{2},\dfrac{1}{2}\right)\right] + \dfrac{\kappa_{D-1}}{D} r_\rcap^{D-1} $ \\ 
\\
\midrule
Doublecone 	& $W_0 = \dfrac{2}{D} \kappa_{D-1} \nu r_\req^D$ \qquad \qquad  $W_1 = \dfrac{2}{D} \kappa_{D-1} \sqrt{1+\nu^2}\, r_\req^{D-1}$  \\ 
\midrule
Cone 	& $W_0 = \dfrac{1}{D} \kappa_{D-1} \nu r_0^D$  \qquad \qquad $W_1 = \dfrac{1}{D} \kappa_{D-1} \left(\sqrt{1+\nu^2} +1 \right) r_0^{D-1}$ \\
\midrule
Truncated & $W_0 = \dfrac{\kappa_{D-1}}{D} \, h\, \dfrac{r_1^D - r_0^D}{r_1 -r_0} $ \\
cone & \\
					& $W_1 = \dfrac{\kappa_{D-1}}{D}  \left[r_0^{D-1} + r_1^{D-1} + \dfrac{r_1^{D-1} - r_0^{D-1}}{r_1 - r_0} \sqrt{h^2 + \left(r_0-r_1 \right)^2} \right]$ \\
\midrule
Spherical 	& \multirow{2}{*}{$W_i = \left\{\begin{array}{lcc}
0 & : & i=0 \vspace*{2mm} \\
\dfrac{\kappa_{D-1}}{\kappa_{i-1}} \dfrac{i}{D} \kappa_i r_0^{D-i} & : & i \geq 1\\
\end{array} \right. $}\\ 
plate\\
\\
\\
\bottomrule
\end{tabular}
\end{center}
\end{table}

\newpage

\newpage

\backmatter


%
%




\end{document}